\documentclass[11pt]{article}

\usepackage[english]{babel}
\usepackage[utf8]{inputenc}

\usepackage{enumerate}
\usepackage{fontenc}
\usepackage{xfrac}
\usepackage{placeins}
\usepackage{amsbsy,amscd, amsthm, amsmath,amsfonts,amssymb}
\usepackage{graphics,graphicx,epsfig}
\usepackage{hyperref}
\usepackage{bbm, natbib}
\usepackage{colortbl}
\usepackage[normalem]{ulem} 
\usepackage{cancel, comment} 
\usepackage{float}
\usepackage{afterpage}
\usepackage{marginnote}
\usepackage{xcolor}
\definecolor{darkgreen}      {cmyk}{0.90,0.00,0.90,0.10}

\newtheorem{prop}{Proposition}[subsection]
\newtheorem{lem}[prop]{Lemma}

\newcommand{\idc}[1]  { \mathbf{1}_{\left \lbrace #1 \right \rbrace}  }
\newcommand{\II}[1]	{[\![#1]\!]}
\newcommand{\bv}  {\ensuremath{\, \big| \,} }
\newcommand{\rA}{\mathbin{\scalebox{0.7}{$\rightarrow$} }}
\newcommand{\lA}{\mathbin{\scalebox{0.7}{$\leftarrow$}  }}
\newcommand{\hcm}[1]  {\hspace{#1 cm}}
\newcommand{\vcm}[1]  {\vspace{#1 cm}}
\newcommand{\bE}  {\ensuremath{\mathbb{E}}}

\newcommand{\bR}  {\ensuremath{\mathbb{R}}}
\newcommand{\cT}  {\ensuremath{\mathcal{T}}}
\newcommand{\cZ}  {\ensuremath{\mathcal{Z}}}

\newcommand{\Wmat}{W}

\let\OLDthebibliography\thebibliography 
\renewcommand\thebibliography[1]{ 
	\OLDthebibliography{#1} 
	\setlength{\parskip}{0pt} 
	\setlength{\itemsep}{5pt plus 0.3ex} 
}

\begin{document}
	
\title{
	Inter-city infections and the role of size heterogeneity\\ in containment strategies
}
\author{Viktor Bezborodov
	\thanks{Wrocław University of Science and Technology, Faculty of Information and Communication Technology, 
		ul. Janiszewskiego 11/17, 50-372 Wrocław,
		Poland; E-mail:
		\texttt{viktor.bezborodov@pwr.edu.pl}} 
	\and Tyll Krueger
	\thanks{Wrocław University of Science and Technology, Faculty of Information and Communication Technology, 
		ul. Janiszewskiego 11/17, 50-372 Wrocław,
		Poland; E-mail: 
		\texttt{tyll.krueger@pwr.edu.pl}} 
	\and Cornelia Pokalyuk
	\thanks{University of Lübeck,
		Institute for Mathematics,
		Ratzeburger Allee 160, D-23562, Lübeck, Germany;
		E-mail:\texttt{cornelia.pokalyuk@uni-luebeck.de}} 
	\and Piotr Szyma\'nski
	\thanks{Wrocław University of Science and Technology, Department of Computational Intelligence,
		Wrocław, Poland; E-mail: 
		\texttt{piotr.szymanski@pwr.edu.pl}}
	\and
	Aurélien Velleret
	\thanks{Université Paris-Saclay, Université d’Evry Val d’Essonne, CNRS, LaMME, \ UMR 8071
		\ 91037 Evry, France; E-mail:
		\texttt{aurelien.velleret@nsup.org}, corresponding author}
}

\date{}

\maketitle
\section*{Abstract}

This study examines the effectiveness of regional lockdown strategies in mitigating pathogen spread across regional units, termed cities hereinafter. We develop simplified models to analyze infection spread across cities within a country during an epidemic wave. Isolation of a city is initiated when infection numbers within the city surpass defined thresholds. We compare two strategies: 
strategy~\((P)\) consists in prescribing thresholds proportionally to city sizes,
while the same threshold is used for all cities under strategy~\((U)\).  Given the heavy-tailed distribution of city sizes, strategy \((P)\)  may result in more secondary infections from larger cities than strategy~\((U)\).

Random graph models are constructed to represent  infection spread as a percolation process. In particular, we consider a model in which mobility between cities only depends on city sizes. We assess the relative efficiency of the two strategies by comparing the ratios of the  number of individuals under isolation to the total number of infections by the end of the epidemic wave under strategy~\((P)\) and~\((U)\). Additionally, we derive analytical formulas for disease prevalence and basic reproduction numbers.

Our models are calibrated using mobility data from France, Poland and Japan, validated through simulation. The findings indicate that mobility between cities in France and Poland is mainly determined by city sizes. However, a poor fit was  observed with Japanese data, highlighting the importance to include other factors like e.g. geography for some countries in modeling. Our analysis suggest similar effectiveness for both strategies in France and Japan, while strategy~\((U)\)  demonstrates distinct merits in Poland.

\textbf{Keywords: } Epidemic spreading; Stochastic processes; SIR epidemic model; Random graphs; Containment strategies; Transportation networks; Lockdowns.

\textbf{AMS Subject Classification:} 92D30, 05C80, 05C82,  60J80, 60G55

\section{Introduction}

The spread of  SARS-CoV-2 in early 2020 
and the subsequent invasions of more pathogenic mutants of the original strain
highlighted the crucial need for effective pandemic control methods.
 Since neither drugs nor vaccines were available at the beginning of the pandemic, non-pharmaceutical interventions were arranged to restrict contacts between infected and uninfected individuals and slow down the pace of the  epidemic.
 
In China and New Zealand,  the first epidemic wave could essentially be stopped with regional containment strategies, \cite{MB20}. These countries followed so called zero-COVID strategies where a regional lockdown was already imposed when only a small number of individuals got infected.  For example in China, the application of this strategy resulted in approximately 60 million people (the residents of Hubei province) being under a very strict lockdown during the first wave in 2020, see \cite{WS20}.  Some precaution measures were in addition applied outside Hubei province, like mandatory mask wearing and certain social distancing controls, sometimes proactively even in the absence of detected cases.

Many European 
countries coordinated regional non-pharmaceutical interventions based on the relative incidence rates, which were calculated from the daily number of positive cases identified in individual counties, see e.g. \cite{JarvisEtAl2021} and \cite{HotspotGermany} for respective rules in England and  Germany and \cite{HaleEtAl21} for a general overview on the pandemic policies applied in different countries.
 These strategies  have been only partially successful  and quickly ended in national lockdowns.

 In general, population sizes differ greatly between counties. Heavy-tailed distributions, like the Zipf law or a log-normal distribution, fit very well the distribution of the county
 population sizes for many countries, for France, Poland and Japan see e.g. Figure~\ref{fig_Heavy_Tailed_Dist}. As a consequence, when mitigation strategies are based on relative incidences, many individuals need to be tested positive
 in counties with large population sizes  until measures are applied.
In  counties with small population on the other hand, size thresholds for interventions are already exceeded  when only a relatively small number of individuals gets tested positive. 
 This is particularly problematic, since pathogens often start to spread first within metropolitan areas with large population sizes and only afterwards hit the rural countryside which typically make up counties with relatively small or intermediate population sizes, see for example the initial spread of COVID-19 in the US \cite{Stier2020COVID19AR}.

 \begin{figure}[t]
\begin{center}
\includegraphics[width = \textwidth]{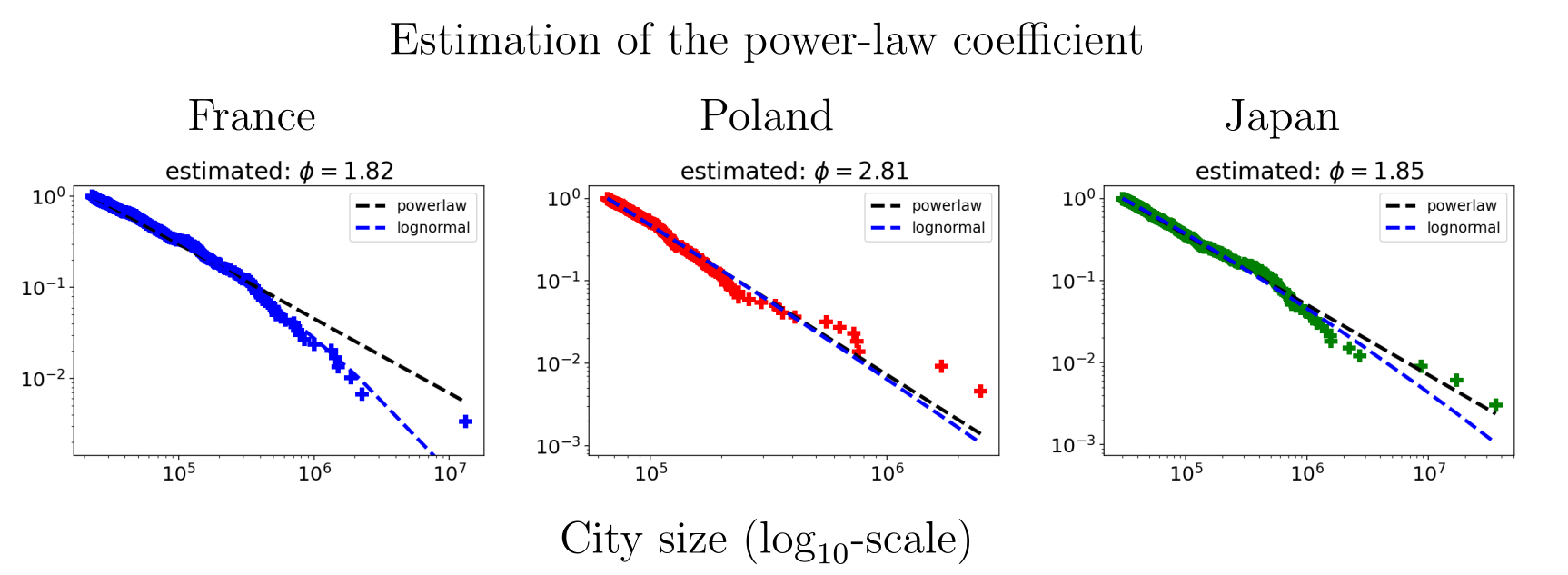}
\end{center}
\vcm{-0.7}
     \caption{Heavy-tailed distributions were fitted to the size distributions of areas of attraction in France and Poland as well as of prefectures in Japan.}
 \label{fig_Heavy_Tailed_Dist}
 \end{figure}

\medskip
To compare the consequences  of the two types of mitigation strategies applied during the COVID-19 pandemic, we  consider in this paper two different criteria at which interventions are imposed, either
when a certain proportion of individuals (later called strategy~\((P)\))
or when a certain fixed number (later called strategy~\((U)\))   gets tested positive in a regional unit.

Specifying thresholds that are determined solely by the unit sizes 
is not merely a simplification for modeling purposes; it mirrors the actual measures implemented during the COVID-19 pandemic 
(e.g. in Germany and New Zealand), 
where authorities often relied on regional population metrics to guide intervention strategies.
Such a choice was primarily motivated by the convenience of implementing them on a national scale,
rather than by a firm evidence of their superiority in regulating propagation.

 We build a simplified model in which 
 the regional units are represented as the vertices of a random graph. 
 From now on, these units are referred to as cities for simplicity. 
 We only  follow  the  spread of the infection  between  cities  and not within a city.  
 The infection spreads according to a SIR model between cities in the sense that each city can participate in the epidemic only once and is isolated 
 (put under lockdown) 
  till the end of the epidemic
 after a certain number of citizens get infected.
A description of preventive measures in this way greatly simplifies the analysis.
Even though complete isolation of a city is a rather extreme scenario,  it is a common pattern that the spread of an arising pathogen occurs in waves. Successive waves may occur if preventive measures are applied that slow down 
pathogen propagation
 and/or if the newly arising pathogen causes at least temporal immunity. Here, we consider a single wave and for simplicity approximate this wave by a SIR-process.

From an infected city, the disease is transmitted to other cities depending on the connectivity between the cities. We consider here two model variants, which we term transportation graph (abbreviated as TG) and kernel graph (abbreviated as KG) in the following. For each pair of cities $(i,j)$ in the transportation graph, the  probability of transmitting the disease from city $i$ to city $j$ is based on the strength of mobility from city $i$ to city $j$. We estimate the strength of mobility between pairs of cities for France, Poland and Japan using commuting data from these countries. 
For the kernel graph, we simplify the model further. We assume that the transmission probability of the infection from city $i$ to city $j$ depends only on the sizes $x_i$ and $x_j$ of the two cities. More precisely, we assume that the transmission probability can be expressed in terms of a kernel function $\kappa_{a, b}(x_i,x_j)$ with parameters $a$ and $b$.
This function $\kappa_{a, b}$ expresses the (country-wide) relationship
between the size of a city and the probability of traveling to and out of that city. The parameters $a$ and $b$ are estimated for France, Japan and Poland,
with commuting data from these countries (as for the transportation graph).

A primary objective of this study is to compare the effectiveness of the two regulation strategies 
 and also to judge the adequacy of model reduction in this scope.
We measure the effectiveness of a strategy by calculating the number of people subject to isolation in relation to a prescribed number of individuals affected by the disease.
The smaller this ratio the more efficient the strategy, since  less persons need to be isolated. 
We analytically benchmark the two containment strategies~\((P)\) and~\((U)\) against each
other in the scenario 
where cities have an equal chance of being contaminated
 by infected visitors
or by inhabitants who got infected during travel
(corresponding to the parameter choice $a = 1+b$). 
In that case, we can determine under
reasonable and simple conditions on the parameters which strategy performs best when
assuming that the numbers of individuals who get infected are under both strategies the
same, see Section \ref{sec_comp}. In the general case, we calculate numerically the theoretical infection
probabilities to compare the two strategies.

According to the estimated best fit parameters for France and Japan, both strategies
perform almost equally well, while strategy~\((U)\)  is preferred over strategy~\((P)\) for Poland.
This result was not expected.
To validate our approximation of the empirical mobility matrix by an (at most) rank-2-
kernel matrix, we simulated epidemics according to the fitted kernel model  as well as with
the fitted TG model. When comparing for the two types of simulations,
several statistics like the probability that a city gets infected during an epidemic  or the
probability that a city of a certain size generates an outbreak, we find overall a
pretty large correspondence for France. For Poland, the correspondence is in general weaker. For
Japan, pretty large difference can be observed.
A possible explanation could be that 
the geographic structure between the cities of Japan
plays a significant role, 
which is neglected in the kernel model.

Before  presenting the detailed methodologies in Sections \ref{sec_model}, \ref{sec_analysis_br} and \ref{sec_data_analysis}, which cover modeling, the probabilistic analysis of the branching process, and data analysis respectively, we give
an overview of our key findings
in Section~\ref{sec_main}. In this overview, we summarize the essential insights and illustrate them with simulation results.
In the final Section~\ref{sec_disc}, we discuss and synthesize our results and offer an outlook on potential future research directions.

\section{Overview on the main results}
\label{sec_main}

By presenting our key findings in this section, 
we aim at emphasizing some essential takeaways without  requiring the reader to engage with the underlying mathematical framework or the data calibration processes,
which will be described in the subsequent sections.

\begin{figure}
\begin{center}
		\includegraphics[width = \textwidth]{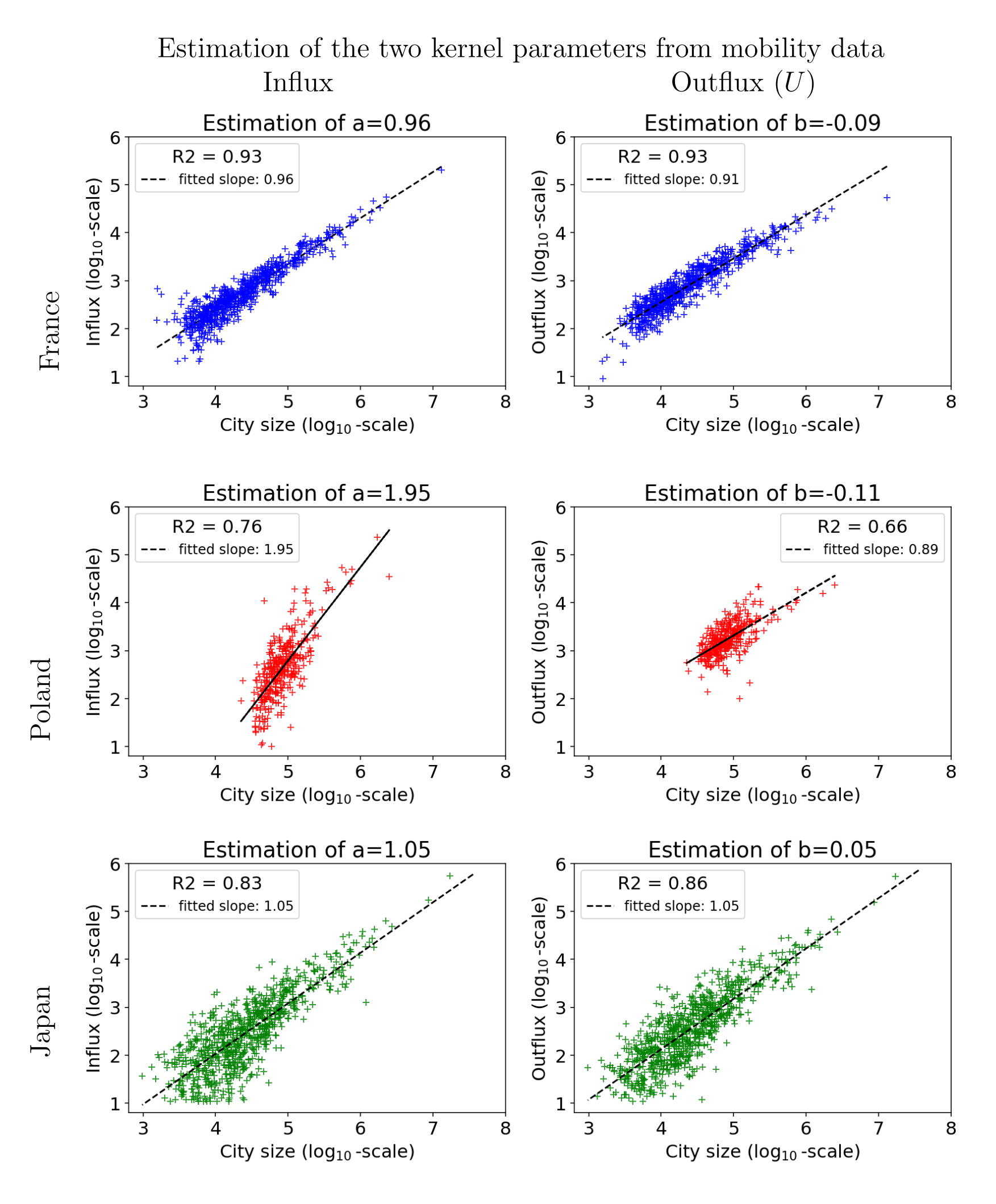}
\end{center}
\vcm{-0.7}
	\caption{
    Inbound travel counts (left column) and outbound travel counts for the estimation of $a$ and $b$, resp., for France, Poland and Japan. }
      \label{fig_Estimate_A_and_B}
  \end{figure}

As mentioned in the introduction, we build simplifying models that mimic the spread of an infection between cities, see Section~\ref{sec_model} for details. 
The epidemic spread
is encoded by a directed random graph, termed the epidemic graph,
where the nodes of the graph are the cities. 
The set of cities that  would  eventually get infected if an infection process starts in city $i$
is given by the forward connected component of city $i$, 
which is the set of cities $j$ for which there is a directed path going from $i$ to $j$.
 
The probability for an edge  between  city $i$ and $j$ to be present in this graph depends essentially on 
two quantities:
the number of individuals that get infected in city $i$ (from its infection till its isolation) and  the travel from city $i$ to city $j$ and back.

We consider two city isolation strategies,
reflecting different criteria on the number of people infected in the city before the lockdown is implemented. 
Under strategy~\((P)\),  the number of infected persons is 
proportional to the city size.
Under strategy~\((U)\),  the number of people that get infected in an infected city is for all cities the same.  So we assume that the number of infected individuals in an infected city until  isolation is  in total
$L(x)= px$ under strategy~\((P)\) and
$L(x)= L$ under strategy~\((U)\).

Each infected person (from city $i$) can directly transmit the disease to city $j$ by traveling there and infecting citizens of that city. On the other hand, people from city $j$ can contract the disease in city $i$ 
while traveling if they meet an infected person. We use commuting data to estimate the number of individuals traveling 
from city $i$ to city $j$ and back and store this information in a matrix $(M_{i,j})_{i,j}$. 
So if city $i$ is infected   the number of individuals bringing the disease to city $j$ is  on average 
\begin{align}\label{eq:transmission_matrix}
 k_{C} \frac{L(x_i)}{x_i} (M_{i,j} + M_{j,i})
 \end{align}
where $k_{C}$ is reflecting the intensity of travel,  contact and virulence of the pathogen.
The first summand is accounting for infected citizens from city $i$ who travel to city $j$ and the second summand for citizens from city $j$ who travel to city $i$ and contract  the disease from infected individuals there. 

We have assumed for clarity and simplicity that the number of individuals
bringing the disease is proportional to 
the total fraction of positive cases in infected city $i$, i.e. $L(x_i)/x_i$.
Though it would be meaningful to state it instead in proportion 
to the integral of the incidence up to the time of isolation,
both quantities are expected to scale in the same way with $x_i$.
By adjusting $k_C$, we thus expect the proportionality to $L(x_i)/x_i$
to be relevant as well,
as explained in Section~\ref{sec_exp_form_prob} 
of the Supplementary Material.

To generate the epidemic graph, 
we thus retain any directed edge $(i, j)$ from the complete graph 
with some probability $p_{ij}$, akin to percolation processes,
where 
\begin{equation}\label{m_pM_ij}
p_{ij} = 1- \exp\left[-k_C\cdot ( M_{ij} + M_{ji}) \cdot \left(\frac{L(x_i)}{x_i} \right)\right]\,.
\end{equation}
The exponential form of $p_{ij}$
is derived from an approximation
based on large city sizes,
see Section~\ref{sec_exp_form_prob} for more details.

Whenever
epidemic graphs are based on matrices $(M_{i,j})_{i,j}$ which respect the pairwise estimates of mobility between cities, we call them \emph{transportation graphs} (TG, for short), see also Section~\ref{sec_model} and in particular Subsection~\ref{sec_TG}.

We used mobility data from France, Poland and Japan to analyze which parameter regimes are empirically relevant for the matrix $(M_{i,j})_{i,j}$.
We exploited commuting data, usually at the municipal level, that were publicly available from statistical offices from France, Poland and Japan.
Before application, the data were processed by aggregating data into regional units tailored to each country’s functional urban areas. For further details, refer to Section~\ref{Sec:MobilityData}.
Since mobility between nearby cities might be stronger influenced by proximity of the cities instead of the city sizes and we aim to depict with our model rather the spread of a disease on a nationwide level rather than on a regional level, we filtered the traveler counts by excluding mobility between municipalities closer to each other than 30 km (abbreviated as D30+ in the figures). The rationale for this threshold is that the mean commuting time in European countries is roughly about 40 to 45 minutes per day, see \cite{NadalEtAl2022}, and that e.g. in Germany more than 78\% of all persons in employment have a distance of at most 25 km to work and roughly 5\% travel more than 50 km or more to work \cite{CommuterDistance}.

Infection processes based on transportation graphs can be simulated, but they are in general too complex to be analyzed theoretically. Hence, we reduce the model further by assuming that
the traveling probabilities to and from a city both depend only on the city size. In Figure~\ref{fig_Estimate_A_and_B}, inbound and outbound travel counts are plotted for France, Poland and Japan. We see in these plots that it is very natural
to consider in the kernel a power-law relation between inbound travels and
city size and similarly for the relation between outbound travels and city size.
We thus
consider mobility matrices of the form 
$$M_{i,j} = k_M \ x_i^{1+b} \ x_j^a,$$
for some $a,b \in \mathbb{R}$ and 
some constant $k_M$, scaling the traveling intensity appropriately.
Here $x_i^b$ 
is (asymptotically) proportional to the probability that an individual of city $i$ leaves the city (when the city is infected)
and $x_j^a$ is (asymptotically) proportional to the probability that an individual that is traveling chooses  city $j$ as a destination. So the average number of
 people that are traveling from city $i$ to a city $j$ is (asymptotically) proportional to $ x_i\, x_i^b\,  x_j^a =  x_i^{1+b}\, x_j^a$.
 See Section~\ref{sec_kernel} for more details on the kernel model.

The chosen kernel greatly synthesizes the migration data.
The two parameters $a$ and $b$ can be estimated separately from respectively the inbound and the outbound travel counts, see Figure~\ref{fig_Estimate_A_and_B}.
Since the model assumes a 
	linear dependence on  the log scale,
we implement a least square regression for each parameter,
described more precisely in Section~\ref{EstimatingAandB}. We call the epidemic graphs that are based on this model \emph{kernel graph} (KG, for short).
In contrast to the transportation graph,
	  a kernel graph can be produced 
	for an arbitrary number of cities or a different city size distribution, while keeping the same mobility parameters $a$ and $b$.
 \afterpage{\clearpage}

We scale $k_M$ with $1/N$ to approximate  the initial spread via a branching process, as well as backward in time potential routes over which a city can get infected,
as explained in Sections~\ref{sec_fwd_br} and \ref{sec_bwd_br}.
Remark that classical theory that investigates the quality of this approximation 
relies on certain moment conditions on the kernel to be fulfilled that, in our context,
might ''asymptotically" not be appropriate for our data
on city size distribution.
In case we consider as a proxy a powerlaw distribution,
the condition is simply stated in terms of a lower-bound 
of the powerlaw coefficient
as compared to $a$ and $b$.
See Section~\ref{sec_inf_cities} for more details
and Section~\ref{EstimatingPhi} about the  inference of such a powerlaw coefficient from our data.
\smallskip

We use the branching process approximations to derive the following implicit formula for the probabilities $\pi_P(x)$ (resp. $\pi_U(x)$)  that a city of size $x$ gets infected during a (large) outbreak (in the following termed infection probability) under strategy~\((U)\) ( resp. under strategy~\((P)\)): 
\begin{equation}\label{piP}
	\pi _P(x) 
	= 1- \exp\Big[  - k_B\, p_\vee  \Big( x^a \int_0^\infty y^{b+1} \pi_P(y)    \beta (dy) 
	+  x^{b+1} \int_0^\infty y^{a} \pi _P(y) \beta (dy) \Big) \Big],
\end{equation}
and
\begin{equation}\label{piU}
	\pi_U(x) 
	= 1- \exp\Big[ - k_B\, L_\vee  \Big( x^a \int_0^\infty y^b \pi_U(y)   \beta (dy) 
	+  x^{b+1} \int_0^\infty y^{a-1} \pi_U(y)  \beta (dy) \Big) \Big],
\end{equation}
where $\beta$ denotes the city size distribution and $k_B = N\, k_C\, k_M$, see Section~\ref{sec:infection probability} for more details.
Similarly, the asymptotic outbreak probabilities $(\eta _P(x))_{x\ge 0}$ and $(\eta_U(x))_{x\ge 0}$,
that are the probabilities that a city of size $x$ initiates an outbreak
for resp. Strategies~\((P)\) and \((U)\),
are  approximated via the following recursive formulas:
\begin{equation}\label{etaP}
	\eta_P(x) 
	= 1- \exp\Big[  - k_B\, p_\vee  \Big( x^a \int_0^\infty y^{b+1} \eta_P(y)    \beta (dy) 
	+  x^{b+1} \int_0^\infty y^{a} \eta_P(y) \beta (dy) \Big)\Big],
\end{equation}
and
\begin{equation}\label{etaU}
	\eta_U(x) 
	= 1- \exp\Big[ - k_B\, L_\vee  \Big( x^{a-1} \int_0^\infty y^{1+b} \eta_U(y)   \beta (dy) 
	+  x^{b} \int_0^\infty y^{a} \eta_U(y)  \beta (dy) \Big) \Big],
\end{equation}
see Section~\ref{sec:probability trigger outbreak} for more details.
The recursive nature of the equations gives rise to an iterative procedure to compute
numerically these quantities, see Section~\ref{sec_KB_errors} in the Supplementary Material
for the details of the implementation.

To compare the efficiency of both strategies, we express the expected number of eventually infected individuals
 in terms of the infection probabilities $\pi_P$ and $\pi_U$,
as well as the expected number of individuals eventually under isolation.
Under our assumption
for strategy~\((P)\), the expected number of eventually infected individuals is derived from the branching approximation as:
\begin{align}\label{defI_P}
	I_P := \int_0^\infty p_\vee\cdot x\cdot \pi_P(x) \beta(dx).
\end{align}
while for strategy~\((P)\), the expected number of eventually infected individuals is
\begin{align}\label{defI_U}
	I_U := \int_0^\infty L_\vee\cdot \pi_U(x) \beta(dx),
\end{align}

Similarly, the expected number of individuals eventually under isolation
is (under our assumptions) given by
\begin{equation}
Q_P := \int_0^\infty x \cdot \pi_P(x) \beta(dx), \qquad Q_U := \int_0^\infty x \cdot \pi_U(x) \beta(dx).
	\label{QDef}
\end{equation}

If the same number $I_\star$ of individuals gets infected under both strategies 
 (by  adjusting only the parameters $L_\vee$ and $p_\vee$), we consider the strategy that requires fewer individuals to be isolated as the more efficient one.
In the case where $a=b+1$, we find that strategy~\((U)\) is more effective than strategy~\((P)\) if $a>1$,
that strategies ~\((U)\)  and~\((P)\) perform equally well if $a=1$ (and $b=0$) and that strategy~\((P)\) outperforms strategy~\((U)\) otherwise, as stated in Proposition~\ref{ineq}, see more generally Section \ref{sec_comp} .
It is important to note that these rankings hold true regardless of the value of $I_\star$.

Since for France and Japan, $a\approx 1+b$ and $a\approx 1$, see Figure~\ref{fig_Estimate_A_and_B}, we conclude that both strategies perform roughly equally well for these countries.
\smallskip

\begin{figure}[t!]
\begin{center}   
    \includegraphics[width = \textwidth]{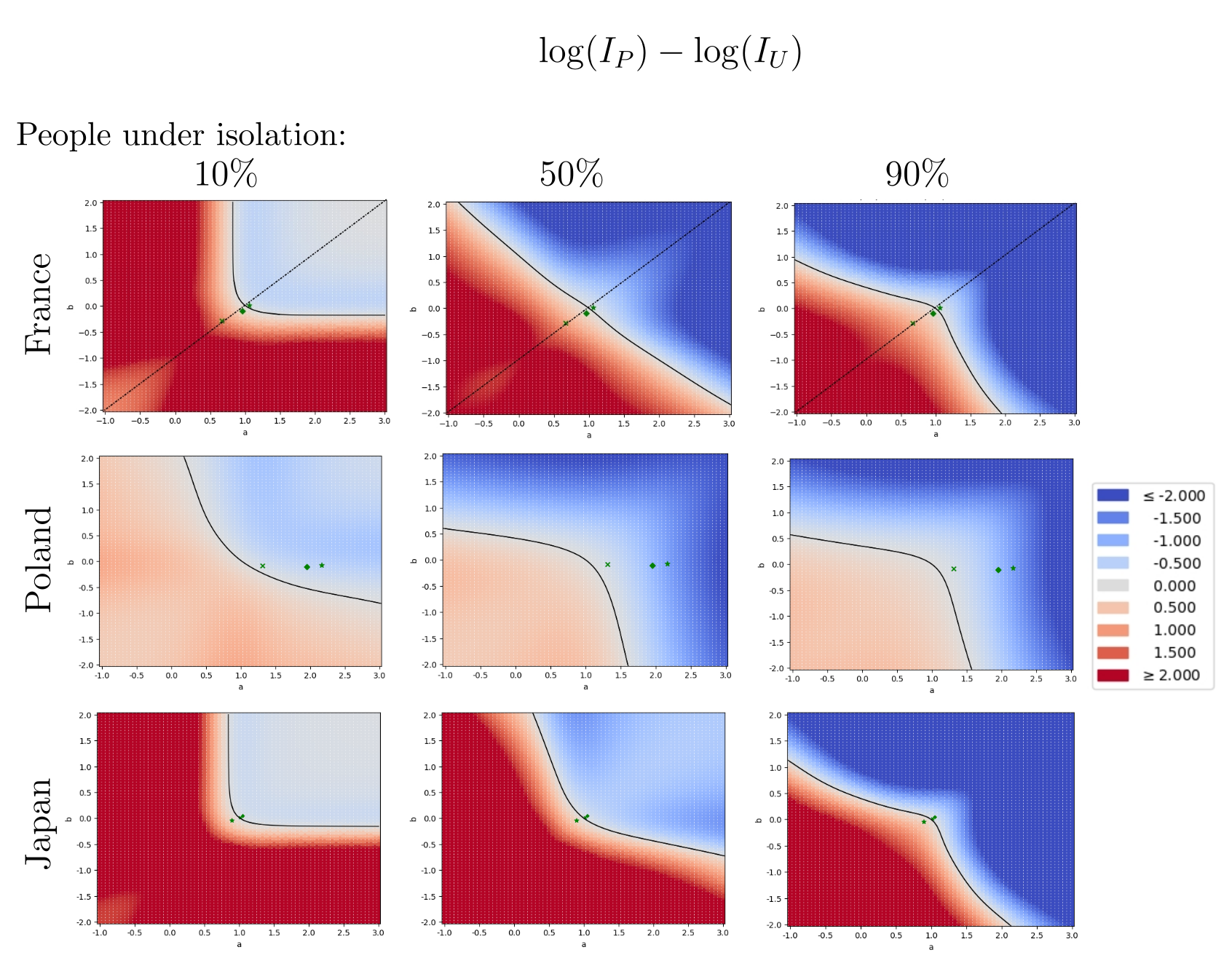}
	\end{center}
    \caption{  
    Logarithm of the ratio of the expected number of isolated individuals under strategy~\((U)\) vs. strategy~\((P)\) for a fixed number of infected individuals.  The colour code for the values of $\log(I_P)-\log(I_U)$ is given on the right of the figures.
Green points represent estimates of $(a,b)$ for each country,
$\times$, $\bullet$, $\star$, resp. give estimated values of (a,b) for various filtering on the distance, described as 
D1+, D30+ and D50+, resp., see Section~\ref{Sec:MobilityData}.
 For the black line $I_P= I_U$ and for the dashed black line $a=b+1$. }
\label{fig_Numerical_Comparison_UP}
\end{figure}

In the general case, where $a$ is not necessarily equal to $1+b$, as estimated for Poland, we do not have analytical results on the efficiency of strategy~\((U)\) in comparison to strategy~\((P)\) and vice-versa.
Instead, we compare strategy~\((U)\) and \((P)\) numerically. 

We consider the city size distributions 
from France, Poland and Japan,
while the parameter range  covers the  values of $a$ and $b$ which  we estimated with mobility data from France, Poland and Japan. Furthermore, we consider three different degrees of  severity of outbreaks, with the fraction of the population in isolated  cities  set to 10\%, 50\% or 90\%.
Fixing this fraction determines uniquely the value $Q_\star$ and, as a consequence of \eqref{QDef}, \eqref{piP} and \eqref{piU}, also the product $m_p:= k_B\, p_\vee$ for strategy \((P)\) and the product $m_L:=k_B\, L_\vee$ for strategy \((U)\). 
It then follows from \eqref{defI_P} and \eqref{defI_U}
that the ratio $I_P/I_U$ does not depend on $k_B$.
Hence, in the numerical calculations, $k_B$ can be chosen as an arbitrary positive real number,  that fulfills $\tfrac{m_p}{k_B} \in [0,1]$  and $\tfrac{m_L}{k_B} \leq \min_x\{\beta(x)\}$.

Given the values of $p_\vee$, $L_\vee$, $k_B$ as well as $a$ and $b$,  we calculate  numerically the values of $I_P$ and $I_U$, i.e. the proportion of infected individuals  under strategy~\((P)\) and \((U)\) conditioned on $Q_{\star}$, and plotted $\log(I_P/I_U)$ in Figure~\ref{fig_Numerical_Comparison_UP}. 
In these figures, we also added the estimated values of $a$ and $b$ for France, Japan and Poland, respectively. 
\smallskip

Based on these estimates, 
strategy~\((U)\) consistently outperforms stra\-tegy~\((P)\) in Poland, where $a\approx 2$ and $b \approx 0$.
Concretely, the number of infected persons is roughly three times larger under strategy~\((U)\) when the fraction of quarantined cities is set to 90\%, while it is roughly 1.5 times larger when  this fraction is set to~10\%.
We see that changing the type of isolation strategy may strongly impact the outcome of the epidemic. 

We confirm the theoretical ranking 
along the axis $a=1+b$,
and in particular that both strategies perform approximately equally well in France and Japan.
However, we can observe that the function $I_P/I_U$ exhibits intriguing patterns and is strongly influenced by the value of the fraction of quarantined cities as well as the city size distribution, see Figure~\ref{fig_Numerical_Comparison_UP}. 

For Poland for instance,  the curve at which both strategies perform equally well, i.e.~at which $I_P= I_U$, exhibits a transition from a concave shape when a large fraction of cities are quarantined to a convex shape when such a fraction is small. Conversely for France, this curve does not exhibit a distinctly concave or convex shape across two of the three scenarios considered within the specified parameter range of $a$ and $b$.
However, we can observe  for all considered city size distributions that
we have a transition from a more concave to a more convex curve for decreasing fractions of quarantined cities. 
\medskip

\begin{figure}[t!]
	\begin{center}
    	\includegraphics[width = \textwidth]{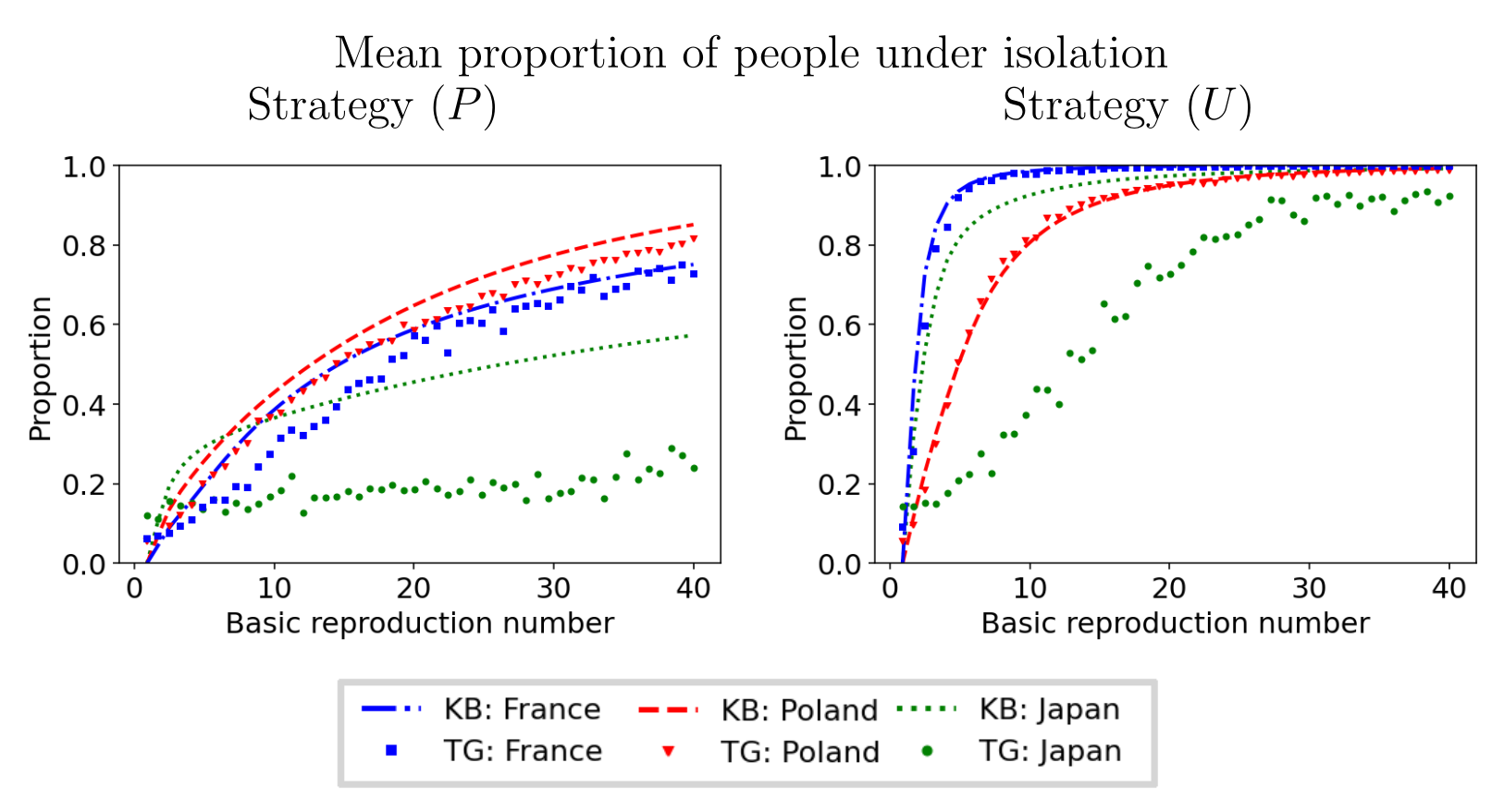}
	\end{center}
    \vcm{-0.7}
	\caption{
		Comparison of the mean TG and KB proportions of people under isolation 
			depending on the basic reproduction number. Dashed lines show proportions when simulations are performed according to the transportation graph (TG), and scatter plots when
            proportions are calculated with KB approximations,  
            for France (in blue), Poland (in red) and Japan (in green),
			under strategy~\((P)\) on the left and strategy~\((U)\) on the right.
	}
	\label{Fig_incidence_ppl}
\end{figure}

To check whether the alignment of the TG and KG infection and outbreak probabilities is robust for varying epidemic strength,
there are two natural summary statistics,
namely the proportion of people (resp. cities) under isolation.
We compare the analytically derived  proportions with the corresponding simulated proportions, that can be calculated from epidemics generated  on the basis of TG and KG graphs for France, Poland and Japan, see Figure~\ref{Fig_incidence_ppl}. In the figures, we abbreviate the analytically derived curves by KB (K for {\em kernel} and B for {\em branching}), since they are based on the branching processes approximations.
In the following,
we refer to the mean proportion of people under isolation
calculated from simulations based on the transportation graph 
(resp.~the kernel graph) by the abbreviated notations
of TG (resp.~KG) proportion under isolation.
\smallskip

For Poland and France, 
we find an agreement of the TG proportion under isolation with the KG proportion of isolation which is surprisingly good given the high level of reduction to arrive at the KG graph. 
The correspondence is almost perfect under strategy~\((U)\).
Under strategy~\((P)\), we observe a downward shift for Poland.
For France under strategy~\((P)\), the KG proportion under isolation 
displays a more regular shape
that appears to interpolate 
between the TG proportion under isolation for resp. small 
(less than 5) and large values (more than 20) of the $R_0$. 
For Japan however,
the correspondence between KG and TG proportions under isolation is 
very weak,
which shows that it is not  appropriate for every country to neglect the geometric structure of a country and base the mobility matrix solely on city sizes.

\begin{figure}[p]
	\begin{center}
    \includegraphics[width = \textwidth]{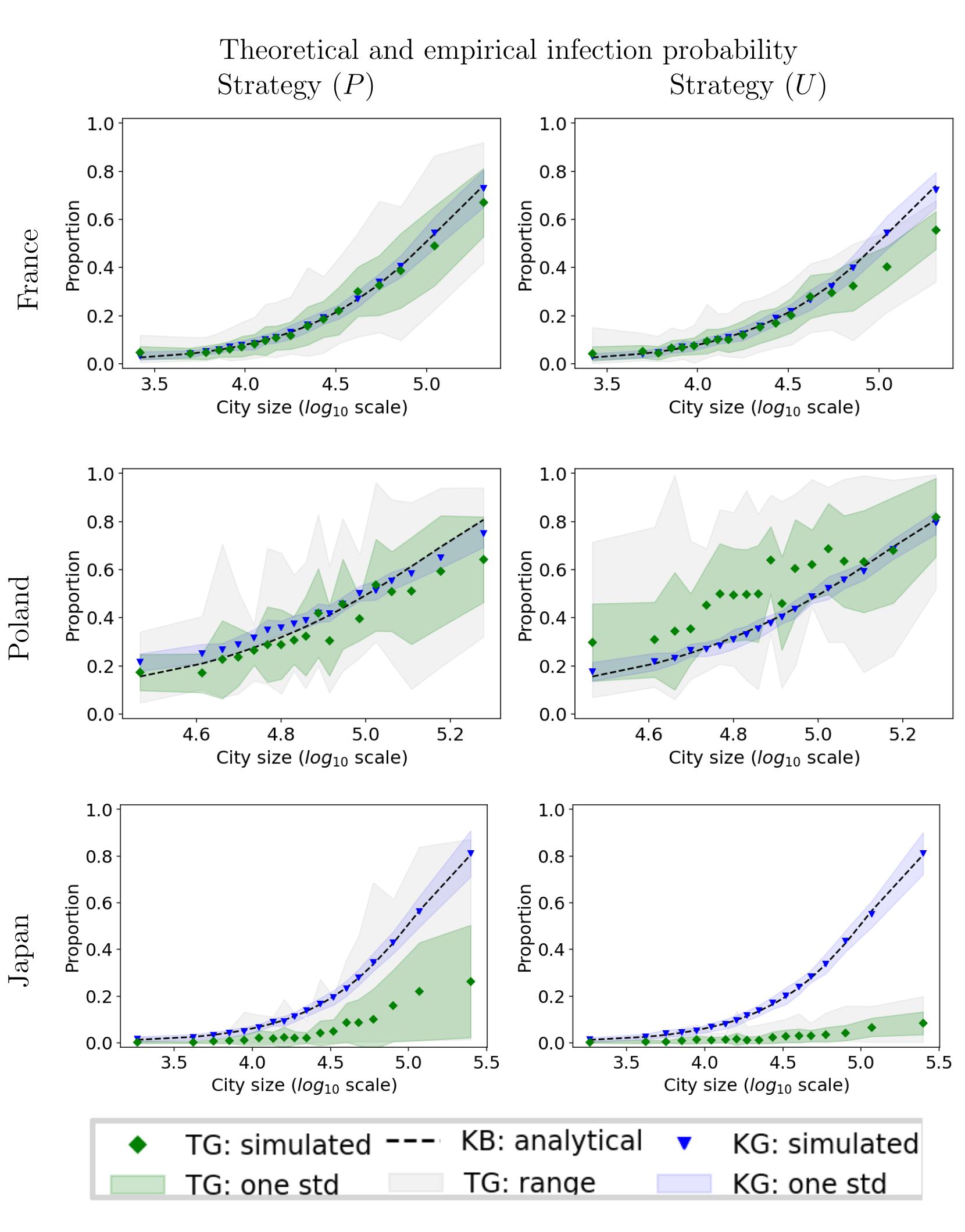}
	\end{center}
    \vcm{-0.7}
	\caption{Comparison of simulated and theoretical infection probabilities, on the left for strategy~\((P)\) and on the right for strategy~\((U)\) for mobility data from Poland and Japan in the upper and lower row, resp.
		The $R_0$ value is adjusted such that  the theoretical infection probability is 0.5 for cities of size $10^{5}$, see Section \ref{sec:infection probability}.}
	\label{Fig_Compare_Analytical_PI_Empirical_PI}
\end{figure}

\begin{figure}[p]
	\begin{center}   
    \includegraphics[width = \textwidth]{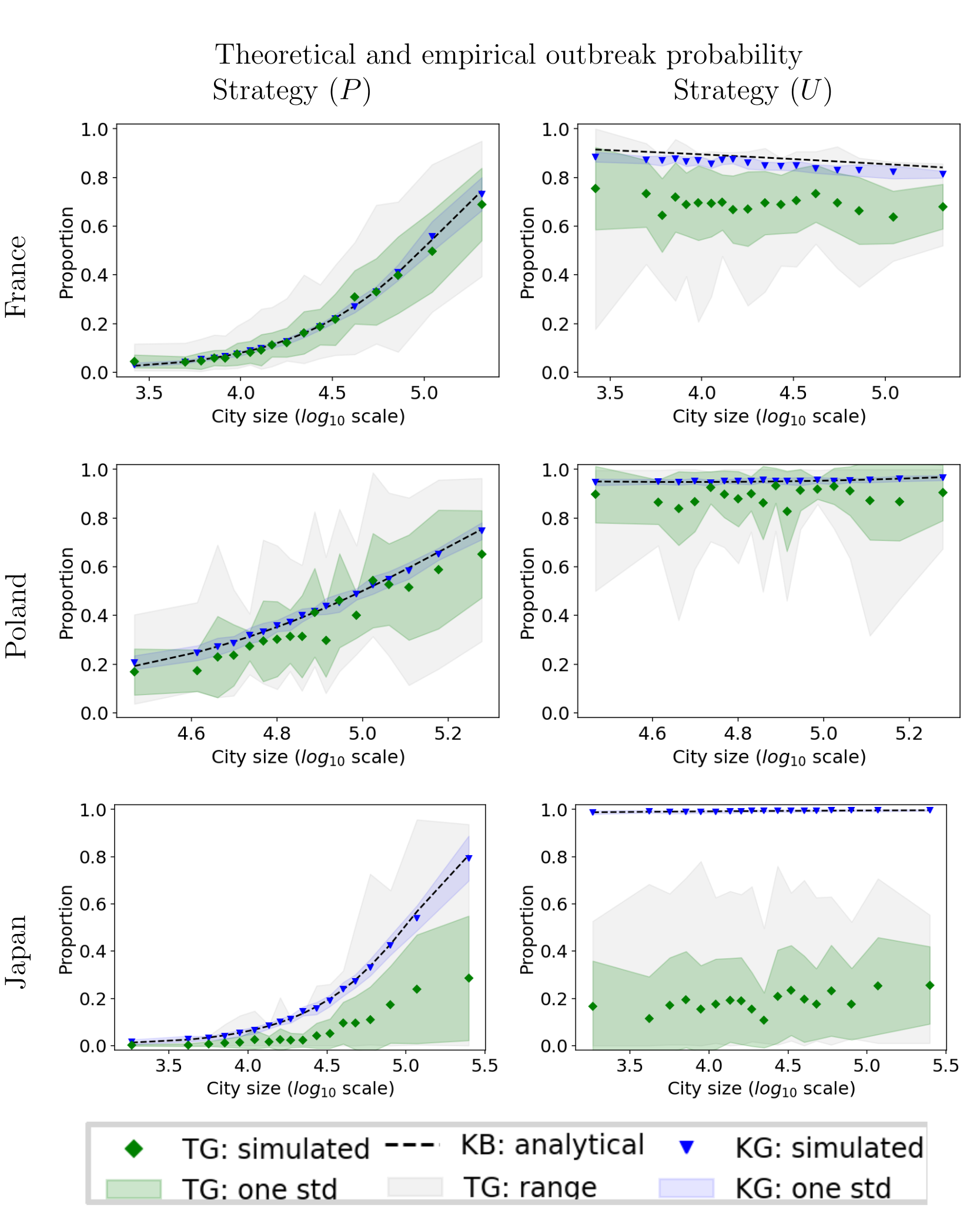}
	\end{center}
    \vcm{-0.7}
	\caption{Comparison of empirical and theoretical outbreak probabilities, on the left for strategy~\((P)\) and on the right for strategy~\((U)\) for mobility data from France, Poland and Japan in the upper, middle and lower row, resp. The $R_0$ value is adjusted such that  the infection probability is 0.5 for cities of size $10^5$, see Section \ref{sec:infection probability} for more details.}\label{Fig_Compare_Analytical_PO_Empirical_PO}
\end{figure}

In the following plots, we compare the infection and outbreak probabilities calculated on the basis on the transportation graph (TG), kernel graph (KG) and the kernel branching process approximation (KB). 
For France, the TG infection probabilities are in very good agreement with KB  and KG infection probabilities, see Figure~\ref{Fig_Compare_Analytical_PI_Empirical_PI}.  For Poland, the correspondence is actually worse, despite the better agreement observed in Figure~\ref{Fig_incidence_ppl}. 
For Japan, TG infection probabilities 
deviate from KG and KB infection probabilities over a broad range of city sizes,
with a more significant impact of the largest cities on the proportion of people under isolation. 
In contrast, 
KG and KB infection probabilities
are in very good agreement
whatever the country, 
which shows that  the branching process approximation is appropriate at least for the considered cases.

Similarly, TG outbreak probabilities 
	are in a very good agreement 
	with KG and KB outbreak probabilities
	under strategy~\((P)\) for both France and Poland, see Figure~\ref{Fig_Compare_Analytical_PO_Empirical_PO}.
     This is expected, since under strategy~\((P)\) infection and outbreak probabilities coincide, see \eqref{piP} and \eqref{etaP}.

     On the other hand, we see that under strategy~\((U)\),
TG outbreak probabilities typically do not agree well with KB nor KG  outbreak probabilities.
The geographical structure has then a much stronger influence, 
leading to a reduced propagation of the disease as compared 
to a dissemination that would be global.
A criterion for comparing mitigation strategies based on infection probabilities is therefore expected to provide more reliable comparisons
with respect to the specifics of migration data
in contrast to those based on outbreak probabilities.

More details on the quantities that we evaluate for these comparisons 
of respectively infection and outbreak probabilities
are provided in Section~\ref{sec_simulated_infection_and_outbreak_prob}. 
In particular, the above conclusions are also discussed further in terms of a more quantitative comparison of the probabilities.

The strong association between the TG proportions under isolation and infection probabilities as determined by the KG values 
for France and Poland argues for approximating the epidemic by the KB approximation.
Consequently, it is also reasonable to use the KB approximation to infer numerically the relative advantage of one strategy over the other.
The fact that the approximations used for the comparison presented in Figure~\ref{fig_Numerical_Comparison_UP} 
are legitimate  reinforces the conclusions drawn from it:
 for France both strategies would perform equally well
 under any scenario, while for Poland,
strategy \((U)\) would generally outperform strategy \((P)\).
\medskip

The kernel branching approximation motivates also the definition of a basic reproduction number, that can be expressed explicitly and be calculated easily numerically, 
as described in Section~\ref{Section Reproduction Number},
see in particular \eqref{iniquity sim wickedness}
and \eqref{dry run sim rehearsal}.

As we show in Section~\ref{sec_adj_str},
we can directly relate 
the thresholds $p_\vee$
and $L_\vee$ of the two strategies~\((P)\)
and \((U)\) to the ratio of the corresponding reproduction numbers.
The most crucial comparison 
is between the thresholds $p^*_\vee$
and $L^*_\vee$ that correspond to a basic reproduction of 1, 
as it relates to the emergence of large outbreaks.
As a first guess, one might expect the  corresponding ratio $L^*_\vee/p^*_\vee$
to be close to the average city size.
However, it may be much larger than the average size obtained by selecting a city uniformly at random. 
In the context of France, this ratio more closely aligns with the average size of the city inhabited 
by a randomly chosen individual,
which is many times larger than the former.

On the other hand, it seems out of reach to accurately assess such basic reproduction numbers based on the temporal increase in the number of infected cities, as demonstrated by the findings in Section~\ref{sec_val_R0}.

We used data from 
the COVID-19 pandemics
to provide an example
for an empirical estimate
of a (city-wise) basic reproduction number,
see Section~\ref{sec_empirical_R0}.
\medskip

To gain a deeper understanding  of the correspondence of the transportation and kernel graphs for different countries, we performed an additional comparative analysis of the distributions of expected in-degree and out-degree as functions of city size.
It appears that many of our observations
regarding the alignment of the infection and outbreak probabilities
are closely linked to  the alignment
of the in- and out-degree, as more extensively discussed in Section~\ref{sec_degree}.
For instance, 
out-degrees in the transportation graph  are notably smaller than those in the kernel graph, particularly in Japan.
In addition, the largest cities of Japan have smaller in-degrees in the transportation graph than 
in the kernel graph, which can lead to significant disparities in the overall epidemic spread.

\section{Models}
\label{sec_model}

\subsection{A model for the spread of an infection between cities}
\label{sec_epidemic_graph}

The underlying mathematical foundation for our analysis
is the relation between SIR epidemic models
and inhomogeneous random directed graph from \cite{CO20},
see also \cite{BJR06} for the undirected case.
The vertices of the random graph are interpreted
as the cities within a country, with a given ordering from $1$ to $N$.
To generate the epidemic graph, directed edges \((i, j)\) from the complete graph are retained with probability \( p_{ij}\). This procedure is similar to percolation, with each edge being sampled independently.

In this epidemic graph,
a directed edge $(i,j)$ between city $i$ and city $j$ represents 
a potential infection event, in the sense that, if city $i$ would get infected, it would spread the infection to city $j$ (given that city $j$ is not yet infected). Hence, the forward connected component of city $i$ in such a random directed graph represents the set of cities that  would  eventually get infected if an infection process starts in city $i$. 
In other words,
from any choice 
of an initial epidemic source, it is possible to determine all the cities that become infected by analyzing the epidemic graph. 
These cities correspond to the vertices accessible by a path issued from the epidemic source, formed by the oriented edges of the graph.
The relation of this conclusion to a discrete-time SIR epidemic model
is well described, e.g. in~\cite{CO20}.
For $k\ge 1$, the $k$-th generation of infected cities
then consists of the vertices whose graph distance to the source of infection
is exactly~$k$.
\smallskip

We specify the probabilities $p_{ij}$
by accounting first for mobility between cities,
second for the effect of local restrictions
and third for a transmission factor.
We take as an input some mobility matrix
$(M_{ij})_{i, j\in \II{1, N}}$ that reflects the migration level
of citizens from city $i$ to city $j$.
Local restrictions limit the number of infected individuals in a city. A function $L:\bR_+\to \bR_+$
that defines this number for any city size reflects the strategy of mitigating the epidemic (which is only based on city sizes). In an infected city of size $x$, a fraction $L(x)/x$ of all citizens is infected. 
A scaling factor $k_C>0$ is considered in addition
to adjust the strength of disease propagation.
The following expression for $p_{ij}$ is considered:
\begin{equation}\label{pM_ij}
p_{ij} = 1- \exp\left[-k_C\cdot ( M_{ij} + M_{ji}) \cdot \left(\frac{L(x_i)}{x_i} \right)\right]\,.
\end{equation}
Note that the disease can disseminate to city $j$ due to infected citizens of city $i$ 
visiting $j$, which leads to the term proportional to
$M_{ij} L(x_i)/x_i $,
but also due to the citizen of city $j$
visiting $i$ and bringing back the disease.
The later is represented by 
the term proportional to
$M_{ji} L(x_i)/x_i$,
with the assumption that any visited citizen would get infected in city $i$ at a rate 
proportional to the incidence, thus proportional to $L(x_i)/x_i$.
The latter transmission events are designated as infections from the inside (the inside of city $i$)
in contrast to the former events others considered to be from the outside (since infected citizens of $i$ 
are traveling and transmit the disease outside of city $i$).

The exponential form of $p_{ij}$
is derived from an approximation
based on large city sizes,
see Section~\ref{sec_exp_form_prob} of the Supplementary material for more details.
Remark that we relate the number of infectious contacts between infected individuals of city $i$ and citizen of $j$ 
to the total number of infected individuals  in city $i$ at isolation, i.e. $L(x_i)$.
Alternatively, considering instead the time integral of the infectivity rate over the "infectious period" in city $i$ may be viewed as a more realistic approach.
However, in fact, these two quantities can be expected to be proportional.
In particular, the factor of proportionality can be incorporated in the parameter $k_C$.
This proportionality naturally arises in the context of an epidemic growing exponentially within the cities, assuming that the growth rate is independent of the specific city. 
More details on this relation are also provided in Section~\ref{sec_exp_form_prob}.
\smallskip

We consider two alternative choices of the function $L$
that correspond to two different  strategies
for the epidemic containment.
Strategy~\((P)\)
consists in prescribing a threshold for isolation that is proportional to the city size,
thus given through a proportion $p_\vee$,
i.e.  $L_P(x) := p_\vee \cdot x$.
Strategy~\((U)\) consists in prescribing a threshold $L_\vee$
for isolation that is independent of the city size,
i.e. $L_U(x) := L_\vee$.

In an actual epidemic, it is in general not possible to know the number of infected individuals, due to e.g. the time-shift between the infection of an individual and the detection of its infection and since not all individuals can be tested at any time point. In particular, it is in general not possible to know with  certainty, if a given threshold $L$ has been reached. In our model, we neglect these uncertainties. 
\smallskip

The migration matrix $(M_{ij})_{i, j}$ 
and the city sizes $(x_i)$
are to be fitted with data,
that are taken in our analysis  from France, Poland and Japan.
 Next, we explain how this fit can be obtained for the different models of approximation.

\subsubsection{The transportation graph model}
\label{sec_TG}

In the transportation model,
the entries of the matrix are estimated separately according to the same procedure for any pair of cities.
$N\times N$ parameters are then to be inferred from migration data.
This graph 
is expected to retain mobility patterns formed 
by other factors than the city sizes,
e.g. by the spatial distribution of the cities.
We call the corresponding epidemic graph  the transportation graph (TG).

\subsection{Reduction to a kernel depending solely on city sizes}
\label{sec_kernel}

\subsubsection{The kernel graph model}
\label{sec_KG}
We propose a reduction of the model
through an alternative choice of the migration matrix 
that serves three purposes:
first we want to have expressions as directly related
to our focal feature, the city sizes,
which we expect to be one of the determining factor
of epidemic progression;
second we look for more efficient numerical approaches;
third we are interested to generate an epidemic graph 
that can scale with the number of vertices, 
by relating to the framework of \cite{CO20}.

The reduction that we propose consists in assuming a specific form
of the mobility matrix $(M_{ij})_{i,j\in \II{1, N}}$,
namely that of a kernel function.
The kernel denoted by $\kappa_M$
 receives as inputs the sizes $x_i$ and $x_j$
of respectively the source and the visited cities, 
i.e. $M_{ij} = \kappa_M(x_i, x_j)$.
Several observations
confirmed to us the interest of looking for 
non-symmetrical gravity type kernels
of the following form:
\begin{equation}
\kappa_M(x, y) = k_M\, x^{1+b}\cdot y^a,
\end{equation}
where the two parameters $a$ and $b$ are to be fitted
from the transportation matrix, for $x, y \in \bR_+$.
These observations include those of the influx and outflux, recalling Figure~\ref{fig_Estimate_A_and_B},
but also of the transportation matrix,
see Section~\ref{sec_degree} 
on the degree distribution of this random graph.
While  $y^a$ corresponds to the bias towards a visit of a city of size $y$,
 $x^b$ refers to mobility of citizens of a city of size $x$. If $a>0$ there is a bias of travels being more often directed into large cities,
while if $b>0$ there is a bias of inhabitants of large cities
to travel more often abroad. 

It is an essential property of $\kappa_M$
that it takes a separable form:
it is the product of two terms, 
namely $x^{1+b}$ and $y^a$,
which depend separately 
on $x$ and $y$,
the source and visited city sizes respectively.
This property simplifies the analysis of the branching process
and allows for very efficient numerical implementation of the epidemic process.

The probabilities $(p_{ij})_{i,j\in \II{1, N}}$
are then computed according to \eqref{pM_ij}
in order to produce the random sampling of the epidemic graph
which we name the kernel graph (KG).
Distinguishing the contribution 
of infections from the outside of $i$ 
to those from the inside leads to consider the two kernels 
$\kappa_O, \kappa_I : \bR_+^2\to \bR_+$
defined respectively as:
\begin{equation}\label{kappa_An}
 \kappa_{O}(x,y):=  \frac{k_B}{N}\, L(x) \cdot x^b \cdot y^a = k_C \frac{L(x)}{x} \kappa_M(x,y) 
\end{equation}
and as
\begin{equation}\label{kappa_Bn}
 \kappa_{I}(x,y):= \frac{k_B}{N}\, L(x) \cdot x^{a-1} \cdot y^{1+b} = k_C \frac{L(x)}{x} \kappa_M(y,x).
\end{equation}
We deduce the following identity 
for the probabilities $p^\kappa_{ij}$ involved in the kernel graph:
\begin{equation}\label{pK_ij}
- \ln\big(1- p^\kappa_{ij}\big) :=  \kappa_O(x_i, x_j) + \kappa_I(x_i, x_j)\,.
 \end{equation}

 \subsubsection{The limit of a large number of cities}
\label{sec_inf_cities}
We are interested moreover in a scaling of the kernel graph
with the number  $N$ of vertices.
The proposed form of the mobility matrix
is directly extended 
for any set of vertices with city sizes as attributes,
leading to a similar kernel graph between these vertices.
For this scaling, we thus interpret the 
observed distribution of city sizes
as the sampling of a  probability distribution $\beta$.
To obtain a kernel graph with $N'$ vertices,
we can draw $N'$ city size values 
independently according to $\beta.$

Basically, $\beta$ 
can simply be evaluated as the empirical distribution
$\sum_{i=1}^{N} \delta_{x_i}$, $N$ being very large.
Alternatively,
empirical size distributions of large cities are generally well described by heavy-tailed distributions, like the log-normal distribution which density takes the form 
 \begin{equation*}
 \beta(x) := \frac{\idc{x>0}}{\sqrt{2 \pi} \sigma x}  \exp\left(- \frac{(\ln(x) - \mu)^2}{2 \sigma^2} \right) , 
\end{equation*} 
for the two parameters $\mu \in \mathbbm{R}$ and $ \sigma>0$,
or power-law distribution with a density of the form
	\begin{equation}
\beta(x) 
:= \frac{\idc{x>x_L}}{Z}x^{-\phi}, 
	\label{beta}
\end{equation}	
for some $x_L>0$ and  $\phi > 1$
and 
with an appropriate normalizing constant $Z$, 
as we have seen e.g. in  Figure~\ref{fig_Heavy_Tailed_Dist}.

Because of the very significant and specific role 
of the few largest cities 
in the epidemic outcomes,
we have chosen the empirical distribution
as the baseline city size distribution, in particular to generate Figures~\ref{Fig_incidence_ppl}-\ref{Fig_Compare_Analytical_PO_Empirical_PO}.
Power-law distributions have also been considered
to test the robustness of our conclusions.
\medskip

As  $N'$  tends to $\infty$,
we want the secondary number of infection 
a.s. to be asymptotically
finite.
This is ensured by appropriately scaling the formula we proposed in \eqref{pK_ij}
in relation to $N'$.
For large $N'$,
the probability that there is an edge from $i$ to $j$ 
should be equivalent to $\kappa(x_i, x_j) / N'$, where 
we integrate the two functions $\kappa_O$ and $\kappa_I$ into
\begin{equation}
\kappa(x, y) 
:= k_B\cdot 
	(L(x) \cdot x^b\cdot y^a
	+ L(x)  \cdot x^{a-1}\cdot y^{1+b})\,.
	\label{pVU}
\end{equation}
So the corresponding scaling of $k_M$ with $1/N'$
is a priori the right one
to arrive at a well-behaved epidemic graph  in the limit $N'\rightarrow \infty$, see \cite{CO20}.

In this setting, we compare our model with some branching processes.
In the limit where $N'\to \infty$,
 the  random epidemic graph is locally
a random tree generated by a branching process.  
This observation establishes a connection between large 
	random graphs and branching processes
	\cite{BJR06}.

It shall be noted that such connections have been rigorously justified,
for instance in \cite{CO20},
under the assumption that 
$\kappa \in L^1(\mathbbm{R}_+ \times \mathbbm{R}_+)$.
If $\beta$ is a power-law distribution with exponent $\phi$,
this assumption means for strategy~\((U)\) that
	\begin{equation}
		1+ b -\phi < -1 \quad \quad \text{ and } \quad \quad  a -\phi < -1,
        \label{ConvergenceProperty_U}
	\end{equation}
while it means for strategy~\((P)\) that
\begin{equation}
	2+ b -\phi < -1 \quad \quad \text{ and } \quad \quad  1+a -\phi < -1.
    \label{ConvergenceProperty_P}
\end{equation}

There are actually two types of branching processes
that we introduce in the next two Sections~\ref{sec_fwd_br} and \ref{sec_bwd_br}:
One follows the epidemic forward in time 
 and  describes the initial spread of the infection. Outbreak probabilities are associated to this process.
Backward in time,
we follow potential routes according to which a city might get infected. 
The infection probability is derived in terms of the  corresponding branching process approximation.
We start with the forward in time process 
whose derivation
is more straightforward.

 \subsubsection{Forward in time branching process}
\label{sec_fwd_br}

The asymptotic property of the secondary infections from a given city $i$ 
as $N'$ tends to infinity
is expressed in terms 
of the following two functions $K_{O, \rA}$, $K_{I,\rA}$ and probability measures $\nu_{O, \rA}, \nu_{I, \rA}$.
To simplify notations, we abbreviate
$\cZ_\gamma := \int y^\gamma \beta(dy)$, for $\gamma >0$.
\begin{align}
	&K_{O, \rA}(x) 
	:= k_B \,  L(x) \, x^b  \cZ_a, 
	&&	 \nu_{O, \rA}(dy) 
	:= \dfrac{y^a \beta(dy)}{\cZ_a}, 
	&\notag\\
	&K_{I, \rA}(x) 
	:= k_B \,  \frac{L(x)}{x} \cdot x^{a}  \cZ_{1+b},
	&&\nu_{I, \rA}(dy) 
	:= \dfrac{y^{1+b} \beta(dy)}{\cZ_{1+b}}.&
	\label{Knu}
\end{align}
The limiting behavior can be understood as follows.
The number of cities that get \emph{infected from the outside} by citizens of city $i$ is Poisson distributed with mean $K_{O, \rA}(x_i)$.   
 In particular, this distribution depends only on $x_i$. 
The size distribution of any city 
infected from the outside
is independently prescribed by the probability measure $\nu_{O, \rA}$,
which happens not to depend on $i$.
Because of this independence property, 
we may refer to a city \emph{typically infected from the outside}
when we consider the randomness of sampling the city size according to $\nu_{O, \rA}$.

Symmetrically,  the number of cities that get \emph{infected from the inside} by citizens of city $i$ 
is Poisson distributed with mean
$K_{I, \rA}(x_i)$ and 
the size distribution
of any city infected from the inside
is independently prescribed by the probability measure $\nu_{I, \rA}$,
which happens also not to depend on $i$.
Because of this independence property, 
we may refer to a city \emph{typically infected from the inside}
when we consider the randomness of sampling the city size according to $\nu_{I, \rA}$.

We can iteratively proceed this derivation and arrive at a discrete-time branching process {\bf $\text{V}^\infty$} associated to the infection process.
  Even though the branching process is not evolving on the actual graph of cities,
we can view it as evolving on a random tree. We will call the nodes of the branching process infected cities, as we did for the infection process.
The next generation of infected cities
is generated 
by means of the two functions $K_{O, \rA}$,
$K_{I, \rA}$ and of the two distributions 
$\nu_{O, \rA}$ and $\nu_{I, \rA}$
with new samplings independent of the previous generations.
Whether an approximation of the infection process by a branching process  is reasonable 
over a few generations, we evaluate by means of simulations in Section \ref{sec_simulated_infection_and_outbreak_prob} and in Section \ref{sec_val_R0}. In particular, we show simulation results that assess the impact of isolating (large) cities during the epidemic.

We are primarily interested in the number and sizes of the infected cities.
However, distinguishing cities that get infected from the outside vs those that get infected from the inside
will be simplifying the analysis. With this respect the following proposition will prove useful later. It is a direct consequence of the above conclusions.
\begin{prop}
	\label{V2}
The projection {\bf $\text{V}^{(2)}$}
on the two-type space $\{O, I\}$
of the process {\bf $\text{V}^\infty$} 
yields a two-type branching process with offspring generated as follows\footnote{
This is a consequence of the fact that the associated transfer operator projects on a two-dimensional subspace, as we discuss in Section~\ref{sec_Spec_An}.
}. Assume in generation $g-1$ that  $n^{g-1}_{O}$ cities are infected from the outside and that $n^{g-1}_{I}$ cities are infected from the inside. Then first city sizes $(x^{O}_i)_{i \in n_{O}^{g-1}}$ and $(x^I_j)_{j\in n_I^{g-1}}$ are drawn independently according to the measure $\nu_{O, \rA}$ and $\nu_{I,\rA},$ resp., for these cities. In the next generation $g$, the number $n_{O}^g$ of cities infected from the outside is  Poisson distributed with parameter $\sum_{i\in n^{g-1}_O} K_{O, \rA}(x^O_i) + \sum_{j\in n^{g-1}_I} K_{O, \rA}(x^I_j)$ and the number $n^g_I$ of cities infected from the inside is Poisson distributed with parameter  $\sum_{i\in n^{g-1}_O} K_{I, \rA}(x^O_i) + \sum_{j\in n^{g-1}_I} K_{I, \rA}(x^I_j)$.

We consider two different kinds of initial conditions. Either only the  numbers $n_O^0$ and $n_I^0$ of cities infected from the outside and from the inside are specified or in addition to the numbers $n_O^0$, $n_I^0$ also the city sizes $(x^{O,0}_i)_{i \in n_O^{0}}$ and $(x^{I,0}_j)_{j\in n_I^{0}}$ are given in generation 0. In the latter case, city sizes are not resampled in generation 1. The number $n_O^1$ of cities infected from the outside is  Poisson distributed with parameter $\sum_{i\in n_O^0} K_{O, \rA}(x^{O,0}_i) + \sum_{j\in n^0_I} K_{O, \rA}(x^{I, 0}_j)$ and the number $n^1_I$ of cities infected from the inside is Poisson distributed with parameter $\sum_{i\in n^0_O} K_{I,\rA}(x^{O,0}_i) + \sum_{j\in n^0_I} K_{I, \rA}(x^{I,0}_j)$.  
 \end{prop}
 
 \noindent
The fact that we have this projection to a two-state branching process
ori\-ginates from the separable form of the matrix kernel $\kappa_M$
in \eqref{pK_ij}.

\subsubsection{Backward in time branching process}
\label{sec_bwd_br}
Similarly as in the previous section,  we can follow backwards in time potential infection routes along which a city can get infected.
To infer (approximately) the infection probability for a city $j$ of size $x_j$ 
we approximate the backward infection chains also by a branching process.
Let
\begin{align}
\begin{split}
	K_{O, \lA}(x)
&:= k_B\, x^a
	\int_0^\infty L(z) \cdot z^b \beta(dz), 
\\ \nu_{O, \lA}(dy) 
	&:= \dfrac{L(y) \cdot y^b \beta(dy) }{\int_0^\infty L(z) \cdot z^b \beta(dz)},\\
	K_{I, \lA}(x) 
	&:= k_B\, x^{1+b}
	\int_0^\infty L(z) \cdot z^{a-1} \beta(dz),\hcm{0.3}
\\	\nu_{I, \lA}(dy) 
	&:= \dfrac{L(y) \cdot y^{a-1} \beta(dy) }
	{\int_0^\infty L(z) \cdot z^{a-1} \beta(dz) }.
 \end{split}
	\label{KCD}
\end{align}

In the limit,
the  number of infectors of $j$ from the outside (resp. from the inside) 
is Poisson distributed with mean
$K_{O, \lA}(x_j)$ (resp. $K_{I, \lA}(x_j)$). 
The sizes of the infectors of $j$ from the outside
(resp. from the inside)
are independently prescribed by the probability $\nu_{O, \rA}$ (resp. $\nu_I$),
which happens not to depend on $j$.

We denote by ${\bf \Lambda}^{\infty}$
the corresponding backward branching process. 
In this discrete-time multitype branching process, 
an element with type $x \geq 0 $ 
is in the next generation replaced by elements drawn according to the above procedure by replacing $K_{O, \lA}(x_j)$ (resp. $K_{I, \lA}(x_j)$) by $K_{O, \lA}(x)$ (resp. $K_{I, \lA}(x)$).
In analogy to Proposition \ref{V2}
we have the following result.

\begin{prop}
	\label{L2}
	The projection {\bf $\Lambda^{(2)}$} 
	on the two-type space $\{O, I\}$
	of the process {\bf $\Lambda^\infty$} 
	started at the first generation
	generates a two-type branching process.
Provided {\bf $\Lambda^\infty$} starts with city $i$ of size $x_i$,
the initial condition of {\bf $\Lambda^{(2)}$} 
	is given by independent Poisson 
	random numbers of cities infectors from the outside and from the inside respectively,
	with averages $K_{O, \lA}(x_i)$ and $K_{I, \lA}(x_i)$ respectively. 
\end{prop}
\noindent
The separability of $\kappa_M$ 
is again a crucial ingredient for this reduction.

\section{Analysis of the branching approximation}
\label{sec_analysis_br}
\subsection{Probability of generating an outbreak}\label{sec:probability trigger outbreak}

In this section, we motivate our
approximate formula for the probability that from a city of size $x$ an outbreak is generated. As we approximated the infection process by the branching process {\bf $\text{V}^\infty$}, it is natural to approximate the outbreak probability by  the probability of survival of {\bf $\text{V}^\infty$} started with a single infected city of size $x$ (see e.g.  Theorem 3.11 in \cite{CO20} for a rigorous convergence result in this direction). Let us denote this probability by $\eta(x)$.
According to Proposition \ref{V2}, $\eta(x)$ coincides with the probability of survival 
of the two-type branching process
{\bf $\text{V}^2$}.
Thanks to Lemma 5.4 in \cite{BJR06},
$\eta$ is the solution 
to  the following equation.
\begin{equation}\label{survivalprobab}
\eta(x) 
	= 1- \exp\Big[  - K_{I, \rA}(x)  \int_0^\infty\eta(y) \nu_{I, \rA}(dy)
	-  K_{O, \rA}(x)\int_0^\infty\eta(y) \nu_{O, \rA}(dy) \Big].
\end{equation}
Formulas \eqref{etaP} and  \eqref{etaU}
are derived from  \eqref{survivalprobab} by inserting in \eqref{Knu} the respective formulas for $L(x)$ for each of the strategies.
Accurate estimations of $(\eta(x))_x$ can be efficiently achieved using an iterative procedure based on these formulas, see also Section~\ref{sec_KB_errors}  of the Supplementary Material.

We recall from Section~\ref{sec_main}
that the numerically calculated values of $(\eta(x))_x$ appear to be close to the outbreak probabilities for simulated epidemics under strategy~\((P)\) for France, see Figure~\ref{Fig_Compare_Analytical_PO_Empirical_PO}. For  Japan and for Poland and under strategy~\((U)\), theoretical outbreak probabilities fit less well, see  also Figures~\ref{fig_comp_outbreak_prob_pO}
- \ref{Fig:Compare PI v9} in the Supplemental Material.


\subsection{Backward in time process and infection probability}\label{sec:infection probability}

In this section, we give an approximation of  the probability that a city (of size $x$) eventually gets infected during an outbreak. Similarly as for the outbreak probability,  we  approximate the probability of a city of size $x$ eventually to be infected by the probability that the associated backward branching process introduced in Section \ref{sec_bwd_br} with initial state $\delta_x$ survives.

Let us denote by $\pi(x)$
the probability of survival of the backward process ${\bf \Lambda}^{\infty}$
starting from a city of size $x$.
According to Proposition \ref{L2} it coincides with the survival probability
of the two-type branching process
{\bf $\Lambda^2$}.
Thanks to Lemma 5.4 in \cite{BJR06}
and similarly as for $\eta$,
$\pi$ is the solution the following equation.
\begin{equation}
	\pi(x) 
	= 1- \exp\Big[  - K_{O, \lA}(x) \int_0^\infty  \pi(y) \nu_{O, \lA}(dy)
	-  K_{I, \lA}(x) \int_0^\infty \pi(y) \nu_{I, \lA}(dy)\Big].
	\label{piDef}
\end{equation}
Similarly as the formulas \eqref{etaP} and  \eqref{etaU}, the formulas  \eqref{piP} and  \eqref{piU}
are derived from  \eqref{piDef} by recalling \eqref{KCD}.
Accurate estimations of $(\pi(x))_x$ can be efficiently achieved using the iterative procedure based on these formulas, as detailed in Section~\ref{sec_KB_errors} of the Supplementary Material.

For large graphs under the condition 
that $\kappa \in L^1(\mathbbm{R}_+ \times \mathbbm{R}_+)$
(recall \eqref{ConvergenceProperty_U}, \eqref{ConvergenceProperty_P}),
and conditionally on an outbreak, 
the relative size of the forward connected component is 
concentrated around the quantity $\int \pi(x) \beta(dx)$  with a probability close to one whatever the initially infected city, see Theorem~3.11 in \cite{CO20} for more details.

We recall from Section~\ref{sec_main}
that the numerically calculated values of $(\pi(x))_x$ appear to be close to the infection probabilities for simulated epidemics under strategies~\((P)\) and \((U)\) for France and Poland, see Figure~\ref{Fig_Compare_Analytical_PI_Empirical_PI}. For  Japan, theoretical infection probabilities fit less well, see  also Figures~\ref{fig_comp_infection_prob_pO}
and \ref{Fig:Compare PI v4} - \ref{Fig:Compare PI v9} in the Supplemental Material.

\subsection{Comparison of the two strategies}
\label{sec_comp}
In this section, we  explicate how we compare the efficiency of the two strategies~\((P)\) and~\((U)\). 
We recall that 
 a city gets isolated under strategy~\((P)\)
 when a certain proportion $p_\vee$ of inhabitants of the city gets infected while strategy~\((U)\) consists in prescribing a threshold $L_\vee$
for isolation that is independent of the city size. To arrive at a comparison of the efficiency of the two strategies, we assume that after isolation of a city no further inhabitants get infected. In particular, at the end of an epidemic in an isolated city of size $x$ under strategy~\((P)\) (under strategy~\((U)\), resp.) there are $p_\vee x$ ($L_\vee$, resp.) individuals that have been infected.  Furthermore, we approximate the infection probabilities by the probabilities introduced in Section \ref{sec:infection probability}.  
Depending on the strategies,
these probabilities are denoted by  $\{\pi_U(x)\}_{x>0}$ and  $\{\pi_P(x)\}_{x>0}$, resp., recalling \eqref{piU} and \eqref{piP}.

We measure the burden of the epidemic by the number $I_\star$ of eventually infected people.
Under our assumption, $I_P$ (resp. $I_U$)
as defined in \eqref{defI_P} (resp.  \eqref{defI_U})
gives the number of eventually infected individuals under strategy~\((P)\) 
(resp. strategy~\((U)\) ),
while  $Q_P$ (resp. $Q_U$)
as defined in \eqref{QDef} 
give the number of of individuals eventually under isolation under strategy~\((P)\) 
(resp. strategy~\((U)\) ).
Assume that we adjust the parameters  $L_\vee$ and $p_\vee$ 
such that $I_\star$ is the same under both strategies.
Then, we regard the strategy for which less people need to be isolated as the more \textit{efficient} \textit{strategy}.

There is a large range of parameters $(a, b)$
for which it is possible to determine 
regardless of the severity of the epidemic
whether one strategy has an advantage over the other, or whether the difference between them is small,
as we could see in Section~\ref{sec_main},
particularly with Figure~\ref{fig_Numerical_Comparison_UP}.
This range of parameters is however affected 
by the city size distribution.
In particular, 
if $b$ departs from $0$ but is not far from $1-a$,
the more efficient strategy can be different
depending on the value of $I_\star$
(or $Q_\star$ if we revert the roles).
This dependency appears to be stronger 
for powerlaw distributions with heavier tails,
i.e. when the exponent $\phi$ is smaller.
\smallskip


In the specific case where $a= 1+b$, 
we can rigorously identify
which strategy outperforms the other
because the kernel of the (two-type)-branching process $\Lambda^{(2)}$ is actually of rank-one, which clarifies the analysis.
The next proposition states which strategy is more efficient depending simply on the value of~$a$.

The probabilities $\pi_U$ and $\pi_P$ depend on $k_B, a$, $\beta$, 
and resp. $L_\vee$ and $p_\vee$.
To compare the two strategies, we fix  $k_B, a, \beta$
and a value $I_\star$ for the number of people that eventually get infected.
We adjust the parameters  $L_\vee$ and $p_\vee$ 
such that $I_U = I_P = I_\star$.
\begin{prop}
	\label{ineq}
	Assume $\beta$ is a non-Dirac probability measure,
	$a$, $k_B$ and $I_\star$ are given and assume that $a = b+1$.
	
	i) If $a = 1$, then $Q_U = Q_P$, i.e. both strategies are equally efficient.
	
	ii) If $a<1$, then $Q_U > Q_P$, i.e. strategy~\((P)\) is more efficient.
	
	iii) If $a>1$, then $Q_U < Q_P$, i.e. strategy~\((U)\) is more efficient.
\end{prop}
As one can see in Section~\ref{sec_ineq}
of the Supplementary Material,
Proposition~\ref{ineq}
is a consequence of  Hölder's inequality.

\subsection{Basic reproduction number}\label{Section Reproduction Number} 

In this section, we define a basic reproduction number $R_0$ for the branching process ${\bf \text{V}^{\infty}}$
(that approximates the transmission process between cities).
We set
\begin{equation}
	R_0 := \lim_{k \rightarrow \infty}\bE_x[X_k]^{1/k},
	\label{RoDef}
\end{equation}
where $X_k$ is the number of infected cities at the $k$-th generation.

We will deduce from Proposition~\ref{Red_spec} below
that this definition actually does not depend on $x$
and has an explicit form for this rank-2 kernel.
We refer to Subsection \ref{sec_val_R0}
for the numerical study of this quantity 
for our datasets.

Proposition~\ref{Red_spec} exploits the following definition of 
the two-by-two matrix $\Wmat$:
\begin{align}\label{MatrixTo}
	\Wmat = 
	\begin{pmatrix}
		\langle \nu_{I,\rA} \bv K_{I, \rA}\rangle & \langle \nu_{I, \rA} \bv K_{O, \rA}\rangle \\
		\langle \nu_{O,\rA} \bv K_{I, \rA}\rangle  & \langle \nu_{O, \rA} \bv K_{O, \rA}\rangle 
	\end{pmatrix}.
\end{align}
The entry $\langle \nu_{I, \rA} \bv K_{I, \rA}\rangle$ can be interpreted 
as the average number of cities that get infected from the outside
from a typical city that got infected from the outside.
A similar interpretation holds for the other entries.
$\Wmat$ is classically related to the long-time behavior of $\text{V}^2$,
cf \cite{KS66}.

\begin{prop}
    \label{Red_spec}
$R_0$ as defined in \eqref{RoDef}
coincides with the dominant eigenvalue of the matrix $\Wmat$.
\end{prop}
In the Supplementary Material, 
we prove Proposition \ref{Red_spec}
as a direct consequence of Proposition~\ref{Spec},
see Section~\ref{sec_Spec_An} on the complete spectral analysis 
of the transfer operator related to $\text{V}^\infty$.
$R_0$ as an eigenvalue of this transfer operator
has corresponding eigenvectors 
that are functions of the city size,
as described in Subsection \ref{sec_eig_cent}
dedicated to eigenvector centrality.
It is a measure for the number of infections that are triggered when city $i$ gets infected and hence, allows to compare the relative importance of the different cities during the course of an epidemic.
By targeting strict measures on cities with a high eigenvector centrality value, one wishes to target cities with a high potential for secondary infections.

\subsection{Computation of the basic reproduction number}
\label{sec_adj_str}

This subsection is dedicated to the computation of $R_0$,
under the two strategies \((P)\) and~\((U)\).
Recall that $L(x) \equiv L_\vee$ 
under strategy \((U)\) 
and  $L(x) =p_\vee\cdot x$
under strategy \((P)\).
A crucial role is played by the moments $\cZ_\gamma$
with exponent $\gamma>0$, cf \eqref{Knu} and just above.

Under strategy~\((U)\),
the entries of the two-type transmission matrix $\Wmat$ defined in \eqref{MatrixTo}, here abbreviated as $\Wmat^U$,
are the following:
\begin{align*}
\Wmat^U_{O, O} 
&:=\int_0^\infty  \nu_{O, \rA}(dx)  k_B \, L_\vee\, x^b\, \int_0^\infty y^a \beta(dy)
\\& = k_B \, L_\vee\, \cZ_{a+b},\quad
\\
\Wmat^U_{O, I} 
&:=
k_B \, L_\vee\cdot \dfrac{\cZ_{2a-1}\cdot \cZ_{1+b}}
{\cZ_{a}},\quad
\\\Wmat^U_{I, O} 
&:=
k_B \, L_\vee\cdot \dfrac{\cZ_{a}\cdot \cZ_{1+2 b}}
{\cZ_{1+b}},\quad
\\ \Wmat^U_{I, I} 
&:= k_B \, L_\vee\, \cZ_{a+b}.
\end{align*}
The largest eigenvalue of such a two-by-two matrix 
has an analytic expression,
so that we deduce under strategy~\((U)\):
$$R^{(U)}_0 = k_B \, L_\vee\, r^U_0,$$ where
\begin{equation}\label{iniquity sim wickedness}
r^U_0 =  \cZ_{a+b} + \sqrt{\cZ_{2a-1} \cZ_{2b+1} }.
\end{equation}

Under strategy~\((P)\),
 the entries of the corresponding two-type transmission matrix $\Wmat^P$
are 
\begin{align*}
\Wmat^P_{O, O} 
= \Wmat^P_{I, I} 
&:= k_B \, p_\vee\, \cZ_{1+a+b},
\\\Wmat^P_{O, I} 
&:=
k_B \, p_\vee\cdot \dfrac{\cZ_{2a}\cdot \cZ_{1+b}}
{\cZ_{a}},\quad
\\ \Wmat^P_{I, O} 
&:=
k_B \, p_\vee\cdot \dfrac{\cZ_{a}\cdot \cZ_{2+2 b}}
{\cZ_{1+b}}.
\end{align*}
We deduce the following formula for the largest eigenvalue  under strategy~\((P)\):
$$R^{(P)}_0 = k_B \, p_\vee\, r^P_0,$$
where
\begin{equation}\label{dry run sim rehearsal}
r^P_0 
=  \cZ_{1+a+b} + \sqrt{\cZ_{2a} \cZ_{2+2b}}.
\end{equation}

\subsubsection*{Elementary adjustment of the two strategies}

A possibility to adjust the two strategies \((U)\) and \((P)\) 
to each other
is to put $R^{(U)}_0$ and $R^{(P)}_0$ 
at the same level $R_\star$
and compare the corresponding threshold values $L_\vee$ and $p_\vee$. 
This is particularly useful, if one knows parameter regimes for which  one of the two strategies leads to a subcritical epidemic and one wants to choose the threshold of the other strategy such that the corresponding epidemic is also subcritical.
In the following proposition,
we give conditions on the parameters $a$ and $b$
that enable to upper-bound $p_\vee$ depending on $L_\vee$.
It is proven in Section~\ref{sec_pr_adj_R0}
in the Supplementary Material.

\begin{prop}
	\label{p_Comp3}
Let us define $\gamma := (2 a)\vee (2+2b)$
and $\delta \le (2 a-1) \wedge (2 b+1)$. 
For any $\beta$ such that the moment $\cZ_\gamma$ is finite, it holds that
	$\cZ_{\delta+1}\cdot  r^U_0 \le \cZ_\delta\cdot r^P_0$.
	Consequently,
	when $R^{(U)}_0$ and $R^{(P)}_0$ are put
	at the same level $R_\star$,
it implies that  $p_\vee \le L_\vee \cdot \cZ_\delta / \cZ_{\delta +1}$.
	
In particular,
	it implies that 
 \begin{align}\label{Upper bound p_v 1}
 p_\vee \le L_\vee / \cZ_1
 \end{align}
	provided $a\ge 1/2$ and $b\ge -1/2$,
	and 
 \begin{align}\label{Upper bound p_v 2}
 p_\vee \le L_\vee \cdot \cZ_1 / \cZ_2
 \end{align}
	provided  $a\ge 1$ and $b\ge 0$.

 In the case $a= 1+b$ we have 
\begin{align}\label{AdjustmendA=1+B}
p_\vee = L_\vee \cdot Z_{1+2b}/Z_{2+2b}.
\end{align}
\end{prop}

The first inequality \eqref{Upper bound p_v 1}
 means that the threshold for strategy~\((P)\) must be smaller than $L_\vee$ divided by the expected city size.
For the second inequality \eqref{Upper bound p_v 2} note that  
\[
p_\vee \frac{\mathcal{Z}_2}{\mathcal{Z}_1} = 
\int \frac{x}{\int y \beta(dy)} x p_\vee \beta(dx), 
\]
which can be interpreted as the average threshold number of infected people in the city of a randomly chosen individual. So Ineq. \eqref{Upper bound p_v 2} means that $p_\vee$ must be chosen  such that this average threshold number
is smaller than $L_\vee$.

In practice,
$\beta$ is often heavy-tailed.
In this case, the two upper bounds can be very far apart.
For the data from France 
(see Subsection \ref{Sec:MobilityData} for the description of the dataset),
the ratio $\mathcal{Z}_2 / (\mathcal{Z}_1)^2$
is actually close to 36,
while for the Japanese data 
the ratio $\mathcal{Z}_2 / (\mathcal{Z}_1)^2$
is even larger than 80.
Since  $a\approx 1$ and $b\approx 0$ in both cases,
only the second estimate of $p_V$
by $L_\vee \cdot \mathcal{Z}_2 / \mathcal{Z}_1$ is thus relevant.

For heavy-tailed distributions $\beta$, the right-hand side of \eqref{AdjustmendA=1+B}
is very much affected by the value of $b$.

\section{Data analysis and simulations}
\label{sec_data_analysis}

\subsection{Datasets}\label{Sec:MobilityData}

We exploited mobility data from France, Poland and Japan for our analysis.
The general data processing for the simulation consisted of: collecting population and travel flow data from original data sources and reaggregating data into regional units, as specified for the different countries in the following paragraphs.  
\medskip

The French data was sourced from the main statistical office in France - INSEE. 
Popu\-lations \cite{sourceFrance2} 
and travel flows
\cite{sourceFrance} 
were extracted on a municipal level and reaggregated in regions which represented the \textit{areas of attraction} defined by INSEE
\cite{frenchAoA}. 
The travel flow dataset represents workforce mobility (fr. \textit{mobilit\'es professionnelles}) 
from the area of living, to the area of work.

The Polish data was sourced from the Polish Statistical Office (rebranded recently as Statistics Poland \cite{PLMobility}), more precisely from the 2016 dataset on workforce mobility based on a census from 2011. Both population and travel flows are provided per commune level (pl. \textit{gmina}). In the case of Poland however, there was no established area of attraction division that would relate to the last fully published census of 2011. Instead, we reaggregated the data to the \textit{powiat} level. In two cases, we reaggregated further a couple of powiats into two regional units: the GZM, which is an interconnected region of around 5 mln inhabitants in Upper Silesia, and the Tricity, which is a region of three densely connected cities in the central northern Poland.

In Japan, we used the data from the inter-regional mobility study of 2015 \cite{sourceJapan,sourceJapan2}. It is the largest census of public mobility in Japan, collected from railway companies - it encapsulates the complete mobility of a country on - among others - a working weekday. This is an important difference from how the previous two data sets were created, as this data set does not only include workforce mobility, but all railway mobility on a weekday - it is obviously heavily dominated by workforce mobility, but does include all other trip reasons. As there are as for Poland no predesigned areas of attraction for Japan, we first divided Japan in regional units according to municipalities. 
Based on the division created by OECD \cite{oecdFUA},
we reaggregated the listed municipalities 
into 61 larger regional units, that correspond to the 61 Functional Urban Areas of Japan.
Some of the municipalities among the remaining ones were additionally removed,
because they were either not listed 
as origin or as destination
in the mobility census (after the filtering by distance).
The remaining municipalities were left as singular regional units.
\smallskip

To effectively model an epidemic wave on a nationwide scale, we filter traveler counts by excluding transits between municipalities located less than 30 km apart.
As highlighted in Section \ref{sec_main}, this threshold is supported by common mobility behavior observed across Europe, see e.g. \cite{NadalEtAl2022} and \cite{CommuterDistance}.
To evaluate the effect of this filtering, we performed our data analysis also for non-filtered data as well as for data with a filtering that takes into account travels  of distances of at least 50 km. In the following, we will abbreviate the corresponding datasets by D1+, D30+ and D50+. The results of the D1+ and D50+ analysis can be found in the Supplemental Material, because our model generally fits worse to this data (see Figures~\ref{Fig:Compare PI v4}-\ref{Fig:Compare PI v9}).
\smallskip

Following the just prescribed procedure, we obtained empirical city size distributions of size 668 for France, 341 for Poland and 777 for Japan and for each country a mobility matrix that gives the (directed) travel counts between each pair of cities.

\subsection{Estimating  the city size distributions}
\label{EstimatingPhi}
To describe the tail of the city size
distributions in France, Poland and Japan, as presented in Figure~\ref{fig_Heavy_Tailed_Dist},
we used 
the Python package “powerlaw”
with a fit according to the Kuiper distance \cite{ABP14}.

We found that the (tail of the) size distribution of French areas of attraction is well
approximated by a log-normal distribution,   while a power-law distribution with $\phi = 2.81$
(respectively $\phi = 1.85$),
recalling Formula~\eqref{beta}, 
can be fitted well to the size distribution of Polish (respectively
Japanese) areas of attraction. 
The fit of French areas of attraction
by a power-law distribution leads to an estimate of $\phi = 1.82$.

\subsection{Estimating cities attractiveness and inhabitants mobility}
\label{EstimatingAandB}

We used the empirical number of inbound and outbound travels of work-related mobility in France, Japan and Poland to estimate the values of $a$ and $b$ in the different countries,
see Figure~\ref{fig_Estimate_A_and_B}.

According to our model, we assume that the probability $p_i(x)$ that a city of size $x$ is chosen as a travel target is proportional to $x^a$. Based on the travel counts to a city, we empirically estimate this probability by the ratio of the "number of travels to the city" to the "total number of travels". In case our model is good, we consequently assume that $$c x^a \approx \tfrac{\text{\# travels to the city}}{\text{total number of travels}}$$ and hence  $$\log(\text{number of travels to the city}) \approx a \cdot \log(x)  + c'$$ for some appropriate constant $c'$. To arrive at an estimate of $a$, we fit a linear regression with least squares to the data points $$(\log(x_i), \log(\text{number of travels to city $i$}))_i,$$ where $x_i$ denotes the size of  city $i$ and set $\hat{a}$ equal to the slope of this regression line.

 We estimate $b$ similarly. According to our model, we assume that the probability that an individual of a city of size $x$ is traveling is proportional to $x^b$. Hence, 
  we arrive at an average of $x \cdot c x^b$ individuals that
will leave a city of size $x$ for travel for some $c>0$.  Based on the outbound travel counts, we estimate this number by the number of travels from the city of size $x$. 
In particular, we assume $$\log(\text{number of outbound travels from the city}) \approx (1+b) \log(x) + c'$$ for some appropriate constant $c'$. 
Consequently,  we fit a linear regression with least squares to the data points
$$(\log(x_i), \log(\text{number of outbound travels from city $i$}))_{i}$$ to arrive at an estimate of $b$, and set $\hat{b} + 1$ equal to the slope of this regression line.

Using the  census data filtered according to D30+ as described in Section \ref{Sec:MobilityData}, we obtain the estimates given in Table~\ref{tab_AB_values},
as presented in Figure~\ref{fig_Estimate_A_and_B}. 
Recall from \eqref{ConvergenceProperty_U} and \eqref{ConvergenceProperty_P}
that the condition $\kappa\in L^1$
corresponds to both $\hat \phi - \hat a-1> 0$ and $\hat \phi - \hat b- 2
> 0$ under strategy~\((U)\)
and to both $\hat \phi - \hat a-1 > 1$ 
and $\hat \phi - \hat b- 2
> 1$ under strategy~\((P)\).
 We remark that notably the condition $\hat \phi - \hat a-1> 0$
is not fulfilled, 
be it for France, Poland and Japan.
  
  These estimates of $\hat a$ and $\hat b$ (as well as the estimates that are obtained when filtering the mobility data according to D1+ and D50+)  are also added in Figure~\ref{fig_Numerical_Comparison_UP}, where the efficiency of strategy~\((U)\) and \((P)\) are compared for a range of $a$ and $b$ values. 

\begin{table}[t]
    \centering
    \begin{tabular}{|c||c|c|c|c|}
 \hline
     Country   & $\hat a$ & $\hat \phi - \hat a-1$& $\hat b$& $\hat \phi - \hat b-2$\\
        \hline
        \hline
  France       &  0.96 &\cellcolor[gray]{0.85} - 0.14 & -0.09&\cellcolor[gray]{0.85} -0.09
  \\
  Poland       &  1.95 &\cellcolor[gray]{0.85} - 0.14
  &-0.11 & \cellcolor[gray]{0.95} +0.92\\
  Japan       & 1.05& \cellcolor[gray]{0.85}-0.2&0.05& \cellcolor[gray]{0.85}-0.2
  \\
  \hline
    \end{tabular}
  \caption{Estimates of $\hat a$ and $\hat b$ and comparison with the estimated powerlaw coefficient, for France, Poland and Japan. The level of gray is shown in 
accordance to the conditions leading to $\kappa \in L^1$. A dark grey color indicates that the conditions are not fulfilled, while for light grey colors,  $\kappa \in L^1$ under strategy~\((U)\) but not under strategy~\((P)\).}
    \label{tab_AB_values}
\end{table}

\subsection{Infection probabilities and probabilities to trigger an outbreak}\label{sec_simulated_infection_and_outbreak_prob}

In Sections \ref{sec:probability trigger outbreak} and \ref{sec:infection probability}, we derived analytical (implicit) formulas for the infection and outbreak probabilities that are based on branching process approximations. We showed in Section~\ref{sec_main} already some figures,  see Figures~\ref{Fig_incidence_ppl}-\ref{Fig_Compare_Analytical_PO_Empirical_PO}, that compare infection and outbreak probabilities based on the transportation graph with the  analytically derived infection and outbreak probabilities.
In this section, we provide a more thorough comparison of these probabilities.

For the transportation graph simulations,
we recall that edge probabilities $(p^{(m)}_{ij})_{i, j\le N}$ are determined by the empirical mobility matrices (based on commuters data) given by 
\begin{equation}
-N\cdot \ln\big(1- p^{(m)}_{ij}\big) :=  
 k_C\, \frac{L(x_i)}{x_i}  \cdot \left(m_{ij}  + m_{ji}\right),
 \label{pT_ij}
\end{equation}
where  $L(x)= L_{\vee}$ under strategy~\((U)\) and $L(x) = p_{\vee} x$ under strategy~\((P)\),
while $m_{ij}$  denotes the number of travels from city $i$ to city $j$,
see Section \ref{Sec:MobilityData}
for the procedure leading to this mobility matrix for each country.
In the case of the kernel graph, 
the value of $m_{ij}$ is replaced by $\hat k_M\, x_i^{1+\hat b}\cdot x_j^{\hat a}$, as stated in \eqref{pK_ij}
with the best fit estimates $\hat a$ and $\hat b$
of $a$ and $b$ for each country as described in Section~\ref{EstimatingAandB}.
More precisely, we  adjusted the two graphs 
by simply calibrating the parameter $k_M$ given $\hat a$ and $\hat b$
such that the total amount of mobility agrees  between the two models, i.e.
$\hat k_M = \sum_{i, j} m_{i, j}/\sum_{i, j} x_i^{\hat a}\cdot x_j^{1+\hat b}$.
The KB infection probability 
is evaluated as described in Section~\ref{sec:infection probability}
with the empirical size distribution for $\beta$.
We adjusted the free parameter $k_B$, such that the KB infection probability is equal to 0.5 for cities of size $10^5$ for France, Poland and Japan.
\footnote{Additional figures are presented in the Supplementary Material 
where the fit to 0.5 for cities of size $10^5$
is set to  the outbreak probability
instead of the infection probability,
to investigate further the robustness of our results,
see Figures~\ref{fig_comp_infection_prob_pO} and \ref{fig_comp_outbreak_prob_pO}.}

After the simulation of infection and outbreak probabilities for any city of the (finite) graph, we collect cities of similar size into 20 bins (such that every bin contains about 30 cities for France, 15 cities for Poland and 40 cities for Japan). For each bin, we calculate the average infection and outbreak probability. The plots concerning the infection probabilities can be found in Figure~\ref{Fig_Compare_Analytical_PI_Empirical_PI}, and the plots concerning the outbreak probabilities in 
Figure~\ref{Fig_Compare_Analytical_PO_Empirical_PO}. Additional figures (Figures~\ref{Fig_Final_Incidence}-\ref{Fig:Compare PI v9}) can be found in the Supplemental Material to check the robustness of these numerical outcomes.

To illustrate the variation of the infection and outbreak probabilities within a bin,  standard deviations are plotted, 
as well as minimal and maximal values of TG-based probabilities.
Furthermore, we calculated
coefficients of determination,
that are  denoted by $R^2$ in the following
and provided in Table~\ref{tab_R2_values}, to
evaluate how well the KB-based 
infection (resp. outbreak) probabilities
predict TG-simulated infection 
 (resp. outbreak) probabilities.
More precisely, we  use the following definition of $R^2$ 
that is expressed in terms 
of the $n$ data observations $(y_i)_{i\le n}$ 
(which are here the TG-simulated infection 
 (resp. outbreak) probabilities),
their average $\bar y$,
and the predictions $(\hat y_i)_{i\le n}$:
\begin{equation}
	R^2 
= 1- \frac{\sum_{i\le n}(y_i - \hat y_i)^2}
{\sum_{i\le n}(y_i - \bar y)^2}.
\label{eq_R2_def}
\end{equation}
It is well known, that if the prediction is given by a regression, 
this coefficient of determination 
is necessarily non-negative
and can be expressed 
as the fraction of the variance in the observation
that is explained by the regression.
In general, 
the prediction may be biased 
(on average over the $n$ observations), 
therefore the numerator in \eqref{eq_R2_def} 
can a priori express both a variance and a squared bias 
over the residuals $(y_i - \hat y_i)_{i\le n}$.
Especially when the simulated outbreak/infection probabilities vary little with the city size (so that the denominator is small),
the $R^2$ value in the comparison with the theoretical outbreak/infection probabilities
can be largely negative.
\begin{table}[t]
    \centering
    \begin{tabular}{|c||c|c|c|c|}
     \hline
    &\multicolumn{2}{|c|}{infection probability}
& \multicolumn{2}{|c|}{outbreak probability}\\
 \hline
     Country   & strategy~\((P)\)&strategy~\((U)\)& strategy~\((P)\)&strategy~\((U)\)\\
        \hline
        \hline
  France       &  0.98 &  0.80&0.98&\cellcolor[gray]{0.85}-28.81\\
  Poland       &  0.78 &
   \cellcolor[gray]{0.95} 0.04& 0.80&\cellcolor[gray]{0.85}-4.02 \\
  Japan       & \cellcolor[gray]{0.85}-5.17& \cellcolor[gray]{0.85}-111.60&\cellcolor[gray]{0.85}-4.28& \cellcolor[gray]{0.85}-370.84
  \\
  \hline
    \end{tabular}
    \caption{Values of the coefficient of determination
    for TG simulations as samples 
    and KB formulas as predictions
    for the infection  (left panel) and outbreak (right panel) probability
 under strategy~\((U)\) and \((P)\),
 for France, Poland and Japan.
 Values close to 1 (in white) indicate good performances of the KB prediction,
 values close to 0 (light-gray) indicate similar performances as simply averaging over TG  samples.
If the variance  with respect to the prediction is larger than the corresponding variance with respect to TG averages, the value is negative (in gray).
}
    \label{tab_R2_values}
\end{table}

The coefficients of determination given in Table~\ref{tab_R2_values} confirm the goodness of fit for France, particularly under strategy \((P)\), but also under strategy ~\((U)\),
and to a lesser extent for Poland under strategy~\((P)\) for the infection probability.
For Poland under strategy~\((U)\),
the TG infection probability values 
are not as regularly increasing as under strategy~\((P)\), with quite concentrated values.
That the fit by KG infection probability is as good as the mere average of TG infection probability values (in the sense of a $R^2$ value close to 0)
is a reasonable performance.
For Japan on the contrary,
 KG infection probabilities
 provide a much worse prediction of the TG values
 as compared to their average.

 Concerning the outbreak probabilities,
 we observe very similar $R^2$ values 
 under strategy~\((P)\)
 as for the infection probabilities.
 This outcome is foreseeable, as the infection and outbreak probabilities coincide under strategy~\((P)\).
Under strategy~\((U)\),
 $R^2$ values for outbreak probability  are largely negative. This
can be explained by the fact that the variance of the TG values is very small, 
resulting in a very small denominator in  \eqref{eq_R2_def} that amplifies  any prediction bias 
in the KG values.

Interestingly,  KB infection probability functions appear to be  extremely similar for strategy~\((U)\)
and strategy~\((P)\)
when adjusted at the reference size of $10^5$ inhabitants,
with the data from France, Poland and Japan.
This implies that, for  a prescribed  expected  number of infected individuals,
the distribution of the city sizes placed under isolation shall be almost the same 
under strategy~\((U)\)
and strategy~\((P)\),
and similarly for a given  expected number of people under isolation.
This property can actually be proved whenever $a = 1+b$
(corresponding to the kernel $\kappa$
in \eqref{pVU} being of rank-one). This observation is crucial for Proposition~\ref{ineq}, implying then a robust comparison of efficiency between the two strategies.
\medskip

To check whether the alignment of the TG and KB infection and outbreak probabilities is robust for varying epidemic strength,
there are two natural summary statistics,
namely the proportion of people (resp. cities) under isolation.
On the other hand, 
the basic reproduction number
provides a natural scale
to compare variations in the virulence levels (or in the threshold value)\footnote{Recalling \eqref{iniquity sim wickedness} and \eqref{dry run sim rehearsal},
we observe that both
$r^U_0 = R^{(U)}_0/(k_B\, L_\vee)$
and $r^P_0 = R^{(P)}_0/(k_B \, p_\vee)$ are directly expressed in terms of $a$, $b$ and the distribution $\beta$.
}.
Note that the proportion of infected people is linearly related to 
the proportion of people (resp. of cities) under isolation under strategy~\((P)\) (resp. under strategy~\((U)\)) when the value of $k_B$ is changed.

We already presented in Figure~\ref{Fig_incidence_ppl}
the mean proportion of people under isolation 
for a range of basic reproduction numbers, 
where
the analytical proportions are calculated according to the theoretical probabilities as defined in \eqref{piDef} and the corresponding in silico proportions are obtained from simulations of TG epidemics. In Figures~\ref{Fig_incidence_cit} (and \ref{Fig_incidence_cO_ppl}),
we display analogous plots 
for the mean proportion of \textit{cities} under isolation 
(respectively the mean proportion of people under isolation \textit{conditionally on a large outbreak},
for France and Poland).
In Figure~\ref{Fig_Final_Incidence}
of the Supplementary Material, 
we show each country on a specific panel 
to better compare the proportions of people and of cities under isolation.

\begin{figure}[t]
	\begin{center}   
    \includegraphics[width = \textwidth]{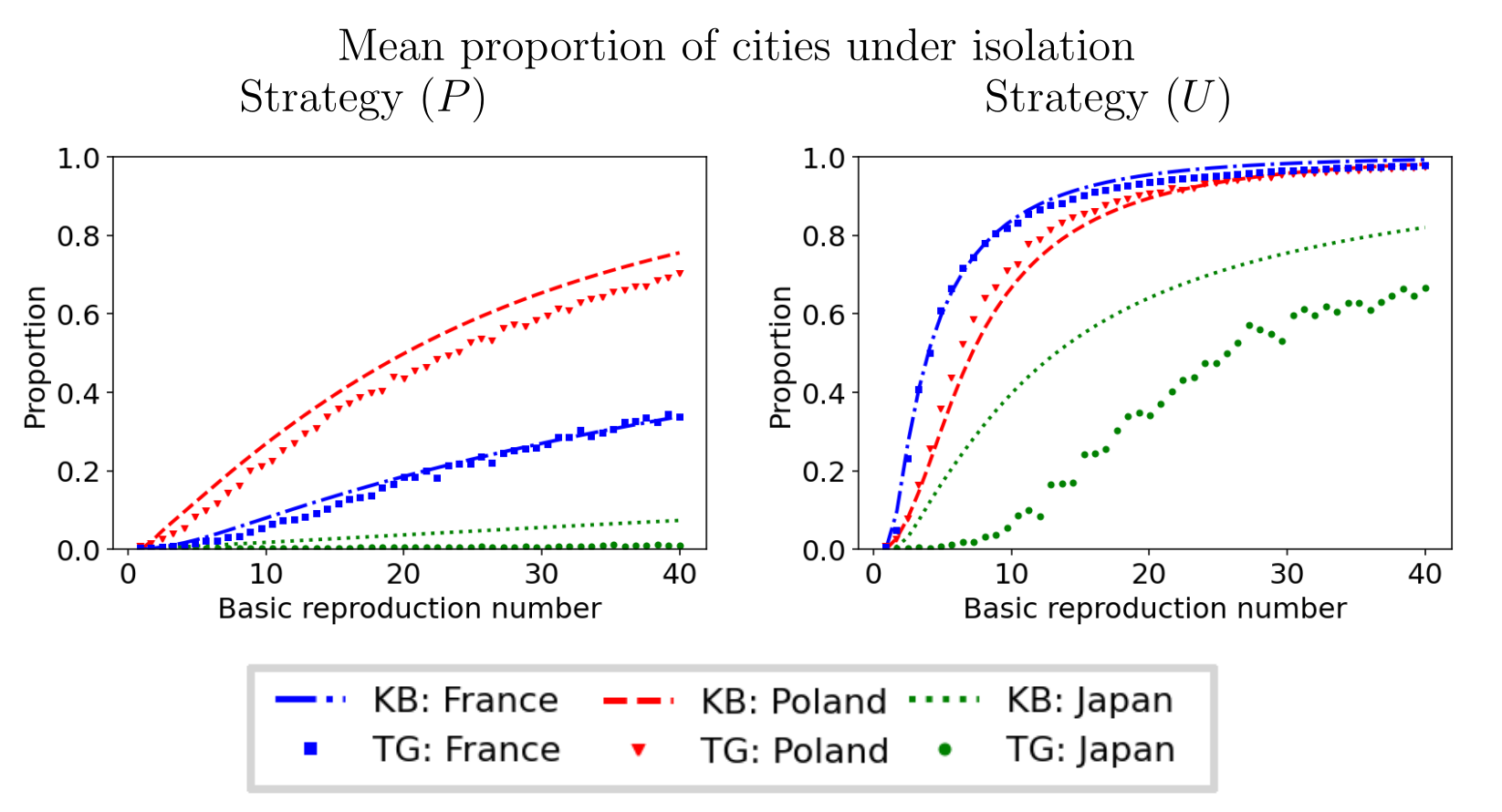}
	\end{center}
    \vcm{-0.6}
	\caption{
        Comparison of the mean TG and KB proportions of cities under isolation 
			depending on the basic reproduction number. Dashed lines show proportions when simulations are performed according to the transportation graph (TG), and scatter plots when
            proportions are calculated with KB approximations,  
            for France (in blue), Poland (in red) and Japan (in green),
			under strategy~\((P)\) on the left and strategy~\((U)\) on the right.
	}
	\label{Fig_incidence_cit}
\end{figure}

To be more precise, the TG- and KG- simulated values of proportions of cities as well as people under isolation
are taken as the simple averages over replicates.
For each replicate, the infection starts with a single infected city randomly chosen according to some distribution,
which we set as $\nu_{O, \rA}$.
Concerning the analytical formulas, recall that for each considered basic reproduction number $r$ a unique value of $k_B$ is associated to,
from which 
we can deduce
the corresponding infection probability function
$\pi^r:\bR_+\rightarrow [0, 1]$ (according to \eqref{piDef})
and outbreak probability function
$\eta^r:\bR_+\rightarrow [0, 1]$ (according to \eqref{survivalprobab}).
Conditionally on an outbreak,
the expected proportion of  people  (resp. of cities)
under isolation
is given as $\int_{\bR_+}x\cdot\pi^r(x)\, \mu(dx)$
(resp. as $\int_{\bR_+}\pi^r(x)\, \mu(dx)$)
and is to be compared  to the average number of people under isolation in simulated epidemics with (relatively) large outbreaks
(defined as to having at least 20 infected cities), see Figure~\ref{Fig_incidence_cO_ppl}.
In Figures~\ref{Fig_incidence_ppl} and \ref{Fig_incidence_cit},
the unconditioned averages on the other hand 
are compared to $\int_{\bR_+}\eta^r(y)\, \nu_{O, \rA}(dy) 
\int_{\bR_+}x\cdot\pi^r(x)\, \mu(dx)$
(respectively to $\int_{\bR_+}\eta^r(y)\, \nu_{O, \rA}(dy) 
\int_{\bR_+}\pi^r(x)\, \mu(dx)\;$),
where the factor $\int_{\bR_+}\eta^r(y)\, \nu_{O, \rA}(dy)$
expresses the probability that the outbreak starts from a single city whose size is randomly chosen according to $\nu_{O, \rA}$.

\begin{figure}[t]
\begin{center}   
    \includegraphics[width = \textwidth]{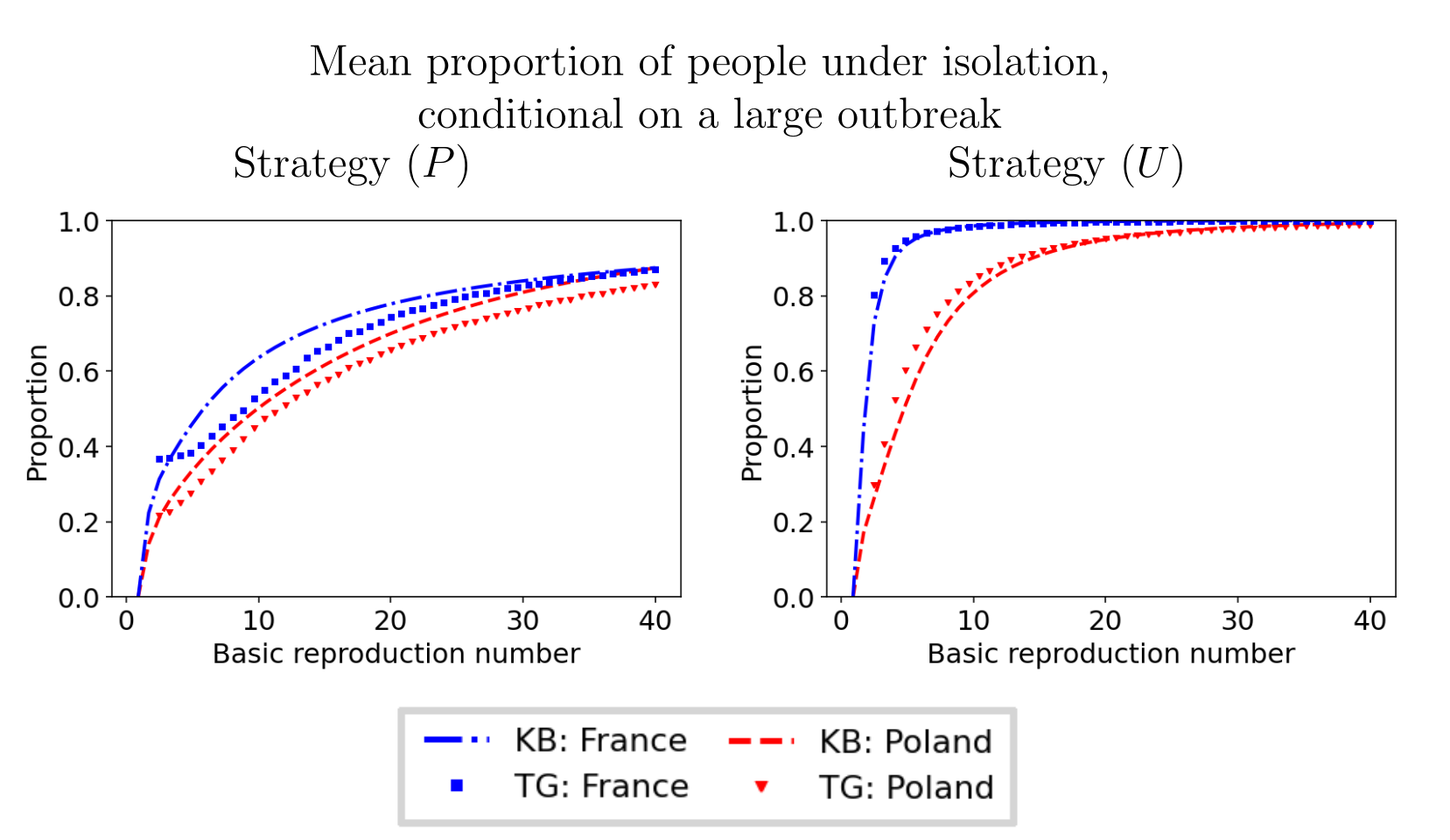}
	\end{center}
    \vcm{-0.5}
	\caption{
Comparison of the mean TG and KG proportions of people under isolation 
depending on the basic reproduction number
conditional on an outbreak with more than 20 cities under isolation,  
for France (in blue), and  Poland (in red),
 on the left 
under strategy~\((P)\) and on the right-panel under strategy~\((U)\).
For Japan, especially under strategy~\((U)\),
the threshold of outbreak size by 20 
is not associated to a relevant distinction in the final incidence, which is why the plot is not shown 
(the observed values do actually still largely disagree between KG and TG proportions).
}
\label{Fig_incidence_cO_ppl}
\end{figure}

In Figures~\ref{Fig_incidence_cit}
and \ref{Fig_incidence_cO_ppl}
similarly as in in Figure~\ref{Fig_incidence_ppl},
we find a very good agreement 
for France and Poland 
of the TG proportion of cities 
under isolation (resp. of people under isolation conditional on an outbreak)
with the KB such proportions.

More precisely concerning France, 
while the agreement is almost perfect  between the TG and KB proportion of cities under isolation 
under both strategies,
there is a slight downward shift
of the TG proportion of people under isolation conditional on an outbreak
as compared to the KB one
under strategy~\((P)\).
Though it was already visible in Figure~\ref{Fig_incidence_ppl},
the shift is actually more pronounced and more regular when such proportions are evaluated conditional on a large outbreak.
Concerning Poland on the other hand, 
the similar slight downward shift under strategy~\((P)\)
is observed for the TG proportion of cities under isolation 
equally as for the TG proportion of people under isolation, 
both in the unconditioned case
and conditionally on a large outbreak.
Under strategy~\((U)\), the shift is rather upwards.

Globally nonetheless, 
we mainly confirm the robustness of the alignment of TG and KB infection and outbreak probabilities for varying values of the basic reproduction number, for both France and Poland.
In contrast for Japan, 
the TG proportion of cities 
(as the one of people)
under isolation  does  not fit 
the KB corresponding proportion.
Large outbreaks are only rarely observed 
for Japanese TG epidemics,
which is why the plot is not shown in Figure~\ref{Fig_incidence_cO_ppl}.

Besides, due to the definition of the basic reproduction number, 
it was expected that the KG proportion under isolation 
should reflect the emergence of large outbreaks
as soon as $R_0$ would exceed 1.
Given the calibration of the virulence level in the transportation graph,
it was less clear  that 
the emergence of outbreaks when $R_0>1$ could similarly be reflected 
by the TG proportion of people or cities under isolation,
with outbreak sizes starting close to 0 for small $R_0$ values.
This is nonetheless what we observed, whatever the country or the strategy, 
both for people and for cities.

It deserves to be noted as well that 
very large $R_0$  values are associated with relatively intermediate TG and KG proportions of people under isolation, especially under strategy~\((P)\).
This hints  at the fact that most of the epidemic spread occurs within a few generations.
This also suggests that it is not absolutely essential to be particularly accurate in determining the incidence threshold 
that serves as an isolation criterion 
(for an anticipated  number of infected people).
\medskip

As the KG infection probabilities do nicely agree with the analytical prediction (the KB values), we suppose  that larger deviations 
from the analytical expressions for the infection and outbreak probabilities as well as expected proportions of people and cities under isolation 
are neither due to the finite number of cities
(in relation to the heterogeneity level in city sizes)
nor due to the values of $a$ and $b$. 
We rather hypothesize that the reason for the discrepancy lies in the geographic structure of Poland and Japan.

In Poland, there are due to historical reasons major socio-economic differences between the Eastern and Western part of Poland. These differences manifest beyond other also in the connectivity structure between cities, in particular in the work related mobility. In the Eastern part of Poland, these distances are of greater relevance than in the Western part, which has a very well developed transport system.  
 Our model is not adapted to these regional differences, 
and therefore we expect a less good fit than for France.

 The linear spatial structure  of Japan as well as the accordingly adapted railway system (of the Shinkansen) influences strongly the mobility matrix of Japan (in particular because our dataset is based on railway mobility). 
 Since our model takes into account only city sizes and not distances between cities measured in terms of geographic distances or railway connectivity, our model is expected to fit less well for countries like Japan where these distances seem to play a role.

\subsection{Regional lockdown strategies during COVID-19 pandemics and estimates of empiricial $R_0$-values for  SARS-Cov2 spread between cities}
\label{sec_empirical_R0}

During autumn/winter  2020/2021, regional lockdown regulations have been applied world-wide in many countries. The stringency of containment policies depended on the number of new cases that have been detected within a region in the last days (in general one to two weeks). In several European countries, the weekly or biweekly cumulative incidences per 100.000 individuals have been recorded during the pandemic and measures have been based on these numbers. For example in Germany, the first set of restrictions have been issued at a seven-day-incidence of 35 per 100.000, followed by additional measures at an incidence of 50 and strictest measures at an incidence of 200, see~\cite{HotspotGermany}. In this section, we aim to estimate empirical $R_0$-values for between-city-transmission of SARS-Cov2 under this form of regulation for Germany based on estimated individual reproduction numbers.

In autumn/winter 2020/2021, the individual reproduction number was estimated by \cite{HotzEtAl20} to lie between 1.3 and 1.5 based on reported cases of the Robert-Koch-institute (RKI).

Workplace-related infections have been estimated (in England and Wales in September/ December 2021)  to make up about 17~\% of all infections \cite{HoskinsEtAl}. We assume that the proportion of workplace-related infections were not lower in Germany  in autumn/winter 2020/2021, because in 2020 private activities were quite restricted in Germany.
Furthermore, the pandemic in autumn/winter 2021/2022 was strongly impacted by vaccination,
which was frequently mandatory for workplaces, hence the 17~\% should be seen as a lower-bound.
About 22~\% of all persons in employment commute more than 30~km, see \cite{CommuterDistance}, as well as about 55~\% of all individuals living in Germany are employed, see \cite{Erwerbstaetigkeit}.

In the same period, the ratio of the true number of corona cases to the number of detected corona cases was estimated to be 2.5-4.5 in Germany, see \cite{GEtAl}.

With these estimates, we arrive at a between-city-reproduction number (only based on work-related infections) of 
$50\cdot 2.5 \cdot 0.22\cdot 0.55 \cdot 0.17  \cdot 1.3= 3.34  $ to $50\cdot 4.5\cdot 0.22\cdot 0.55 \cdot 0.17 \cdot 1.5= 6.94$
for a city of size $10^5$, when commuting to other regions is prohibited from an incidence of 50, i.e.~assuming that 
the infection is spread to other cities
only in the last week before a city is put under lockdown. As a result, from an infected city, roughly 3-7 individuals in other cities get infected,  potentially triggering  an infection wave in those cities (cities which may not all be different). From this perspective, it comes as no surprise that the regional lockdown strategy was not successful and turned quickly into a country-wide lockdown.

\subsection{Validity of the estimation of $R_0$}
\label{sec_val_R0}

The basic reproduction number  is an important characteristic of an epidemic process. In a branching process, it is the expected offspring number, i.e. the expected number of infections caused by a typical infected entity.
In SIR models, and especially in the one we are considering here, 
it firstly defines the threshold (corresponding to $R_0>1$) below which no large outbreak would be expected to emerge from infection clusters.
In models with sufficiently well mixed populations,
an exponential growth is observed at the beginning of the epidemic
with $R_0$ as the per-generation infection rate.
At later time points, the initial $R_0$ fails to describe the rate of exponential growth.
A typical reason for the subsequent discrepancies is 
the decrease in the number of susceptible entities, the cities in our case.
Over time,
it becomes  less and less adequate to assume that the law of the size of newly infected cities is unaffected by the progression of the epidemic.
We expect that this second factor plays a  significant role in our models,
given that the large cities 
tend to be contaminated sooner
than the small ones,
which may render
the estimate of the basic reproduction number 
from per-generation infection numbers
unreliable.

To shed light on these effects,
 we calculated estimates for the basic reproduction number from simulated epidemics. 
Since the beginning of an epidemic strongly depends on the city size of the primarily infected city, we start estimating the $R_0$-value only when the epidemic has been run for several generations. For this purpose, we filtered the simulations for epidemics generating a relatively large outbreak with at least  
10 cities infected (within one generation).
For the $i$-th such simulated epidemic, let $Z^{(i)}(\ell)$ be the number of infected cities in generation $\ell$
and define $T^{(i)}$ as the first generation $g$ at which this process reaches at least $10$ 
(i.e. $Z^{(i)}(g)\ge 10> \max_{\ell \le g-1} Z^{(i)}(\ell)$). 
Given some parameter $G\in \mathbbm{N}$ to be adjusted, 
the $R_0$-value is evaluated
by considering the interval of $G$ generations 
after time $T^{(i)}$.
We thus consider the subset $\mathcal R_G$ 
 consisting of all simulation indices $i$
for which the epidemic remains ongoing in generation  $T^{(i)}+G$
for the $i$-th simulation.
The $R_0$-value corresponding to $G$ 
is expressed in terms of the 
trajectories $:g\in \{0,..., G\}\mapsto Z^{(i)}(T^{(i)}+g)$ for $i \in \mathcal R_G$.

Let $Z(\ell)$ 
be the random number of infected cities
at generation $\ell$
according to the branching process described in Section~\ref{sec_fwd_br}
and $T$ the first generation at which it reaches at least $10$.
According to the theorem of Heyde-Seneta,
there exists a random variable $W$ 
such that the following equivalence holds 
for $\ell$ such that $(R_0)^\ell$ is large,
on the event where the branching process 
survives, event that shall correspond to the occurrence of an outbreak:
\[Z(\ell) \approx W \cdot (R_0)^{\ell}\] 
and hence, we have on this event
\[\ln[Z(T+g)] 
\approx \ln(W) + T \ln(R_0) + g\ln(R_0).
\]

\begin{figure}[t]
	\begin{center}   
		\includegraphics[width = \textwidth]{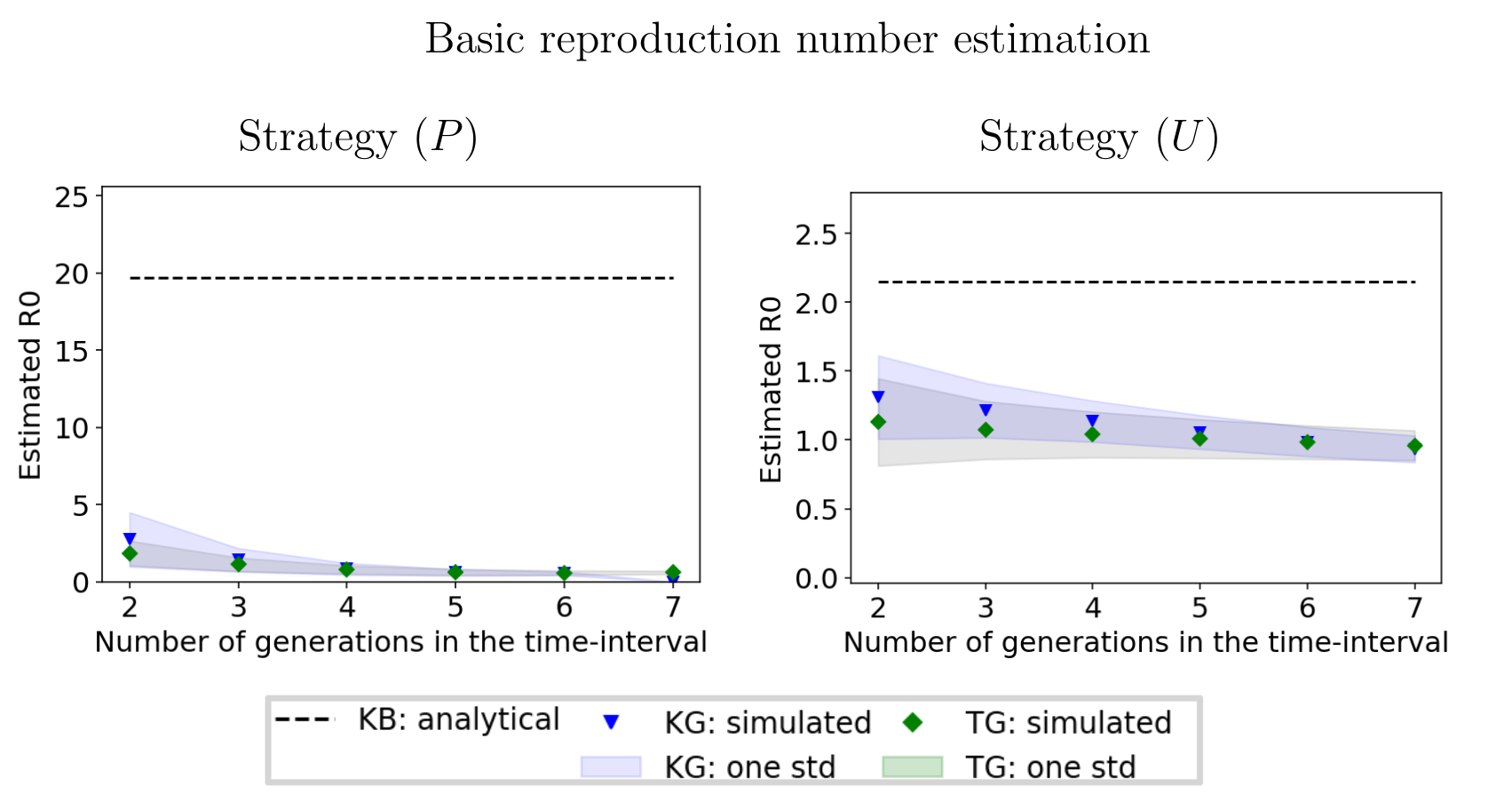}
	\end{center}
	\caption{
		$R_0$-expected value and variations calculated from simulated epidemics, see Section \ref{Section Reproduction Number}, 
		as compared to the analytical value. 
		The considered city size  distribution 
		is the one of France filtered according to D30+.
		Epidemics have been simulated according to the transportation graph (TG, in green with gray shading)
		or to the kernel graph (KG, in blue),
		on the left under strategy~\((U)\) 
		and for the right-panel under strategy~\((P)\).
	} 
	\label{Fig_R0_est_FC}
\end{figure}

This motivates to infer the logarithm of the basic reproduction number 
 by performing a least-square regression 
on the $i$-th trajectory
for each $i\in \mathcal R_G$
(for $g\in \II{0, G}$).
By not accounting 
for the $i$-th simulation
if the outbreak has stopped before generation $T^{(i)}+G$ (or before $T^{(i)}$ is reached),
we mimic the restriction to the event 
where the branching process 
survives.
The inferred value is then denoted $R_0^{(i)}(G)$
and we look how the distribution 
of $R_0^{(i)}(G)$ over $i\in \mathcal R_G$ varies 
with increasing $G$ up to 9.

In Figure~\ref{Fig_R0_est_FC}, the average over $i$
of these estimates of $R_0$  are depicted for various scenarios, together with intervals corresponding to one standard deviation on both side and the 5 and 95\% quantiles. 
With values of $G$  up to generation 7, 
we always keep more than $10\%$ trajectories accounted for
(the evaluation is put to 0 in generations 8 and 9 in the left-panel for strategy~\((U)\) due to this lack of trajectories).

For the different panels of the figure,
we used the 
same setup as  for the outbreak probabilities with the French dataset, see Section~\ref{sec_simulated_infection_and_outbreak_prob}.
Dataset D30+ was similarly used, as explained in Section~\ref{Sec:MobilityData},
and we 
 adjusted $k_B$ (thus fitting the expected $R_0$ value)
so that the KB infection probability 
is equal to 0.5 for a city size of $10^{5}$.
The left panel
shows the estimation under strategy~\((P)\),
the right panel under strategy~\((U)\).
The estimates derived from the transportation graph (TG, in green)
are displayed together with the ones from the kernel graph (KG, in blue),
while 
the dashed line corresponds to the theoretical $R_0$-value computed according to Section~\ref{sec_adj_str}.

In any case, the procedure for estimating 
the value of $R_0$ does not produce satisfying 
outcomes. The inferred values do not agree with the expected one, 
and the inferred values are rapidly declining with increasing $G$ (though it reduces variability in the estimation).
We can also observe that the TG and KG $R_0$-values are extremely close,
which demonstrates that the reason for the bad estimation of $R_0$
lies in the city size distribution
rather than more intricate aspects of the connection graph.
This is discussed in more details in Section \ref{sec_R0_est} of the Supplemental Material,
in which we show that the procedure could be used  to infer $R_0$-values that align with the expected one,
if the city size distribution were less heavy-tailed.

\subsection{Indegree and Outdegree}
\label{sec_degree}
\begin{figure}
\begin{center}   
    \includegraphics[width = \textwidth]{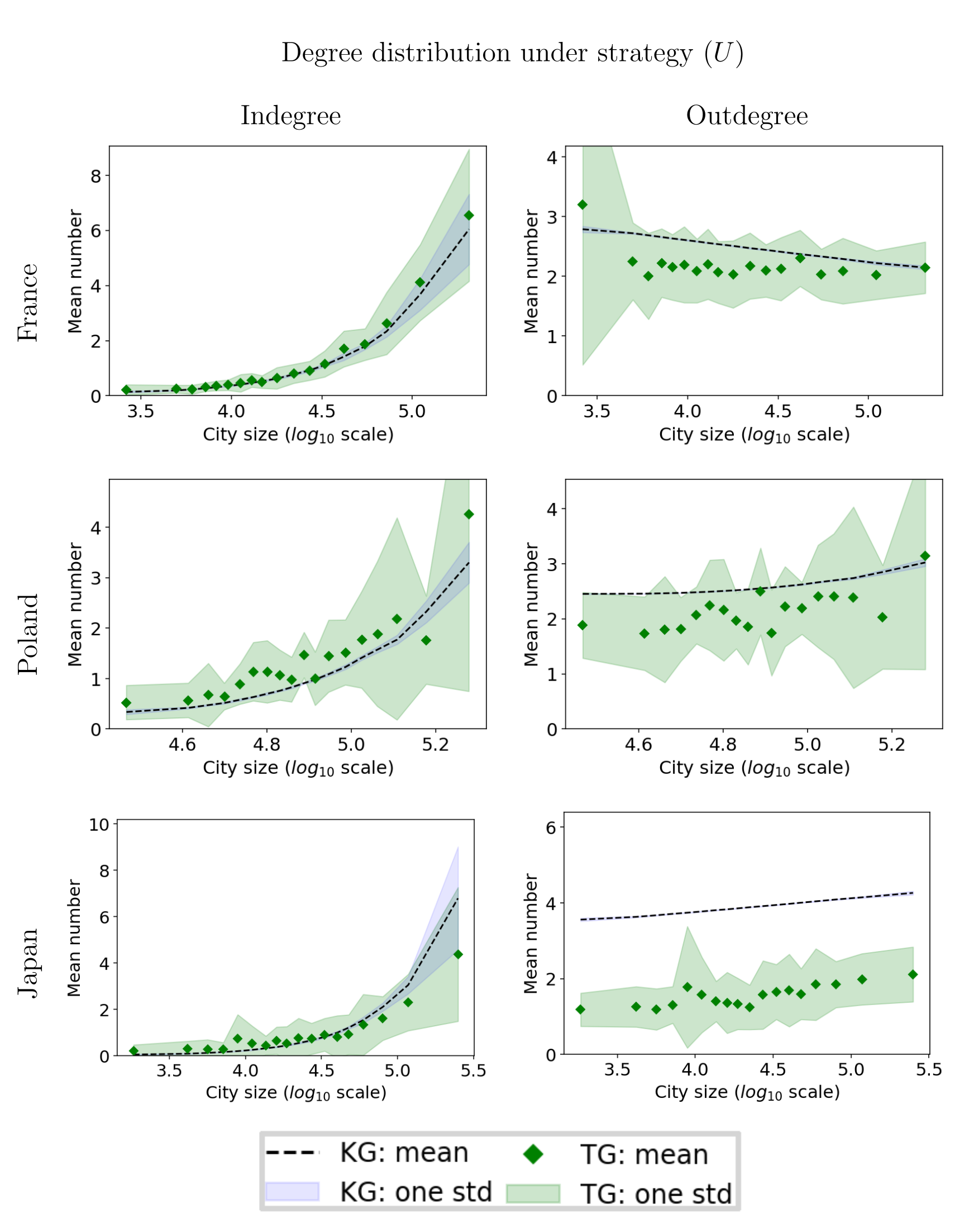}
	\end{center}
\caption{
Comparison of the theoretical and the empirical indegrees, on the left, and outdegrees, on the right,
of the epidemic graph for strategy~\((U)\) depending on the city size
for France, Poland and Japan.  The $R_0$ value is adjusted such that  the theoretical infection probability  is 0.5 for cities of size $10^{5}$, see Section \ref{sec:infection probability}.}
\label{Fig_Degree_Distrib}
\end{figure}

To assess the fit between the transportation and the kernel graphs,
we conducted an additional comparative analysis of the distributions of expected in- and outdegree as a function of city size. 
In Figure~\ref{Fig_Degree_Distrib},
we present the TG in- and outdegree distributions alongside the corresponding KG distributions based on strategy \((U)\) for  France, Poland, and Japan.

Recalling \eqref{pT_ij},
a natural proxy for the expected TG indegree 
of city $i$ is given for strategy \((P)\) 
by 
\begin{equation*}
\mathfrak{i}^{(P)}_i = k_C\, p_\vee\cdot \Big(\sum_{j} m_{ji} + m_{ij}\Big),
\end{equation*}
 while it is given for strategy \((U)\) by:
\begin{equation*}
\mathfrak{i}^{(U)}_i= k_C\, L_\vee\cdot \Big(\sum_{j} \frac{m_{ji} + m_{ij}}{x_j}\Big).
\end{equation*}
Similarly, a natural proxy for the expected TG outdegree 
of city $i$ is given for strategy \((P)\) 
by $\mathfrak{o}^{(P)}_i$ equal to $\mathfrak{i}^{(P)}_i$, while it is given for strategy \((U)\) by:
\begin{equation*}
\mathfrak{o}^{(U)}_i= k_C\, L_\vee\cdot \Big(\sum_{j} \frac{m_{ji} + m_{ij}}{x_i}\Big).
\end{equation*}

Recalling \eqref{pK_ij} and \eqref{Knu},
we deduce the following proxies 
$(\mathfrak{i}_i^{\kappa, (P)})_i$
$(\mathfrak{o}_i^{\kappa, (P)})_i$,
$(\mathfrak{i}_i^{\kappa, (U)})_i$,
and  $(\mathfrak{o}_i^{\kappa, (U)})_i$
for the kernel graph 
to be compared with $(\mathfrak{i}_i^{(P)})_i$,
$(\mathfrak{i}_i^{(U)})_i$,
and  $(\mathfrak{o}_i^{(U)})_i$:
\begin{align*}
\mathfrak{i}_i^{\kappa, (P)} &= \frac{k_B}{N} p_\vee \cdot \big( x_i^a \cdot \cZ_{1+b} 
+ x_i^{1+b} \cdot \cZ_{a} \big)
= \mathfrak{o}_i^{\kappa, (P)} ,\\
\mathfrak{i}_i^{\kappa, (U)} 
&= \frac{k_B}{N} L_\vee \cdot \big( x_i^a \cdot \cZ_{b} 
+ x_i^{1+b} \cdot \cZ_{a-1} \big),\\ 
\mathfrak{o}_i^{\kappa, (U)}
&= \frac{k_B}{N} L_\vee \cdot \big( x_i^{a-1} \cdot \cZ_{1+b} 
+ x_i^{b} \cdot \cZ_{a} \big).
\end{align*}

Note that the indegree and outdregree vectors  coincide (like their proxies $(\mathfrak{i}_i^{(P)})_i$ and  $(\mathfrak{o}_i^{(P)})_i$) in the case of strategy~\((P)\),
be it for the transportation graph $TG$ or the kernel graph $KG$
(for the kernel branching process $KB$ as well).
This is not the case for strategy~\((U)\),
which is why we focus on strategy~\((U)\)
in Figure~\ref{Fig_Degree_Distrib}.
 Actually, 
the indegree distribution 
under strategy~\((P)\) 
is also very similar 
to the one under strategy~\((U)\)
for both the transportation graph and the kernel graph
(figure not shown).

In Figure~\ref{Fig_Degree_Distrib},
 we observe that the outdegrees are nearly constant (i.e. almost not dependent on the city size), while indegrees are clearly increasing in city sizes.
It is striking that the degrees in the transportation graph
are much more variable than in the kernel graph, for cities of similar sizes.
Nonetheless, for France and Poland, the degrees in the transportation graph
align well with the one in the kernel graph 
when averaged over the bins 
(recall that there are 20 bins for each country, thus 15 to 40 cities per bin).
The fit is not  as good for the outdegree as for the indegree under strategy~\((U)\),
with  KG outdegrees
consistently overestimating the TG ones.
For Japan, we see that the fit of the degrees is rather good for small city sizes yet worsens as city size increases,
except for the outdegree under strategy~\((U)\) where there is a strong discrepancy (again with KG outdegrees being consistently larger than TG outdegrees).

At least with the data from France and Poland, 
we see that the kernel approach and the chosen form of the kernel
do not alter the relation between city size 
and the epidemic graph indegree
(and by extension between city size and outdegree under strategy~\((P)\)).
The reduced TG outdegree values under strategy~\((U)\)
as compared to KG ones
is presumably due to the geographic structure: the transmission routes at the individual level are restricted to a smaller subset of neighboring cities. Such increased competition probably reduces the number of cities ultimately struck.
It is thus expected (and can also by observed in Figure \ref{Fig_Degree_Distrib}) that the effect is greater in Japan, given the shape of the country.

\section{Discussion and conclusion}
\label{sec_disc}

In this study, we were interested in the performance of containment regulations that shall prevent the initial spread of a pathogen in a population within a country.
The models 
we considered compromise
between analytic tractability
and representation of reality. We eased epidemic dynamics
with respect to several aspects, which will be discussed in this
section.
\smallskip

A key step of our analysis is to approximate the epidemic process by
an infection process on a kernel graph.
In the kernel model, we assume that the strength of mobility between
cities only depends on the sizes of the cities (and no other parameters, like
geographical proximity or the like). The probability to travel to and to
travel from, resp., a city of size $x$ is proportional to $x^a$ and $x^b$, resp., that
is the parameter $a$ reflects the attractiveness of a cities and the parameter
$b$ the traveling habit of inhabitants depending on the size of the city they
are living in.
We estimated the parameters for France, Poland and Japan.

Interestingly, the parameter combination for France and Japan happen
to be very similar, with city size distribution with a power-law coefficient close to 1.8, an estimated attractiveness
$a \approx 1$ and an estimated emissiveness $b\approx0$. These values of $a$ and $b$ are all the more
surprising that they correspond to a kind of neutral case: All contact pairs
are equally likely, i.e. all individuals have roughly the same likelihood to
travel independent of the size of the city they are living in (i.e. $b = 0$) and
the target city is chosen proportional to its number of citizens ($a = 1$). On
the other hand for Poland, the power-law coefficient is larger (more than
3) meaning that the population is more evenly distributed between
powiats. While the estimated emissiveness coefficient $b$ is still quite close
to 0, the attractiveness coefficient $a$ is close to 2, meaning that attractivity of
powiats is significantly biased towards larger powiats.

In summary, although the choice 
of strategy~\((P)\)
over strategy~\((U)\)
has a relatively minor impact in France,
the impact in Poland is large.
This disparity in Poland is strongly influenced by 
the concentration of mobility towards the largest powiats.
These inferences demonstrate how various the level of city heterogeneity
can be. 
\smallskip

To check the empirical relevance of the kernel model 
that relies solely on the  contribution of population sizes,
without any specific reference to the spatial distribution of cities,
we conducted simulations of epidemics which take pairwise mobility patterns between cities into account. 
We termed 
the corresponding epidemic graphs as transportation graphs (TG).
Beyond others, we calculated the infection probability 
(abbreviated as $\pi$)
and the outbreak probability 
(abbreviated as $\eta$)
 as a function of city size, based on a branching approximation.
We compare these theoretical probabilities to simulated probabilities
that are obtained with averages of infection outcomes 
for 20 to 30 cities of similar size
over many epidemic runs for the transportation graph.
We are led to distinguish many situations,
according to the country of reference,
the regulation strategy (\((P)\) or \((U)\)),
and the rule for adjusting the scaling factor $k_B$.

For France and to a lesser extend for Poland,
we observe a remarkably good fit
of the simulated infection probabilities
by the theoretical values.
By comparing the simulated and theoretical proportions 
of infected cities and people 
for varying reproduction numbers, 
so varying scaling factor $k_B$,
we observe that this fit is robust to the stringency of the regulation,
which makes it a reasonable choice of optimization.
The quality of the fit is particularly noticeable 
given the level of simplifications
and the heterogeneity in the data.
A much poorer fit 
of the TG simulated infection probabilities is observed for Japanese data, which 
indicates that the spatial structure of the country 
is not well-captured by the dependency on city size.

A possible explanation for the good fit could be the centralist transportation structure of France, which diminishes the impact of the geographic distance between cities. In Poland, the transportation network is in the east less developed than in the west and north. This implies that in particular in the eastern part of Poland geographical distances are more relevant. In Japan, the linear configuration  of the country influences strongly the transportation network, even though cities located along the railway of the same Shinkansen line are in terms of accessibility relatively  close to one another. These two effects generate a particular geometric structure
that is likely to be of significant interest for inclusion in the model.
\smallskip

Approximations with branching processes are a key tool in our analysis.
We use them to derive analytic formula for the outbreak and infection  probabilities. Furthermore, we base our definition of a basic reproduction number $R_0$ 
on an approximation with branching processes,
which implicitly assumes an infinite reservoir of cities in which the new generation of infections can emerge.

Even though this approximation is valid only for very few generations
 in our setting, the derived analytical infection probabilities yield (as confirmed by our simulation studies)  good approximations of the corresponding TG-infection probabilities (for French and Polish mobility networks). 
Our comparison of the two strategies are based on the infection probabilities,
making it reasonable to assume that our results on the relative effectiveness of strategy~\((U)\) compared to~\((P)\) are solid, particularly for France and, to a slightly lesser degree, for Poland.

The reason for the good fit is potentially  that the probability for a city to get infected is depending on some random outcomes occurring within a few generations (potentially a few -- backwards in time -- success events characterized by a sufficiently high number of secondary infections).

The formula for the KB infection probability
reflects a structural property of the epidemic graph  that is independent of how the epidemic is initiated,
namely the relative size of the largest forward connected component.
This property presumably carries over
to the TG infection probabilities
when they align with the KB infection probabilities.
\medskip

We evaluated the fit of our model to actual transmission patterns
by exploiting commuting data retrieved from census data. An alternative source of mobility data is given by GPS data obtained from mobile phones. An obstacle for the analysis of this kind of data are privacy restrictions as well as their frequent only commercial availability. 
With this kind of data, one could however investigate the effect of mobility variability in time, e.g. in winter and in summer during holidays.

Considering that many countries attempt to contain epidemic waves during a pandemic, it would be reasonable to investigate our model in more general settings, 
i.e. to consider epidemic scenarios that include e.g. successive waves, maybe in terms of an SIRS-like epidemic between cities,
and/or vaccination.
 Finally, it would be  valuable to explore the optimization of timing and strength of containment strategies by incorporating utility functions which factor in economic costs, health and social burden (for first steps in this direction see \cite{SchaeferEtAl}).

\subsection*{Data availability}

The data we analysed as well as the code we used for simulation is available on forgemia (a platform for data storage at Inrae):
\newline
\href{https://forgemia.inra.fr/aurelien.velleret/simulations_containment_strategies_and_city_size_heterogeneity.git}{\path{https://forgemia.inra.fr/aurelien.velleret/simulations_containment_strategies_and_city_size_heterogeneity.git}}

Additionally some data processing code and intermediate data file are available at: \url{https://github.com/MOCOS-COVID19/pl-mobility-versus-size}.\\
\\

\subsection*{Acknowledgements}

CP and AV acknowledge support from the German Research Foundation through grant
PO-2590/1-1. VB, TK, CP, PS and AV acknowledge support during the JTP 2022 ”Stochastic Modelling
in the Life Science” funded by the Deutsche Forschungsgemeinschaft (DFG, German Research
Foundation) under Germany’s Excellence Strategy – EXC-2047/1 – 390685813. We thank University of Luebeck and Wroclaw University of Science and Technology for financing a visit of TK in Luebeck.

\subsection*{Author contributions: CRediT}

\textbf{Viktor Bezborodov:} Conceptualization, Methodology, Software, Formal analysis, Writing - Review \& Editing.
\textbf{Tyll Krueger:} Conceptualization, Methodology, Writing - Original Draft, Review \& Editing, Visualization.
\textbf{Cornelia Pokalyuk:} Conceptualization, Funding acquisition, Methodology, Writing - Original Draft, Review \& Editing, Visualization.
\textbf{Piotr Szym\'anski:} Software, Formal analysis, Data Curation, Writing - Review \& Editing.
\textbf{Aur\'elien Velleret:} Conceptualization, Methodology, Software, Validation, Formal analysis, Data Curation, Writing - Original Draft, Review \& Editing, Visualization.

\subsection*{Declaration of generative AI and AI-assisted technologies in the writing process.}
During the preparation of this work AV used YIAHO and DeepL Translation to explore alternative formulations that could enhance readability. After using this tool/service, the authors reviewed and edited the content as needed and take full responsibility for the content of the published article.

\bibliographystyle{elsarticle-num-names} 
\bibliography{LS}

\renewcommand{\thesection}{\Alph{section}}
\setcounter{section}{18}

\renewcommand{\thefigure}{S\arabic{figure}}

\pagenumbering{roman}

\section{Supplementary Material}
\label{sec_SM}

In the Supplementary Material, we provide additional 
figures, analyses and insights   that strengthen the basis of the main study's findings.

In Section~\ref{sec_exp_form_prob}, we explain why we chose to express the infection probability in \eqref{pM_ij} with an exponential form, highlighting the theoretical reasons  behind this  choice.
In Section~\ref{sec_ineq}, we present the proof of Proposition~\ref{ineq}, which compares the efficiency of the two strategies in the case of  rank-one kernels.
In Section~\ref{sec_Spec_An}, we conduct the spectral analysis of the transfer operator, which is further divided into two parts. 
Beforehand, Section~\ref{sec_Prames} recalls the branching approximation,
stated more formally in terms of Poisson random measures.
In  Section~\ref{sec_r0_spec}, we introduce the transfer operator
and examine its relationship with the reproduction number (R0). In Section~\ref{sec_long_time_T}, we extend our study of the spectral properties associated with the transfer operator.
In Section~\ref{sec_eig_cent}, we discuss the topic of eigenvector centrality. The theoretical basis of this notion is introduced in Section~\ref{EigC},
while we detail in Section~\ref{sec_est_eigC}  the methods employed for its estimation with the data from France, Poland and Japan.
Finally in Section~\ref{sec_factors}, we investigate additional influential factors affecting the outcomes of numerical evaluation.

\subsection{Justification for the exponential form of the edge probability}
\label{sec_exp_form_prob}

We provide in this Section~\ref{sec_exp_form_prob} some motivations behind 
our choice of formula \eqref{pM_ij}
to define edge probabilities,
as an approximation based on large city sizes.

Let us first note an advantage of this exponential form 
in that we can interpret $p_{ij}$
as a failure probability 
of two independent events, 
occurring with probability $p^O_{ij}$
for infections from the outside,
and with probability $p^I_{ij}$
for infections from the inside,
where:
\begin{equation*}
	p^O_{ij} = 1- \exp\left[-k_C\, M_{ij} \frac{L(x_i)}{x_i}\right]\,,
	\quad p^I_{ij} = 1- \exp\left[-k_C\,M_{ji} \frac{L(x_i)}{x_i}\right]\,.
\end{equation*}
Indeed: $1- p_{ij} = (1-p^O_{ij})\cdot (1-p^I_{ij})$.

Formula \eqref{pM_ij} can actually be interpreted as an approximation motivated by 
the large number of involved visits,
as in the following illustrating model.
Say that $M_{ij}/(x_i x_j)$ 
is the rate of visits of individual $m$
by individual $\ell$,
for any pair $(\ell, m)$
such that $\ell$ is a citizen of city $i$
and $m$ a citizen of city $j$.
$L(x_i)/x_i$
corresponds to the probability that such an individual $\ell$ has been infected before the lockdown.
The effectiveness at which an infected individual transmits the disease during a visit
leading to a large outbreak in city $j$
and the duration of its infectious period is captured by the term $k_C$.
A natural proxy for the probability 
that an outbreak is generated in city $j$
due to a visit of $m$ by $\ell$
is then given by:
\begin{equation*}
	p^v_{ij} = k_C
    \cdot \frac{M_{ij}}{x_i x_j}  
  \cdot   \frac{L(x_i)}{x_i}.
\end{equation*}
Assuming these events to be independent for the pairs $(\ell, m)$
leads to a global probability
for an infection from the outside
of 
\begin{equation}\label{expapprox}
	1 - \prod_{(\ell, m) \in C_{ij}} (1-p^v_{ij}) 
	= 1- (1-p^v_{ij}) ^{x_i\cdot x_j},
\end{equation}
where $C_{ij}$ denotes the set of pairs $(\ell, m)$ (which has cardinality $x_i\cdot x_j$).
\eqref{expapprox} is approximated by $p^O_{ij}$
due to  $x_i\cdot x_j$ being large.
A similar reasoning leads to the probability $p_{ij}^I$.
\medskip

It may be argued that 
considering $L(x_i)/x_i$
in the above reasoning 
as the probability that such an individual 
$\ell$ has been infected
is somewhat misleading as not all infected individuals have the same potential for further infections. The individuals that got infected just before the lockdown 
are prevented for further spread. 
Nonetheless, we can justify as follows that considering a proportionality between the number of infectious contacts and $L(x_i)$ remains reasonable. We simply assume that the epidemic grows exponentially within the cities, with the growth rate $\rho$ being independent of the specific city, as well as the average infectivity profiles of the individuals.

Let us denote by $T_i$ the time at which isolation of city $i$  is enforced,
and by $\mathcal F_i(t)$
the rate at which new infections occur
at time $t$ in city $i$.
Under this assumption, 
$\mathcal F_i(T_i - t) \approx \mathcal F_i(T_i) e^{-\rho t}$
should hold provided $\mathcal F_i(T_i - t)$ is not too small, which leads to the following estimate:
\begin{equation}
    L(x_i) \approx \int_0^\infty \mathcal F_i(T_i) e^{-\rho t} dt = \rho^{-1} \cdot \mathcal F_i(T_i).
    \label{eq_cFi}
\end{equation}

Let us denote by $\mathcal A_i$
the value obtained 
by integrating the average infectivity 
of these infected citizens of $i$ up to time $T_i$.
The probability that a contact between an individual $\ell$ of city $i$ and an individual $m$ of city $j$
is actually an infectious contact 
is then assumed to be proportional to $\mathcal A_i/x_i$, 
which is presumably more accurate than the previous estimate based on $L(x_i)/x_i$.

If we would assume that infectious individuals have a fixed infectivity rate,
the average infectivity would be proportional to the incidence. 
Then, $\mathcal A_i$ would be obtained up to a given factor by integrating the incidence up to time $T_i$.

We may generally assume that the average infectivity 
is prescribed by a function $\lambda$
that only depends on the duration since the infection of the individual,
in other words of the infection-age.
The previous assumption then corresponds to the case 
where $\lambda(u)$ would be proportional to the probability that the individual is still infectious 
after a duration $u$ since its infection.

At time $T_i -s$,
individuals that were infected at time $t$
have an average contribution of $\lambda(t-s)$ to $\mathcal A_i$.
Thus:
\begin{equation}
    \mathcal A_i \approx \int_0^\infty 
   \mathcal F_i(T_i-t)    \int_0^t
   \lambda(t-s) ds\, dt
 =\widetilde \lambda\cdot \mathcal F_i(T_i),
\label{eq_Ai}
\end{equation}
where the factor $\widetilde \lambda$
that integrates the infectivity effect
is independent of $i$:
\begin{equation*}
 \widetilde \lambda = \int_0^\infty e^{-\rho t}
 \int_0^t
   \lambda(u) du\, dt
 = \rho^{-1} \int_0^\infty \lambda(u) e^{-\rho u} du.
\end{equation*}
Combining \eqref{eq_cFi} and \eqref{eq_Ai}
leads to a factor $\rho \cdot \widetilde \lambda$ between $\mathcal A_i$ and $L(x_i)$.
It can be incorporated into $k_C$
to correct for the prevented infections.

Deviations from this proportionality 
could naturally arise due to temporal variations in transmission rates or heterogeneity factors specific to each city. However,  assuming a proportionality relation to $L(x_i)$
serves as a practical but also reasonable proxy.

\subsection{Proof of Proposition \ref{ineq}
	and comparison of strategy efficiency under rank-one kernels}
\label{sec_ineq}

Recall our claim in Proposition \ref{ineq}-$ii)$
that $Q_U > Q_P$ 
when $a= 1+b<1$
and the values of $L_\vee$ and $p_\vee$
are set according to a prescribed value $I_\star$
common to $I_U$ and $I_P$.

Instead of prescribing the values for $I_U$ and $I_P$,
it is more convenient to 
prescribe the same value $Q_\star$
for the numbers $Q_U$ and $Q_P$ of people eventually under isolation.
Since $I_U, I_P, Q_U$ and $Q_P$ are increasing functions of the threshold values,
proving Proposition \ref{ineq} is equivalent to showing
that whatever $Q_\star$, it holds:
\begin{enumerate}
	\item[i)] $I_U = I_P$ if $a = 1$,
	\item[ii)] $I_U < I_P$ if $a < 1$,
	\item[iii)]  $I_U > I_P$ if $a > 1$.
\end{enumerate} 

In the rank-one kernel situation,
\eqref{piDef}  simplifies into:
\begin{align*}
	\pi(x) 
	= 1- \exp\Big[  - 2 K_{O, \lA}(x) \int_0^\infty \pi(y) \nu_{O,\lA}(dy)\Big].
\end{align*}
Hence, 
we deduce from  \eqref{KCD} that $\pi_U$ and $\pi_P$
can be expressed in the form 
\begin{align*}
	\pi_\gamma(x) := 1- \exp(-\gamma \cdot x^a),
\end{align*}
that is $\pi_U = \pi_{\gamma_U}$ with
\begin{equation}
	\begin{split}
		\gamma_U 
		&:= 2 k_B \int_0^\infty L_\vee\cdot y^{a-1}\cdot \pi_U(y) \beta(dy)
		\\&= 2 k_B \int_0^\infty L_\vee\cdot y^{a-1}\cdot \pi_{\gamma_U}(y) \beta(dy),
	\end{split}
	\label{gammaU}
\end{equation}
and similarly, 
$\pi_P = \pi_{\gamma_P}$
with 
\begin{equation}
	\begin{split}
		\gamma_P 
		&:= 2 k_B \int_0^\infty p_\vee\cdot y^{a}\cdot \pi_P(y) \beta(dy)
		\\&= 2 k_B \int_0^\infty p_\vee\cdot y^{a}\cdot \pi_{\gamma_P}(y) \beta(dy).
	\end{split}
	\label{gammaP}
\end{equation}

Since the function $:\gamma\mapsto \pi_{\gamma}$
is increasing, by \eqref{QDef} $Q_U = Q _P= Q_\star$
implies $ \gamma_U = \gamma_P=\gamma_\star$ for some constant $\gamma_\star>0$.
Furthermore,
by exploiting \eqref{gammaU} and \eqref{gammaP},
we can express $L_\vee$ and $p_\vee$ in terms of $\gamma_U$ and $\gamma_P$ resp. 
We obtain the ratio $I_U/I_P$ as a function of $\gamma_\star$, $\beta$ and $a$.

Let us first consider the case where $a\ge 1$.
For ease of notations, 
we denote for any $\zeta>0$: 
\begin{align*}
	\cZ^\star_\zeta 
	:= \int_0^\infty x^{\zeta}\cdot \pi_{\gamma_\star}(x) \beta(dx).
\end{align*}
By combining \eqref{gammaU} and \eqref{gammaP} we then deduce 
$L_\vee \cZ^\star_{a-1} = p_\vee \cZ^\star_{a}$.
Then, the ratio $I_U/I_P$ is expressed as follows:
\begin{align*}
	\dfrac{I_U}{I_P}
	=\frac{L_\vee \cZ^\star_0 }{p_\vee \cZ^\star_1 }
	= \dfrac{\cZ^\star_0}{\cZ^\star_{a-1}}\cdot \dfrac{\cZ^\star_a}{\cZ^\star_1}.
\end{align*} 
When $a= 1$, we directly obtain that $I_U = I_P$.

Otherwise, note that $0< 1\wedge (a-1)<1\vee (a-1)< a$.
This ratio is compared to 1 thanks to Hölder's inequality:
\begin{align*}
	&\Big( \dfrac{\cZ^\star_{a-1}}{\cZ^\star_0}\Big) ^{1/(a-1)}
	> \Big( \dfrac{\cZ^\star_{a}}{\cZ^\star_0}\Big) ^{1/a}\cdot 1,
	\quad \Big( \dfrac{\cZ^\star_{1}}{\cZ^\star_0}\Big)
	> \Big( \dfrac{\cZ^\star_{a}}{\cZ^\star_0}\Big) ^{1/a},
\end{align*}
which entails
\begin{align*}
	\dfrac{\cZ^\star_{a-1}\cdot \cZ^\star_{1}}{(\cZ^\star_0)^2}
	> \Big( \dfrac{\cZ^\star_{a}}{\cZ^\star_0}\Big) ^{\frac{a-1}{a} + \frac{1}{a}}
	= \dfrac{\cZ^\star_{a}}{\cZ^\star_0}.
\end{align*}
This inequality is equivalent to  
$I_P>I_U$, which concludes the proof in the case where $a>1$.

When $a<1$, we look for the shifted moments:
\begin{align*}
	\cZ^\ltimes_\zeta 
	:= \int_0^\infty x^{\zeta}\cdot x^{a-1} \pi_{\gamma_\star}(x) \beta(dx).	
\end{align*}
Then, the ratio  $I_U/I_P$  is expressed as follows:
\begin{align*}
	\dfrac{I_U}{I_P}
	= \dfrac{\cZ^\ltimes_{1-a}}{\cZ^\ltimes_{0}}\cdot \dfrac{\cZ^\ltimes_1}{\cZ^\ltimes_{2-a}}.
\end{align*} 
Note that $0< 1\wedge (1-a)<1\vee (1-a)< 2-a$.
Thanks to Hölder's inequality, with a similar reasoning as before,
we deduce this time that $I_P<I_U$ whatever $Q_\star$.
This concludes the proof of Proposition \ref{ineq}.

\subsection{Proof of Proposition \ref{p_Comp3}}
\label{sec_pr_adj_R0}

Recall that Proposition \ref{p_Comp3}
involves the following 
quantities in terms of the two parameters $a, b\in \bR$
and of the probability measure $\beta$ on $\bR_+$.
$\cZ_\gamma := \int y^\gamma \beta(dy)$ for $\gamma>0$,
$r^U_0 =  \cZ_{a+b} + \sqrt{\cZ_{2a-1} \cZ_{2b+1} }$
and $r^P_0 
=  \cZ_{1+a+b} + \sqrt{\cZ_{2a} \cZ_{2+2b}}$ (as stated in \eqref{iniquity sim wickedness} and \eqref{dry run sim rehearsal}).
These notations were motivated by the observations that $R^{(U)}_0 = k_B \, L_\vee\, r^U_0$
(resp. $R^{(P)}_0 = k_B \, p_\vee\, r^P_0$)
corresponds 
to the basic reproduction number under strategy~\((U)\) (resp. under strategy~\((P)\)).
Moreover, we have assumed that $\cZ_{\bar\gamma}<\infty$,
where 
$\bar \gamma := (2 a)\vee (2+2b)$,
and that $\delta \le (2 a-1) \wedge (2 b+1)$.

Independently of the value $R_\star$ to which we set both $R^{(U)}_0$
and $R^{(P)}_0$,
we first observe that 
$p_\vee \le L_\vee \cdot \cZ_\delta / \cZ_{\delta +1}$
(which we generally want to prove)
is actually equivalent to $\cZ_{\delta+1}\cdot  r^U_0 \le \cZ_\delta\cdot r^P_0$.
To deduce the later inequality,
it is enough to show that:
\begin{equation} \label{morass2}
	\begin{split}
		\cZ_{a+b} \cZ_{\delta+1} &  \leq  \cZ_\delta \cZ_{a+b+1}, 
		\\
		\cZ_{2a-1} \cZ_{2b+1}  \Big( \cZ_{\delta+1}\Big)^2 & \leq \Big( \cZ_{\delta}\Big)^2 \cZ_{2a} \cZ_{2b+2}.
	\end{split}
\end{equation}

The inequalities \eqref{morass2} follow from H\"older's inequality.
In general 
for $\gamma>~\delta$, choosing  $q=\gamma-\delta+1$,
$p=\frac{q}{q-1}$,
we note:
\begin{align*} 1/p+1/q=1, \quad x^{\delta+1} &= x^{\delta/p}\cdot x^{(\gamma+1)/q},\quad 
	\\ x^\gamma &= x^{(\gamma+1)/p}\cdot x^{\delta/q}
\end{align*}
which entails
\[
\cZ_{\delta+1}  \leq \Big( \cZ_\delta\Big) ^{\frac{1}{p}}\cdot \Big( \cZ_{\gamma+1}\Big) ^{\frac {1}{q}},
\quad 
\cZ_{\gamma}  \leq \Big( \cZ_{\gamma+1}\Big)^{\frac{1}{p}}
\cdot \Big( \cZ_\delta \Big)^{\frac {1}{q}} .
\]
The inequalities \eqref{morass2} are deduced 
from the particular cases of $\gamma$ in $\{a+b,\ 2a-1, \ 2b+1\}$. In these three cases,
$\delta < (2 a-1) \wedge (2 b+1)$ entails $\gamma> \delta$.

In the case where $a = 1+b$, 
$r^U_0 = \cZ_{2a-1}= \cZ_\delta$
while $r^P_0 = \cZ_{2a}= \cZ_{\delta+1}$,
which entails Eq. \eqref{AdjustmendA=1+B}.

This ends the proof of  Proposition \ref{p_Comp3}.

\subsection{Spectral analysis}
\label{sec_Spec_An}

\subsubsection{Construction of the branching approximation in terms of Poisson random measures}
\label{sec_Prames}

We recall the construction of the forward and backward 
branching approximation described in Section~\ref{sec_fwd_br},
while expressing it in the more general terms of a convergence of point measures.
This section is intended for researchers who are well versed in probability and look for a more formal derivation. It 
may be skipped by recalling the more accessible description of Section~\ref{sec_fwd_br}.

We write $i\xrightarrow{O} j$ if there is an edge between $i$ and $j$ corresponding to an infection of $j$ from the outside generated by a visitor from city $i$,
and by $i\xrightarrow{I} j$   if there is an edge between $i$ and $j$ corresponding to an infection of $j$ from the inside initiated by a visit in city $i$. 
We consider the following random point process on $\{O, I\}\times \bR_+$:
\begin{equation*}
\xi^{N'}_{i,\rA}
:= \sum_{i\xrightarrow{O} j} \delta_{(O, x_j)}
+ \sum_{i\xrightarrow{I} j} \delta_{(I, x_j)}.
\end{equation*}

The independence properties and  \eqref{pVU} yield that 
$\xi^{N'}_{i,\rA}$
converges weakly in law as $N'$ tends to infinity
to a Poisson  random measure (a PRaMe) $M_{i, \rA}$ 
on the state space $\{O, I\}\times \bR_+$
with intensity $K_{O, \rA}(x_i) \cdot \delta_O \otimes  \nu_{O, \rA} 
+ K_{I, \rA}(x_i) \cdot \delta_{I}\otimes  \nu_{I, \rA}$.
The two functions $K_{O, \rA}$, $K_{I,\rA}$ and probability measures $\nu_{O, \rA}, \nu_{I, \rA}$
 are given in \eqref{Knu}.

By extension,
each generation of the branching process  {\bf $\text{V}^\infty$} 
is composed of elements of $\{O, I\}\times \bR_+$.
It is produced 
in terms of the previous generation
by means of PRaMes 
 with intensity measures $K_{O, \rA} \cdot \delta_{O}\otimes \nu_{O, \rA} + K_{I, \rA}\cdot \delta_{I}\otimes \nu_{I, \rA}$
 that are drawn independently for each infected city,
 as a function of its city size (whether the infector is in state $O$ or $I$ do not matter for the next generation).

As argued in Proposition~\ref{V2}, the projection of the process  {\bf $\text{V}^\infty$} 
to the state space $\{O, I\}$ yields a process  {\bf $\text{V}^2$}  that is itself a branching process. 
The next generation is produced 
by means of PRaMes that are drawn independently for each infected city,
 with intensity measures 
 being either $\int_0^\infty K_{O, \rA}(x) \nu_{O, \rA}(dx)
 \cdot \delta_{O} + \int_0^\infty K_{I, \rA}(x) \nu_{O, \rA}(dx)\cdot \delta_{I}$
 for an infector of type~$O$
 or $\int_0^\infty K_{O, \rA}(x) \nu_{I, \rA}(dx)
 \cdot \delta_{O} + \int_0^\infty K_{I, \rA}(x) \nu_{I, \rA}(dx)\cdot \delta_{I}$
 for an infector of type~$I$.
 Proposition~\ref{V2} can be deduced as a consequence of the decomposition properties of such PRaMes.
 \smallskip

The backward in time branching process is justified similarly.
For a target city $j$, we say that $j'$ is an infector of $j$ from the outside (resp. from the inside) if $j' \xrightarrow{O} j$ (resp. if $j\xrightarrow{I} j') $. Then the random variable
\begin{align*}
\xi^{N'}_{j,\lA}
:= \sum_{j'\xrightarrow{O} j} \delta_{(O, x_{j'})}
+ \sum_{j'\xrightarrow{I} j} \delta_{(I, x_{j'})}.
\end{align*}
gives the possible infectors of $j$, where $O$ and $I$, resp., indicate if the infecting city introduced the infection from the outside or from the inside and $x_{j'}$ denotes the size of the infecting city.
In analogy to the process $\xi^{N'}_{j,\rA}$  and thanks to \eqref{pVU} and to the independence properties   
in the large population limit, $\xi^{N'}_{j,\lA}$
converges weakly to a PRaMe $M_{j, \lA}$ on the state space $\{O,I\} \times \mathbbm{R}_+$ 
with intensity $K_{O, \lA}(x_j) \cdot \delta_{O} \otimes  \nu_{O, \lA} 
+ K_{I, \lA}(x_j) \cdot \delta_I \otimes  \nu_{I, \lA}$.
The two functions $K_{O, \lA}$, $K_{I,\lA}$ and probability measures $\nu_{O, \lA}, \nu_{I, \lA}$
 are given in \eqref{KCD}.

By extension,
each generation of the branching process  {\bf $\Lambda^\infty$}  is produced 
in terms of the previous generation
by means of PRaMes 
 with intensity measures $K_{O, \lA} \cdot \delta_{O} \otimes  \nu_{O, \lA} 
+ K_{I, \lA} \cdot \delta_I \otimes  \nu_{I, \lA}$
 that are drawn independently for each city,
 as a function of its city size.

 As argued in Proposition~\ref{L2}, the projection of the process  {\bf $\Lambda^\infty$} 
to the state space $\{O, I\}$ yields a process  {\bf $\Lambda^2$}  that is itself a branching process. 
The next generation is produced 
by means of PRaMes that are drawn independently for each infected city,
 with intensity measures 
 being either $\int_0^\infty K_{O, \lA}(x) \nu_{O, \lA}(dx)
 \cdot \delta_{O} + \int_0^\infty K_{I, \lA}(x) \nu_{O, \lA}(dx)\cdot \delta_{I}$
 for a city of type~$O$
 or $\int_0^\infty K_{O, \lA}(x) \nu_{I, \lA}(dx)
 \cdot \delta_{O} + \int_0^\infty K_{I, \lA}(x) \nu_{I, \lA}(dx)\cdot \delta_{I}$
 for a city of type~$I$.

\subsubsection{Transfer operator and reproduction number}
\label{sec_r0_spec}
$R_0$ as defined in \eqref{RoDef}
is more directly seen as the principal eigenvalue 
of an  integral operator $\cT$
which is defined as follows.
$\cT$ is directly
associated to the intensity of the PRaMes $M_{i, \rA}$,
recall Section~\ref{sec_Prames}.


We denote by
$\mathcal B(\bR_+)$ the set of bounded measurable functions on $\bR_+$ and write $\langle \mu \bv f\rangle := \int_0^\infty f(x) \mu(dx)$
for any $\mu\in \mathcal M_1(\bR_+)$ and $f\in \mathcal B(\bR_+)$.
Then, for any $x\in\bR_+$ and $f\in \mathcal B(\bR_+)$ we define
\begin{equation}
	\begin{split}
		&\cT  f (x): = \langle \delta_x \bv  \cT  f\rangle
		\\&:= K_{I, \rA}(x) \cdot \langle \nu_{I, \rA}\bv f\rangle + K_{O, \rA}(x) \cdot \langle \nu_{O, \rA}\bv f\rangle
		\\
		&= \int_{\bR_+} \kappa(x, y) f(y) \beta(dy)\,.
	\end{split}
	\label{TfDef}
\end{equation}
where $\kappa$ is defined in \eqref{pVU}.
$\cT :=  K_{I,\rA} \langle\nu_{I, \rA}  \bv \cdot \rangle  +  K_{O,\rA} \langle\nu_{O,\rA}  \bv \cdot \rangle$ 
is another equivalent writing of the same definition.
In this notation,  it is clear that $\cT f$ is a measurable function (a priori not bounded),
whereas $\delta_x  \cT$ is a positive measure on $\bR_+$.
\smallskip

From  \eqref{RoDef}, we can interpret  $R_0$ as follows
in terms of $\cT^k$:
\begin{align*}
	R_0 = \lim_{k \rightarrow \infty}
	\langle \delta_x \bv  \cT^k  \mathbf{1}\rangle^{1/k},
\end{align*}
where $\mathbf{1}$ is the function uniformly equal to 1.

Indeed, denote by $M^{(u; k)}$
the point measure on $\bR_+$
which points (counted with multiplicity) 
correspond to the sizes 
of the cities that get infected from city $u$ 
after $k$ generations.
By considering the conditional expectation 
with respect to 
the city sizes 
of the first generation 
of cities infected by $u$, one shows the following equality by means of the branching property:
\begin{equation*}
	\bE\Big( \langle M^{(u; 2)}\bv f\rangle \bv M^{(u; 0)} = x\Big)
	= \bE\Big( \langle M^{(u; 1)}\bv \cT f\rangle 
	\bv M^{(u; 0)} = x\Big) = \langle \delta_x \bv \cT^2 f\rangle.
\end{equation*}
By induction,
for any $k\in \mathbb N$,
$\cT^k$ corresponds to the expectation 
of $M^{(u; k)}$, i.e.:
\begin{equation}
	\langle \delta_x \bv \cT^k f\rangle 
	= \bE\Big( \langle M^{(u; k)}\bv f\rangle \bv M^{(u; 0)} = x\Big).
	\label{Muk}
\end{equation}
Evaluated for $f= \mathbf{1}$, \eqref{Muk}
entails that $\langle \delta_x \bv  \cT^k  \mathbf{1}\rangle
= \bE_x[X_k]$.
Recalling \eqref{RoDef},
this concludes our claim that $R_0 = \langle \delta_x \bv  \cT^k  \mathbf{1}\rangle^{1/k}$.
\medskip

If the kernel $\kappa$ would be symmetric  and 
would satisfy 
\begin{equation}\label{k2}
	\int_{\bR_+}\int_{\bR_+} k(x, y)^2\, \beta(dx) \beta(dy)
	< \infty,
\end{equation}
then our definition would coincide with the definition of $R_0$ as in \cite{BJR06} as 
$$\sup\{ \|\cT f\|_{L^2(\beta)} \,;\; f \in L^2(\beta), \|f\|_{L^2(\beta)}\le 1\}.$$
This follows from the following Proposition \ref{Spec} and essentially Lemma 5.15 in \cite{BJR06}.  Condition \eqref{k2} ensures the compactness of the operator,
as proved in Lemma 5.15 of \cite{BJR06}
and noted in Remark 3.12 of \cite{CO20}.
Without symmetry nor condition \eqref{k2},
$R_0$ is defined as the spectral radius of $\cT$, see in particular Theorem 3.10 in \cite{CO20}.

The spectral radius of $\cT$
is identified in Proposition \ref{Spec}. 
Our rank-two kernel has exactly two real eigenvalues values, the leading eigenvalue is $R_0$. 

\subsubsection{Spectral analysis of $\cT$}
\label{sec_long_time_T}

To analyse the long time behavior of $\cT$,
the following lemma is helpful
in that it relates $\cT$ 
to a matrix operation on a two-dimensional space.
For this, we recall the definition in \eqref{MatrixTo}
of the matrix $\Wmat$.

\begin{lem}
	\label{Proj}
	For any $k\in \mathbb N$, $x\in\bR_+$ and $f$ a non-negative measurable function:
	\begin{align*}
		\langle \delta_x \bv \cT^k  f\rangle
		= \begin{pmatrix}
			K_{I, \rA}(x) & K_{O, \rA}(x)
		\end{pmatrix}
		\cdot \Wmat^{k-1}
		\cdot \begin{pmatrix}
			\langle \nu_{I, \rA} \bv f\rangle\\
			\langle \nu_{O, \rA} \bv f\rangle
		\end{pmatrix}.
	\end{align*}
\end{lem}
\begin{proof}
	The equality follows by induction 
	and the definition of $\cT$ given in \eqref{TfDef}.
\end{proof}

The projection property given in Lemma \ref{Proj}
is strongly connected to the projection
of $\text{V}^\infty$ onto $\text{V}^2$
given in Proposition \ref{V2}.
It greatly simplifies the spectral analysis of $\cT$,
as we can see thanks to the following proposition.
\\

If any of the entries of $\Wmat$ is infinite,
then $\cT^2 \mathbf{1} \equiv \infty$
because 
\begin{multline*}
	\langle \delta_x \bv \cT^2 \mathbf{1} \rangle
	= K_{I,\rA}(x) (\langle \nu_{I,\rA} \bv K_{, \rA}\rangle + \langle \nu_{I, \rA} \bv K_{O, \rA}\rangle)
	\\+ K_{O, \rA}(x) (\langle \nu_{O, \rA} \bv K_{I, \rA}\rangle + \langle \nu_{O,\rA} \bv K_{O,\rA}\rangle).  
\end{multline*}
This directly implies that $R_0 = \infty$ in this case.
Therefore, in the following, we assume that all entries are finite.

If $\beta$ is a power-law distribution with exponent $\phi$,
this assumption translates to
\begin{equation}
	1+ 2 b -\phi < -1 \quad \quad \text{ and } \quad \quad  2 a -1 -\phi < -1,
	\label{strU_reg}
\end{equation}
in the case of strategy~\((U)\) and to
\begin{equation}
	2+ 2 b -\phi < -1 \quad \quad \text{ and } \quad \quad  2 a -\phi < -1
	\label{strP_reg}
\end{equation}
in the case of strategy~\((P)\).
Recall that a bounded measurable function $f$
is called an eigenfunction of $\cT$ 
if there exists some value $\lambda$ such that $\cT f = \lambda f$.
Similarly, a signed measure $\mu$
is called an eigenmeasure of $\cT$ 
if there exists some value $\lambda$ such that $\mu \cT = \lambda \mu$.
$\lambda$ is then called an eigenvalue of $\cT$.
If an eigenmeasure $\mathfrak{q}$ of $\cT$
is a probability distribution,
it satisfies the property 
for being a quasi-stationary distribution (QSD),
namely $\mathfrak{q}\cT (dx)/\langle \mathfrak{q}\cT \mathbf{1}\rangle = \mathfrak{q}(dx).$
Denote by $\mathbf{v}^T$ the transposition of a vector $\mathbf{v}$.
We have the following relationships between the eigenvalues and eigenvector of $\Wmat$ and the eigenvalues and eigenmeasures of $\cT$.  

\begin{prop}
	\label{Spec}
	Assume that the entries of $W$ are finite.
	Then $\cT$ has two distinct and real eigenvalues $\lambda_0$ and $\lambda_1$, such that 
	$\lambda_0>\lambda_1 \vee 0$  and
	that coincide with the ones of $\Wmat$.
	The leading eigenmeasure $\mathfrak{q}_0$ of $\cT$ 
	can be chosen 
	as a probability measure, thus as a QSD.
	Similarly,
	the leading left eigenvector of $\Wmat$ 
	can be chosen as $\begin{pmatrix}
		q_0^I & q_0^O
	\end{pmatrix}$
	such that $q_0^I+ q_0^O=1$, $q_0^I\wedge q_0^O\ge 0$.
	The following relation holds between them:
	\begin{equation}
		\mathfrak{q}_0 = q_0^I\cdot \nu_{I,\rA} + q_0^O\cdot \nu_{O, \rA}.
		\label{qMFq}
	\end{equation}
	On the other hand, 
	the leading eigenfunction $h_0$ of $\cT$ 
	can be chosen as a positive measurable function $h_0$
	such that $\langle \mathfrak{q}_0 \bv h_0\rangle =1$.
	Similarly,
	the leading right eigenvector of $\Wmat$ 
	can be chosen as 
	$\begin{pmatrix}
		h_0^I & h_0^O
	\end{pmatrix}^T$
	such that $q_0^I h_0^I+ q_0^O h_0^O =1$,
	$h_0^I\wedge h_0^O\ge 0$.
	The following relationship holds between the function $h_0$ and the vector $\begin{pmatrix}
		h_0^I & h_0^O
	\end{pmatrix}^T$:
	\begin{equation}
		h_0 = \frac{h_0^I}{\lambda_0}\cdot K_{I, \rA} + \frac{h_0^O}{\lambda_0}\cdot K_{O,\rA}.
		\label{hMFh}
	\end{equation}

	In addition, there exists a function $ h_1$, a measure $ \mathfrak{q}_1$
	and constant $C>0$ such that
	we have the following exact result of exponential convergence 
	at rate $\lambda_1/\lambda_0$:
	\begin{equation}
		(\lambda_0)^{-k}\langle \delta_x \bv \cT^k  f\rangle
		- h_0(x)\cdot \langle \mathfrak{q}_0 \bv f\rangle
		= (\lambda_1/\lambda_0)^k\cdot 
		h_1(x)\cdot \langle \mathfrak{q}_1\bv f\rangle,
		\label{SDT}
	\end{equation}
	with the following bounds:
	\begin{equation*}
		|h_1|\le C\cdot ( K_{I,\rA}\vee K_{O, \rA}),\quad
		| \mathfrak{q}_1|(dx) \le C\cdot ( \nu_{I, \rA}\vee \nu_{O, \rA})(dx).
	\end{equation*}
	In particular, 
	\begin{align*}
		R_0 = \lim_k \langle \delta_x \bv \cT^k  \mathbf{1}\rangle^{1/k}
		= \lambda_0.
	\end{align*}
\end{prop}

Before we proceed with the proof of Proposition \ref{Spec},
let us interpret the quantities $h_0$ and $\mathfrak{q}_0$.

It follows from Proposition \ref{Spec} 
that $\bE_x[X_k]=\langle \delta_x \bv \cT^k  \mathbf{1}\rangle$ is asymptotically equivalent 
to $(R_0)^k\cdot h_0(x)$.
This property is 
why it is reasonable to call $h_0$ the survival capacity 
and what makes $h_0$ relevant 
as an indicator of network centrality, 
see Section \ref{EigC}.
Furthermore, it motivates to consider the eigenmeasure $\mathfrak{q}_0$  because $\mathfrak{q}_0$ is involved in the normalization condition on $h_0$.

Actually, the density of $\mathfrak{q}_0$ at value $x$
can be interpreted 
as the likelihood for a city
randomly chosen 
among the infected cities at generation $g$
to have size $x$,
for large $g$ in the branching approximation.
More precisely, 
for any initially infected city $u$
and on the event of survival of the branching approximation,
we can demonstrate that 
the sequence of normalized random measures
$M^{(u; g)}(dx)
/\langle M^{(u; g)}\vert \mathbf{1}\rangle$
converges to $\mathfrak{q}_0(dx)$
as $g$ tends to infinity.
In the above expression, we recall that $M^{(u; g)}$ is
the point measure on $\bR_+$
which points (counted with multiplicity) 
correspond to the sizes 
of the cities that get infected from city $u$ 
after $g$ generations.

Indeed, 
the city sizes at  generation $g$
are prescribed by independent sampling with distributions $\nu_{I, \rA}$ and $\nu_{O, \rA}$
resp.,
conditionally on the numbers of cities infected from the inside and of those infected from the outside, at generation
$g$.
According to Proposition \ref{V2},
these two numbers can be inferred by studying the process $\textbf{V}^2$.
In such setting of a discrete-type Galton-Watson processes,
\cite{KS66} provides a description of the relative proportions of the different types.
In our case, it implies that 
asymptotically a proportion $q_0^I$ of cities are infected from the inside.
We also recall that on the event of survival, the number of infected cities at generation $g$ tends a.s. to infinity with $g$.
We then 
account for the next sampling of cities sizes,
and recall that 
$\mathfrak{q}_0(dx)
= q_0^I \nu_{I, \rA}(dx)
+ (1-q_0^I) \nu_{O, \rA}(dx)$
to deduce the above claim of convergence thanks to the law of large number.\\

The study of the backward-in-time process is analogous.
Recall that the
eigenvalues of an adjoint operator are complex conjugates of the eigenvalues of the original operator.
In our case, the eigenvalues of the backward operator $\hat{\cT}$
simply coincide with the ones of $\cT$,
namely $R_0=\lambda_0$ and $\lambda_1$.
On the other hand, the QSDs and the eigenvectors do not generally
coincide, except in the specific case of the  strategy~$P$.

\paragraph{Proof of Proposition \ref{Spec}}
Since $\Wmat$ has positive entries,
the Perron-Frobenius theorem  ensures that 
$\Wmat$ has two distinct real eigenvalues $\lambda_0, \lambda_1$
of the form $\lambda_0>\max\{\lambda_1, 0\}$.
Also, the leading left and right eigenvectors have necessarily entries of the same sign, 
contrary to the corresponding eigenvectors of the second eigenvalue.

Assume that $\mathfrak{q}$ 
is an eigenmeasure of $\cT$
with eigenvalue $\lambda$
(when $\cT$ is treated as an adjoint operator on measures).

Under $\cT$, any non-negative measure $\mu$ such that
$\langle \mu \bv K_{O, \rA}\rangle$
and $\langle \mu \bv K_{I, \rA}\rangle$
are both finite
is mapped to a measure 
uniquely prescribed as a linear combination of $\nu_{O, \rA}$ and $\nu_{I, \rA}$.
If either $\langle \mathfrak{q} \bv K_{I, \rA}\rangle$ or $\langle \mathfrak{q} \bv K_{O, \rA}\rangle$
would be infinite, then $\mathfrak{q} \cT = \infty$,
which would contradict the fact that $\mathfrak{q}$ 
is an eigenmeasure of $\cT$. This implies that $\mathfrak{q}$ 
can necessarily be expressed as follows:
\begin{align*}
	\mathfrak{q}= \begin{pmatrix}
		q^O & q^I
	\end{pmatrix}
	\cdot \begin{pmatrix}
		\nu_{O, \rA}\\ \nu_I
	\end{pmatrix}.
\end{align*}
From this representation
and since $\nu_{O, \rA}$ and $\nu_{I, \rA}$ are not colinear, 
with Lemma~\ref{Proj} it follows that   
$\mathfrak{q} \cT = \lambda \mathfrak{q}$ is equivalent to 
$\begin{pmatrix}
	q^O& q^I
\end{pmatrix}$
being a left eigenvector of $\Wmat$ with eigenvalue $\lambda$. 
In particular,
recalling the Perron-Frobenius theorem
the leading left eigenvector
$\begin{pmatrix}
	q^O& q^I
\end{pmatrix}$ of $W$
and the leading eigenmeasure $\mathfrak{q}_0$ of $\cT$
(with leading eigenvalue $\lambda_0$)
can be chosen to be non-negative.
By assuming further that
$q_0^O$ and  $q_0^I$ sum up to one and that $\mathfrak{q}_0$
is a probability measure
and since $\nu_{O, \rA}$ and $\nu_{I, \rA}$
are probability measures,
relationship \eqref{qMFq}
is fulfilled.

Similarly, 
under $\cT$, any non-negative function $f$ such that
$\langle \nu_{O, \rA} \bv f\rangle$
and $\langle \nu_{I,\rA} \bv f\rangle$
are both finite
is mapped to a function 
uniquely prescribed as a linear combination of $K_{O, \rA}$ and $K_{I,\rA}$.
Let $\mathfrak{h}$ be an eigenfunction of $\cT$.
If either $\langle \nu_{O, \rA} \bv \mathfrak{h}\rangle$ or $\langle \nu_{I, \rA} \bv \mathfrak{h}\rangle$
would be infinite, then $ \cT \mathfrak{h} = \infty$
which would contradict the fact that $\mathfrak{h}$ 
is an eigenfunction of $\cT$. This implies that $\mathfrak{h}$ 
can necessarily be expressed as follows:
\begin{align*}
	\mathfrak{h}= \begin{pmatrix}
		h^O & h^I
	\end{pmatrix}
	\cdot \begin{pmatrix}
		K_{O, \rA}\\ K_{I, \rA}
	\end{pmatrix}.
\end{align*}
From this representation
and since $K_{O, \rA}$ and $K_{I, \rA}$ are not colinear, 
with Lemma~\ref{Proj} it follows that   
$ \cT \mathfrak{h} = \lambda \mathfrak{h}$ is equivalent to 
$\begin{pmatrix}
	h^O & h^I
\end{pmatrix}^T$ 
being a right eigenvector of $\Wmat$ with eigenvalue $\lambda$. 
Recalling Perron-Frobenius theorem, 
the leading eigenfunction $h_0$ of $\cT$ 
and the leading right eigenvector $\begin{pmatrix}
	h_0^O & h_0^I
\end{pmatrix}^T$ 
can be chosen non-negative 
and such that 
$\langle \mathfrak{q}_0 \bv h_0\rangle =1$
and $q_0^O h_0^O+ q_0^I h_0^I =1$.
Let us denote by $c>0$ the constant such that $h_0(x) = c \cdot (h_0^O K_{O, \rA}(x) + h_0^I K_{I, \rA}(x))$.
Then, with Lemma \ref{Proj},
$\langle \mathfrak{q}_0 \bv h_0\rangle =1$
translates into 
\begin{equation*}
	c \cdot
	\begin{pmatrix}
		q_0^O & q_0^I
	\end{pmatrix}
	\cdot W \cdot 
	\begin{pmatrix}
		h_0^O\\ h_0^I
	\end{pmatrix}
	= 1.
\end{equation*}
Recalling that $\begin{pmatrix}
	h_0^O & h_0^I
\end{pmatrix}^T$ 
is a right eigenvector of $\Wmat$ with eigenvalue $\lambda_0$ and that $q_0^O h_0^O+ q_0^I h_0^I =1$,
it implies that $c = 1/\lambda_0$
so that relation \eqref{hMFh}
fulfilled. 

Finally, let $\begin{pmatrix}
	q^O_1& q^I_1
\end{pmatrix}$
be a left eigenvector of $\cT$
corresponding to the eigenvalue $\lambda_1$.
It must satisfy that $q^O_1\cdot h^O_0 + q^I_1\cdot h^I_0 = 0$, because when eigenvectors of the different eigenvalues $\lambda_0>\lambda_1$ are considered, we have 
\begin{align*} \lambda_1 \cdot (q^O_1\cdot h^O_0 + q^I_1\cdot h^I_0) 
	&=
	\begin{pmatrix}
		q^O_1& q^I_1
	\end{pmatrix}
	\cdot \Wmat\cdot
	\begin{pmatrix}
		h^O_0\\ h^I_0
	\end{pmatrix}
	\\&= \lambda_0 \cdot (q^O_1\cdot h^O_0 + q^I_1\cdot h^I_0).
\end{align*}
Since $h^O_0$ and $h^I_0$ are positive,
the signs of $q^O_1$ and $q^I_1$ are necessarily different.
Similarly, any right eigenvector of $\Wmat$ 
with eigenvalue $\lambda_1$
has entries of opposite signs.
Therefore, by rescaling appropriately, we can define
$\begin{pmatrix}
	h^O_1& h^I_1
\end{pmatrix}^T$
as the unique left eigenvector of $\Wmat$ such that 
$q^O_1\cdot h^O_1 + q^I_1\cdot h^I_1 = 1$.
Define similarly as for the leading eigenvectors:
\begin{equation}
	\begin{split}
		\mathfrak{q}_1
		&:= q^O_1\cdot \nu_{O, \rA} + q^I_1\cdot \nu_{I,\rA}
		\\
		h_1
		&:= \frac{h^O_1}{\lambda_1}\cdot K_{O, \rA} + \frac{h^I_1}{\lambda_1}\cdot K_{I, \rA}.
	\end{split}
	\label{mqMFh}
\end{equation}

The spectral decomposition of $\Wmat$ implies that for any reals $\ell_O, \ell_I, f_O, f_I$
and $k\in \mathbb N$:
\begin{align*}
	& \begin{pmatrix}
		\ell_O & \ell_I
	\end{pmatrix}
	\cdot \Wmat^k
	\cdot \begin{pmatrix}
		f_O\\ f_I
	\end{pmatrix}
	\\&\quad
	= \Big[ (\lambda_0)^k
	\cdot (\ell_O h_0^O + \ell_I h_0^I)
	\cdot (q_0^O f_O + q_0^I f_I)
	\\&\qquad + \lambda_1^k
	\cdot (\ell_O h_1^O + \ell_I h_1^I)
	\cdot (q_1^O f_O + q_1^I f_I)
	\Big].
\end{align*}
Combined with Lemma \ref{Proj}, \eqref{qMFq}, \eqref{hMFh} and \eqref{mqMFh},
it directly entails \eqref{SDT}
and concludes the proof of Proposition \ref{Spec}
with $C:= \max\{q_1^O, q_1^I, h_1^O, h_1^O\}$.
$\hfill \square$

\subsection{Eigenvector centrality}
\label{sec_eig_cent}

\subsubsection{Notion of eigenvector centrality}
\label{EigC}

In terms of the approximation by the forward in time branching process,
one can get explicit formulas 
for a classical notion of  centrality
for epidemics evolving on a network, 
namely eigenvector centrality, which assigns to each city size $x$ the value of the leading eigenfunction $h_0(x)$, see \cite[Chapter 7]{N10}, \cite{Bo87}.   It is a measure for the number of infections that are triggered when node $x$ gets infected and hence, allows to compare the relative importance of the different cities during the course of an epidemic. 
The value is primarily of interest in situations where the goal of containing the epidemic is no longer within reach and the aim is instead to delay its progression. By targeting strict measures on cities with a high eigenvector centrality value, one wishes to target cities with a high potential for additional infections.

According to Proposition \ref{Spec}, the leading eigenfunction $h_0$
takes the following form:
\begin{equation}
	h_0(x) = (R_0)^{-1}\cdot 
	(K_{O, \rA}(x)\cdot h_0^O + K_{I, \rA}(x)\cdot h_0^I),
	\label{hX}
\end{equation}
where the $2\times 1$ vector $(h_0^O\; h_0^I)^T$
is the eigenvector of the $2\times 2$ transmission matrix $\Wmat$. Since $\Wmat_{O, O} = \Wmat_{I, I}$ we arrive at the following expression for 
$(h_0^O\; h_0^I)^T$
(with similar expressions for  $(\mathfrak{q}_0^O\; \mathfrak{q}_0^I)$, the left eigenvector of $\Wmat$):
\begin{equation}
	\begin{split}
		h_0^O &:= \dfrac{\sqrt{\Wmat_{O, I}}+ \sqrt{\Wmat_{I, O}}}{2\sqrt{\Wmat_{I, O}}},
		\\
		h_0^I &:= \dfrac{\sqrt{\Wmat_{O, I}}+ \sqrt{\Wmat_{I, O}}}{2\sqrt{\Wmat_{O, I}}},
	\end{split}
	\label{hAhB}
\end{equation}
With these expressions one observes
that the normalisation is such that $(h_0^O\; h_0^I)^T$
is actually independent of the value of the thresholds
$p_\vee$ and $L_\vee$ under strategy~\((P)\) and \((U)\), respectively.

The eigenvector centrality value $h_0(x)$ can be decomposed
into the two factors $f_O(x)=(h_0^O/R_0)\cdot 
K_{O, \rA}(x)$ and  $f_I(x)= (h_0^I/R_0) \cdot K_{I, \rA}(x)$.
Since by definition $K_{O, \rA}$ and $K_{I,\rA}$ (cf~\eqref{Knu}) depend polynomially on $x$, both factors are on a log scale linear functions in the city size.

\begin{figure}
	\begin{center}
		Eigenvector centrality
	\end{center}
	\begin{tabular}{lcc}
		&\quad Strategy \((P)\) & \quad Strategy \((U)\)\\
		\rotatebox{90}{\hcm{1.6} France}
		&  \includegraphics[draft = false, width = 0.44\textwidth]{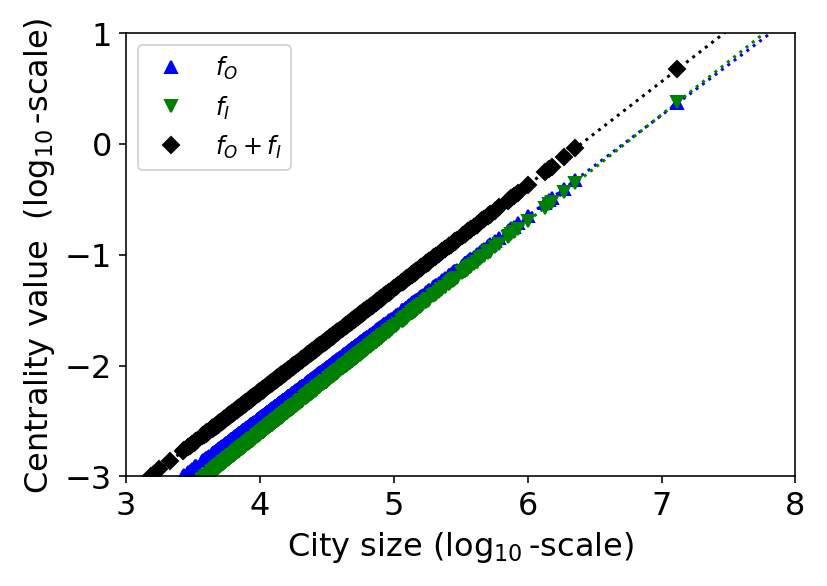}   
		& \includegraphics[draft = false, width = 0.44\textwidth]{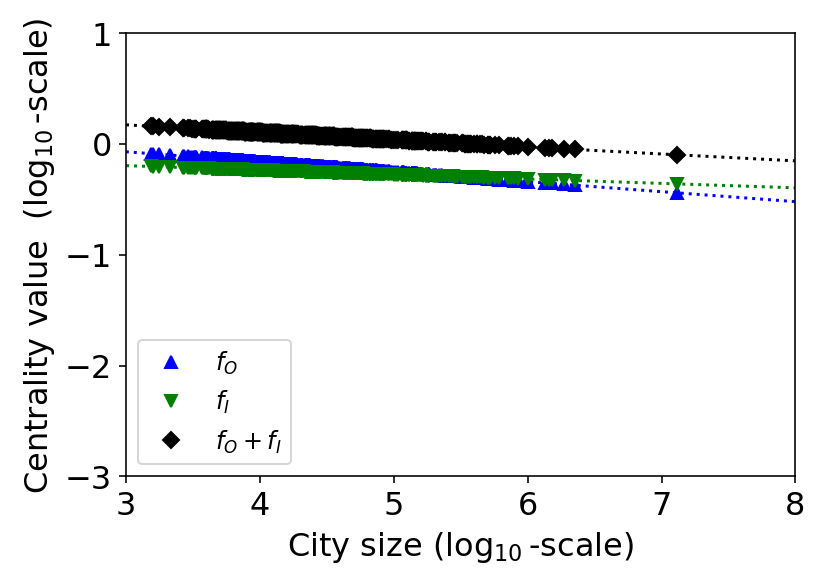} \\
		\rotatebox{90}{\hcm{1.6} Poland}
		&\includegraphics[draft = false, width = 0.44\textwidth]{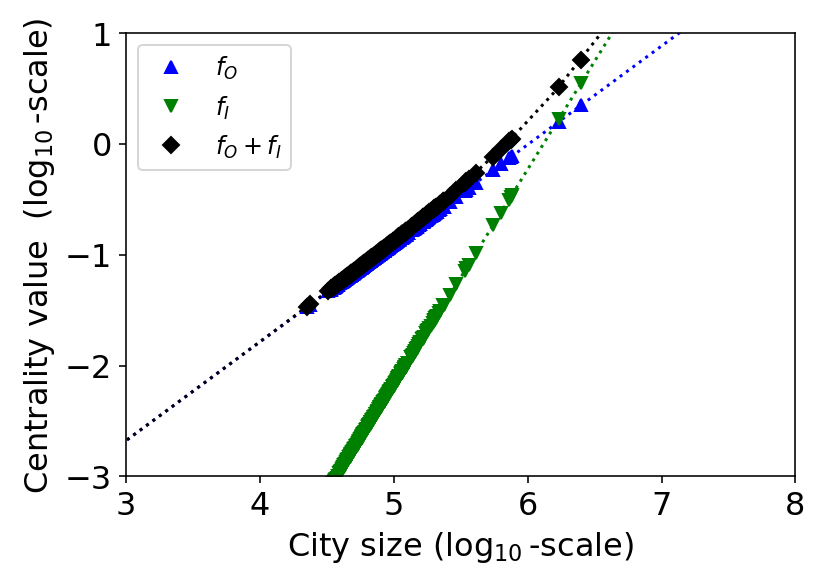} 
		& \includegraphics[draft = false, width = 0.44\textwidth]{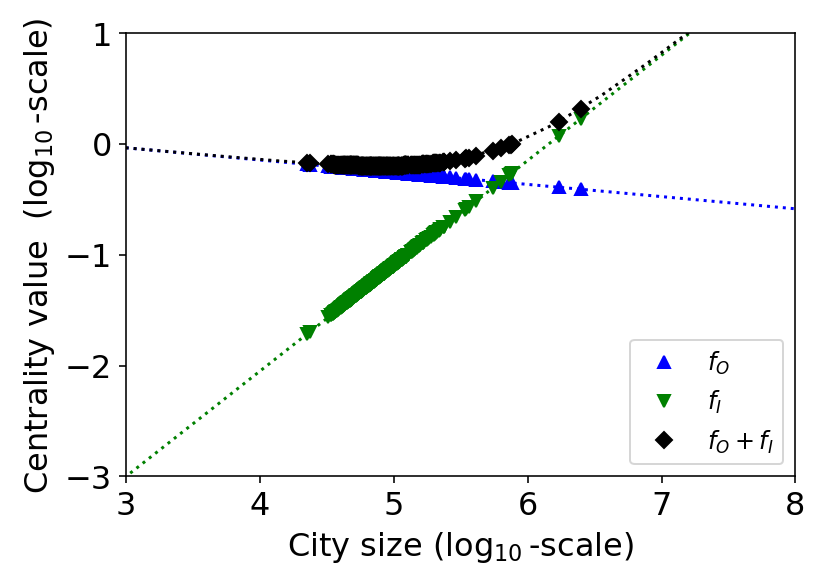}\\
		\rotatebox{90}{\hcm{1.6} Japan}
		&\includegraphics[draft = false, width = 0.44\textwidth]{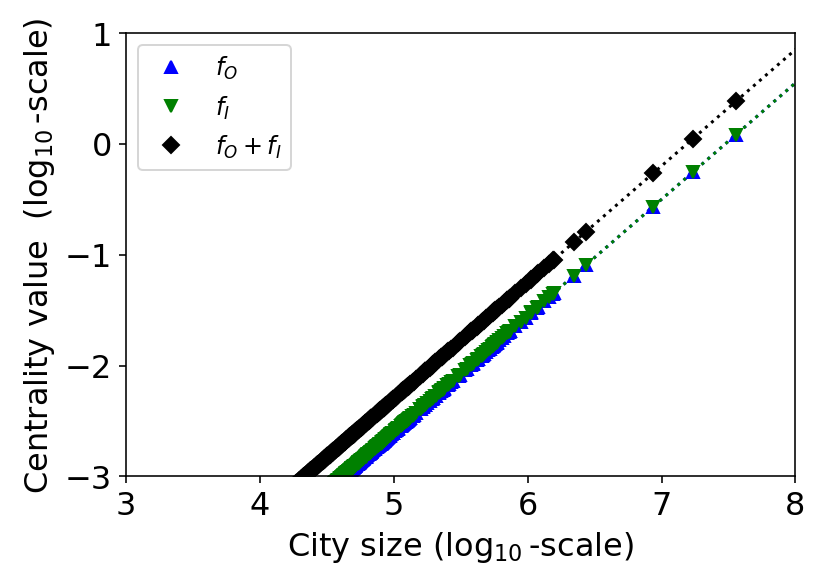} 
		& \includegraphics[draft = false, width = 0.44\textwidth]{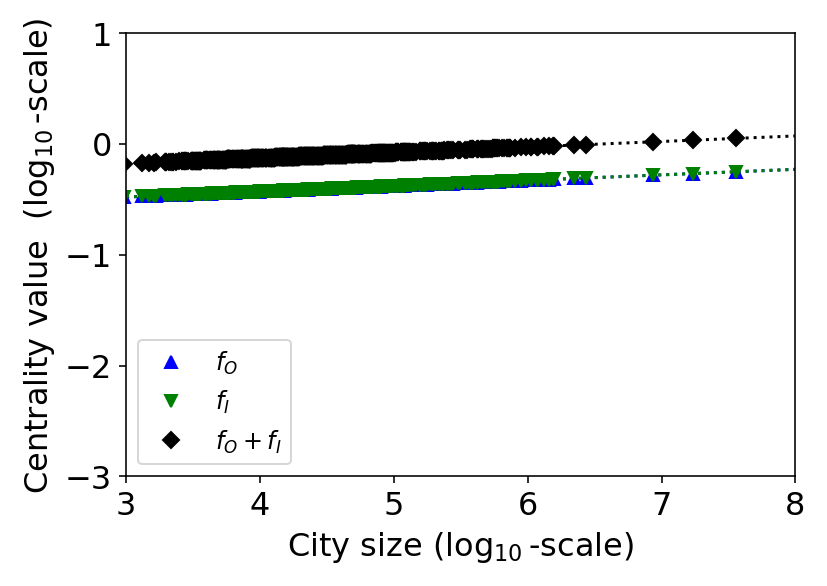}
	\end{tabular}
	\caption{
		Decomposition of eigenvector centrality
		into the additive components $f_O$ and $f_I$
		depending on city size
		in log-log plot,
		on the left under strategy \((P)\)
		and on the right-panel under strategy \((U)\),
		for data from France, Poland and Japan.}
	\label{fig:EigenvectorCentrality}
\end{figure}

\subsubsection{Estimation of eigenvector centrality}
\label{sec_est_eigC}
Complementary to the outbreak probability,
the eigenvector centrality value  provides 
additional insight 
into which cities mainly drive the epidemic.  
For large city sizes
and large enough infection rate,
outbreak probabilities cannot be distinguished.
Eigenvector centralities of these cities are nonetheless still different, 
especially under strategy~\((P)\),
and higher centrality values actually correspond to higher risk factors (figures not shown). This lends support to the incentive 
to restrict the cities according to the full order of size. 

Admittedly, the largest cities are in any case most likely 
to become infected early without strong prior regulation, 
so their contribution to the growth rate should not be sustained for long.  
However, this is all the more reason 
to put in place preventive restrictions 
in order to avoid secondary cases coming out of these cities.

In Section \ref{EigC}, we argued that in our analytical model 
(with a kernel of the form \eqref{Knu}) the eigenvector centrality value is composed of the two factors $f_O$ and $f_I$ which depend on a log scale linearly on the city size. 
In the case of the French best fit kernel, both factors are nearly of equal magnitude,
see Figure~\ref{fig:EigenvectorCentrality}. In contrast, for the Polish best fit kernel, the eigenvector centrality value for large cities is predominantly influenced by the factor \( f_I \), i.e. by infections imported by residents of large cities traveling to other areas and bringing back the disease. Conversely, for small Polish cities (those with fewer than \( 10^5 \) inhabitants), the eigenvector centrality value is primarily driven by the factor \( f_O \), i.e. by infections imported into these smaller cities by visitors arriving from other locations. 

Concerning Japan, the lower-panel displays very similar patterns as for France.
Due to the lower quality of the kernel graph approximation for Japan, 
the exact values are not very reliable, yet the general trend still deserves to be remarked.

\FloatBarrier

As expected,
the heterogeneity in eigenvector centrality 
is much more pronounced for strategy~\((P)\)
than for strategy~\((U)\),
where the centrality value is close to a constant for Poland, slightly decreasing with population size for France,
and slightly increasing for Japan. Since in France the estimated value of $a$ is less than 1 and of $b$ is less than 0,
both attractiveness and emissivity
reduce the role of large cities. Furthermore, under strategy~\((U)\)
the number of infected citizens does
not increase with the city size.

As noted in Subsection \ref{EigC},
eigenvector centrality values are independent of the threshold values,
as would be outdegree with any natural normalisation.
This is in stark contrast with outbreak probability, 
for which the inflection point 
(where the sigmoid function changes  from being concave to convex) 
moves up with increasing threshold values (figures not shown).
Contrary to outdegrees and eigenvector centrality values,
outbreak probability makes visible the distinction 
between cities that are not likely to produce any outbreak
and the others.
On the other hand, 
it might be misleading in suggesting 
that all cities well above the inflection point 
are contributing equally.
For strategy~\((P)\),
this suggestion is efficiently corrected by other measures like
eigenvector centrality or outdegree.

\subsection{Further examination of goodness of approximations and estimates}
\label{sec_factors}

In this section, we investigate the impact of several factors influencing the goodness of the applied approximations and estimates.
In Section~\ref{sec_R0_est}, we discuss how the number of cities, size heterogeneity, and graph structures impact the estimation of the basic reproduction number. Section~\ref{sec_KB_errors} addresses the quality of KB analytical estimations of infection and outbreak probabilities, while Section~\ref{sec_data_filter} is focused on the crucial role of data filtering methods.

\subsubsection{$R_0$ estimation: number of cities, size heterogeneity, graph structures}
\label{sec_R0_est}

Recall from Section~\ref{sec_val_R0}
that the $R_0$ value could not get reliably estimated
from the dynamics in the number of cities infected at each generation (even when we average over replications of outbreaks).
So we investigate the effects 
responsible for this bad quality of the $R_0$ estimation
possibly due to $(i)$ the too small number of cities,
$(ii)$ a tail of city size distributions that is too heavy and
$(iii)$ the specific graph structure.
We check the validity of our proposed procedure for a range of power-law distributions.

\begin{figure}[t]
	\begin{center}
		$R_0$ estimation for extended city size distributions
	\end{center}
	\begin{center}
		\begin{tabular}{lcc}
			&Strategy \((P)\) & Strategy \((U)\)\\
			\rotatebox{90}{\hcm{0.4} Kernel Graph}
			&\includegraphics[width = 0.42\textwidth]{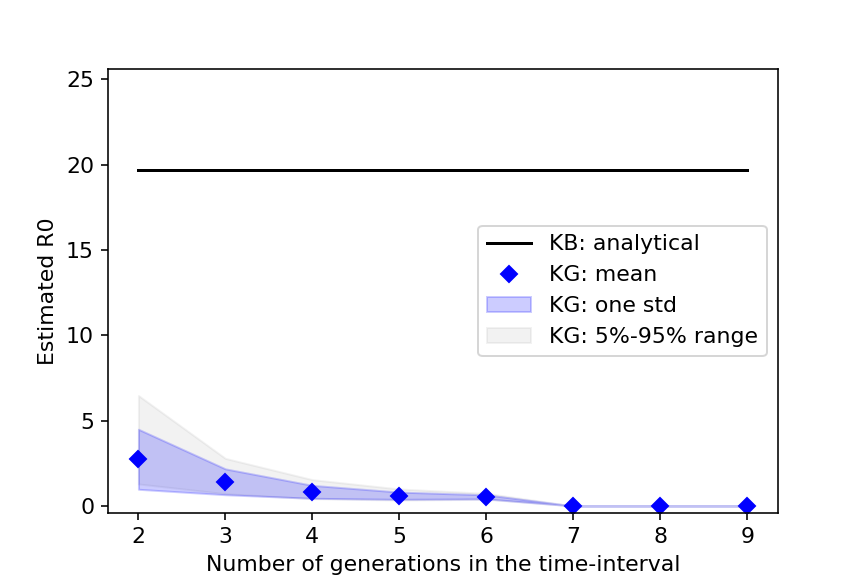}
			&
			\includegraphics[width = 0.42\textwidth]{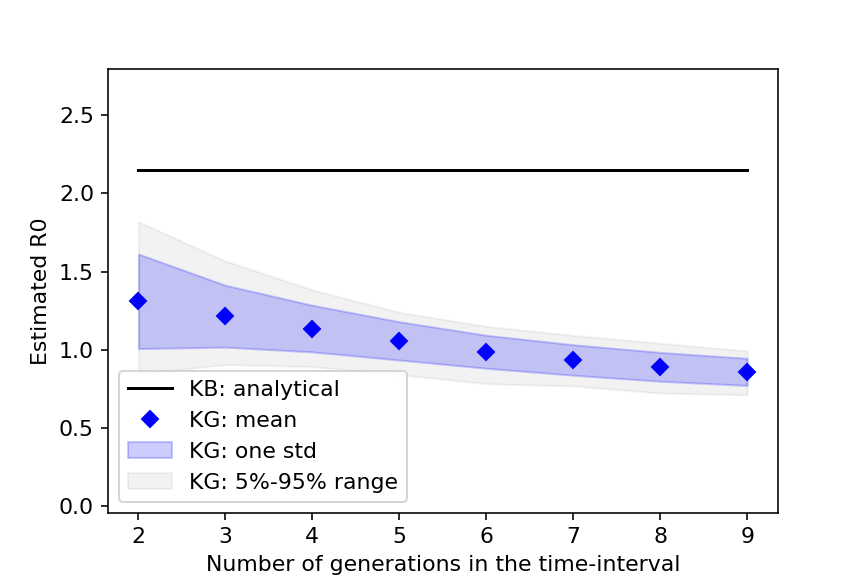}
			\\
			\rotatebox{90}{$10 \times$ Kernel Graph}
			&\includegraphics[width = 0.42\textwidth]{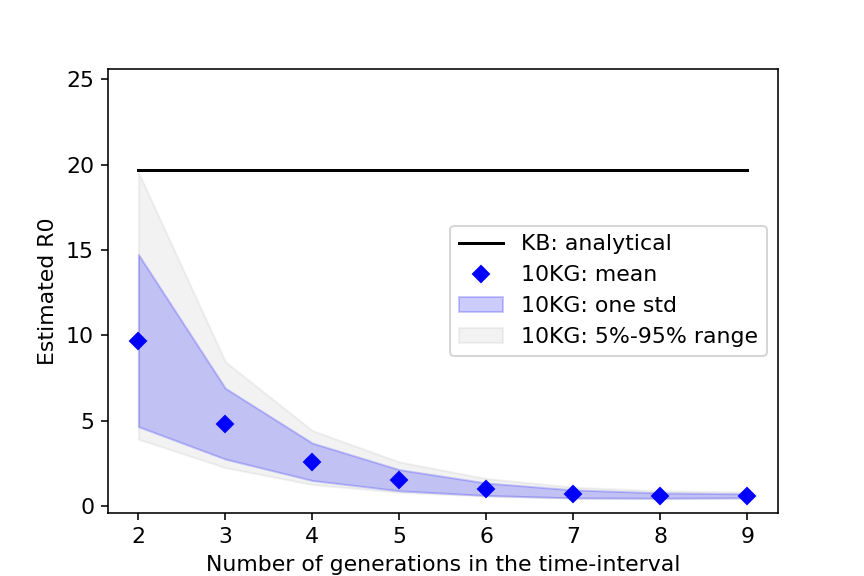}
			&
			\includegraphics[width = 0.42\textwidth]{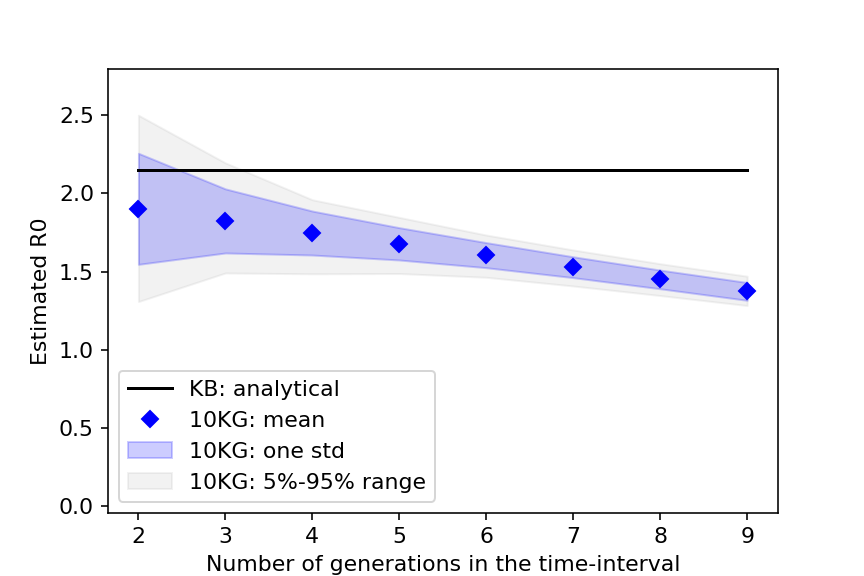}
		\end{tabular}
	\end{center}
	\caption{$R_0$ estimation for the kernel graphs derived from the values of $a$ and $b$ of France (D30+) with the associated city size distributions (KG, top tow) 
		compared to the graph where these vertices are duplicated 10 times
		(10KG, bottom row),  on the left for strategy~\((P)\) and on the right for strategy~\((U)\).}
	\label{R0_10KG}
\end{figure}

	\begin{figure}[t!]
		\begin{center}
			$R_0$ estimation for various power-law distributions
		\end{center}
		\begin{center}
			\begin{tabular}{lcc}
				&Strategy \((P)\) & Strategy \((U)\)\\
				\rotatebox{90}{ {\footnotesize After 10 infected cities}}
				&\includegraphics[width = 0.42\textwidth]{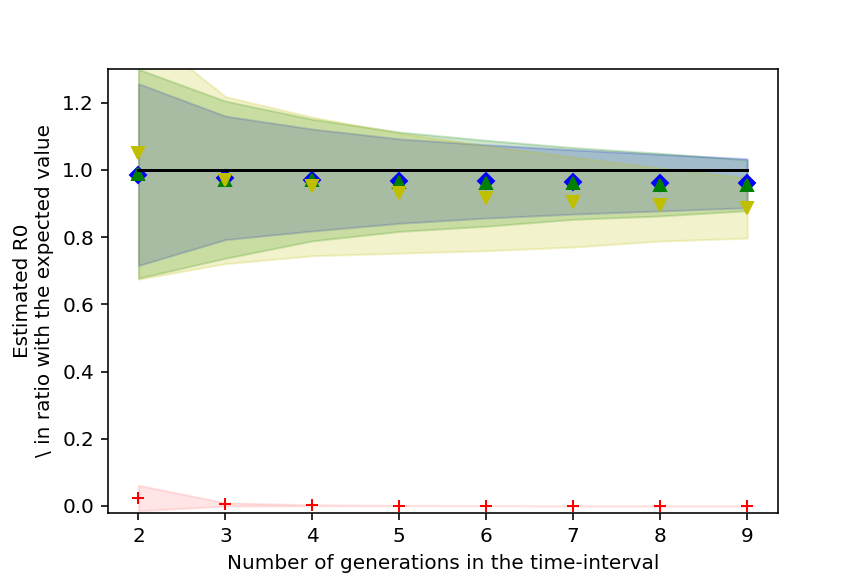}
				&
				\includegraphics[width = 0.42\textwidth]{R0_KG_pI_strP_Comp_Est_varying_phi_T10}
				\\
				\rotatebox{90}{{\footnotesize After 100 infected cities}}
				&\includegraphics[width = 0.42\textwidth]
				{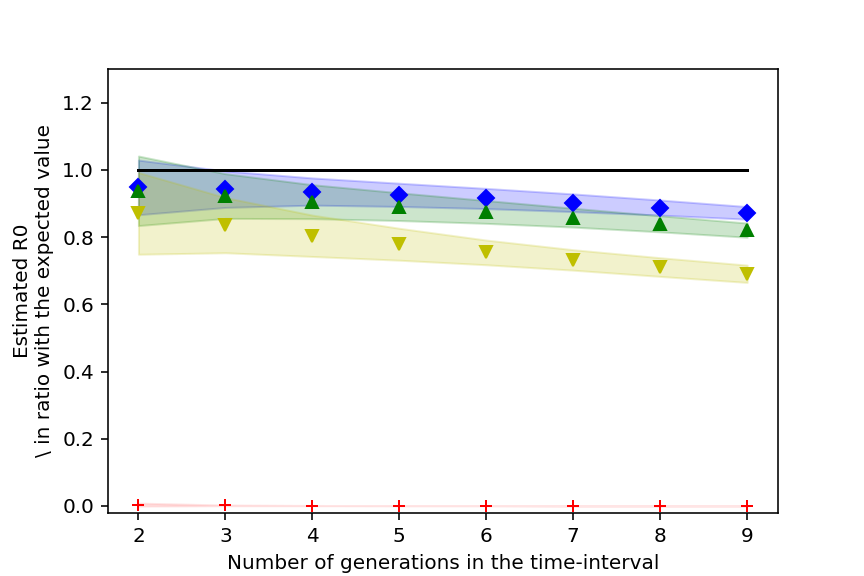}
				&\includegraphics[width = 0.42\textwidth]
				{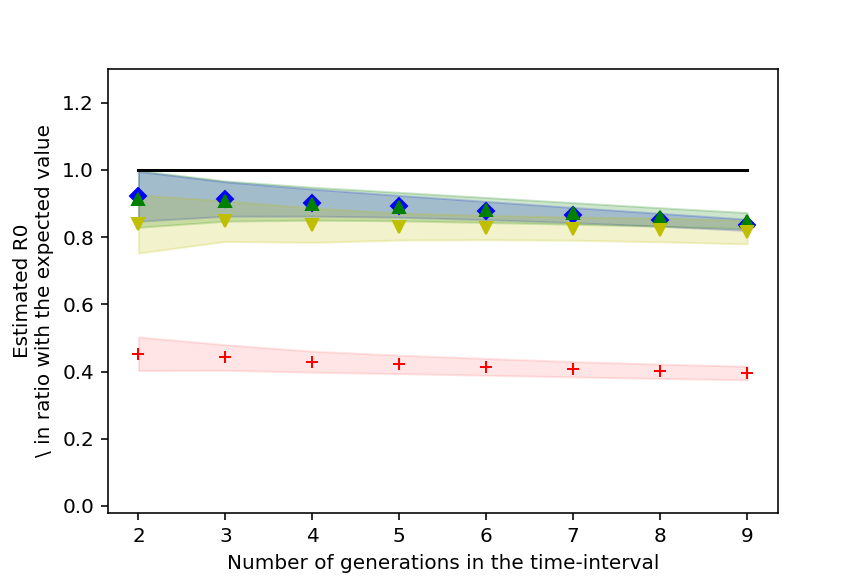}\\
				&\multicolumn{2}{c}{\includegraphics[draft = false, width = 0.80\textwidth]{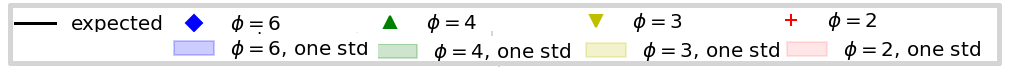}}
			\end{tabular}
		\end{center}
		\caption{$R_0$ estimation for the kernel graphs derived from the values of $a$ and $b$ of France (D30+) comparing city size distributions of $10^4$ cities (of the order of 10KG) given with power-laws of coefficient 2, 3 4 and 6,  on the left for strategy~\((P)\) and on the right for strategy~(U).
			For the top row, the original threshold of 10 cities currently infected is considered, while for the bottom row, this threshold
			is set to 100 cities.}
		\label{R0_varPhi}
	\end{figure}


Concerning (i), we anticipate that large cities are quickly hit and then isolated. This strongly reduces the number of secondary infections. To reduce this effect,
we increase the number of cities by a factor 10 (without changing the city size  distribution),
by replicating each city size ten times. 
For both clarity and computational efficiency, 
we produce a kernel graph structure on this extended dataset,
which we abbreviate as ``10KG'',
and compare in Figure~\ref{R0_10KG}
the results of the $R_0$ estimation between the original kernel graph
and the extended kernel graph.
The $R_0$ values 
is in both cases adjusted such that the theoretical infection probability is 0.5 for cities of size $10^5$, see Section \ref{sec:infection probability},
which leads to identical values between these two graphs 
(and for the $k_B$ values as well).
There is a clear improvement with the extended graph,
yet the estimations come close to the theoretical value
only under strategy~\((U)\) for a few generations (then still with large variations between runs).
Furthermore, 
the decline of the estimated value with the number of generations considered
is more strikingly visible. 
This confirms that the rapid establishment of immunity 
hinders the estimation of $R_0$.

Concerning (i) and (ii) we then consider other size distributions 
of $10,000$ cities (of the order associated to $10KG$)
that are randomly sampled according to power-law distributions with various coefficients, namely $\phi= 2, 3, 4$ and $5$.
In Figure~\ref{R0_varPhi},
the corresponding estimations of $R_0$ are compared
between these different distributions, both under strategy~\((P)\) (on the left)
and under strategy~\((U)\) (on the right-panel),
in the upper-panel 
with the original threshold of $10$ infected cities
(in a single generation, for the start of the estimation interval)
and in the lower-panel with a larger threshold of $100$
infected cities.
As in Figure~\ref{R0_10KG}, the $R_0$ values 
are adjusted such that the theoretical infection probability is 0.5
 for cities of size $10^5$.
Since the corresponding values differ between the various sampled distributions,
the estimations are rescaled by the expected value to ease the comparison.
Furthermore, we set the estimated $R_0$-value to 0 for those generations
for which less than 20 replicates (out of the 200 produced)
keep a persistent outbreak over the whole interval of estimation.

We observe that the quality of $R_0$-estimation 
is also very poor with the most heterogeneous distribution, sampled according to a power-law coefficient of $\phi = 2$,
yet much more suitable for the other distributions.
The fluctuations between replicates are still very large when estimation starts with 10  infected cities. With the larger threshold of 100 infected cities, 
these fluctuations are largely reduced, which improves the quality of estimation 
though the estimator is significantly biased downwards.
The increase in the number of generations produces similar effects as the one of the threshold,
though not as large for the considered values.

Finally, concerning (iii),
the role of the specific graph structure does not appear 
to be as significant
as compared to the crucial role
of the tail distribution,
that we evaluated directly in the simplified kernel graph structure.
Recall that
simplifying the graph structure 
does not lead to a significant improvement 
of the estimation, as can be seen in the comparison between the transportation 
graph and the kernel graph (see Figure~\ref{Fig_R0_est_FC}).

\subsubsection{Quality of the analytical estimations
	of infection and outbreak probabilities}
\label{sec_KB_errors}

We present in Figure~\ref{Fig:Convergence}
the convergence rate of the iterated procedures
towards the infection/outbreak probabilities, see Section~\ref{sec:probability trigger outbreak} and \ref{sec:infection probability}.

\begin{figure}[t]
	\begin{center}
		Convergence of the analytical parameters
	\end{center}
	\begin{center}
		\begin{tabular}{lcc}
			&\quad Strategy \((P)\) & \quad Strategy \((U)\)\\
			\rotatebox{90}{{\small Outbreak probability}} \rotatebox{90}{\hcm{0.7} $\eta^k_O,\; \eta^k_I$}
			&\includegraphics[height = 0.17\textheight]{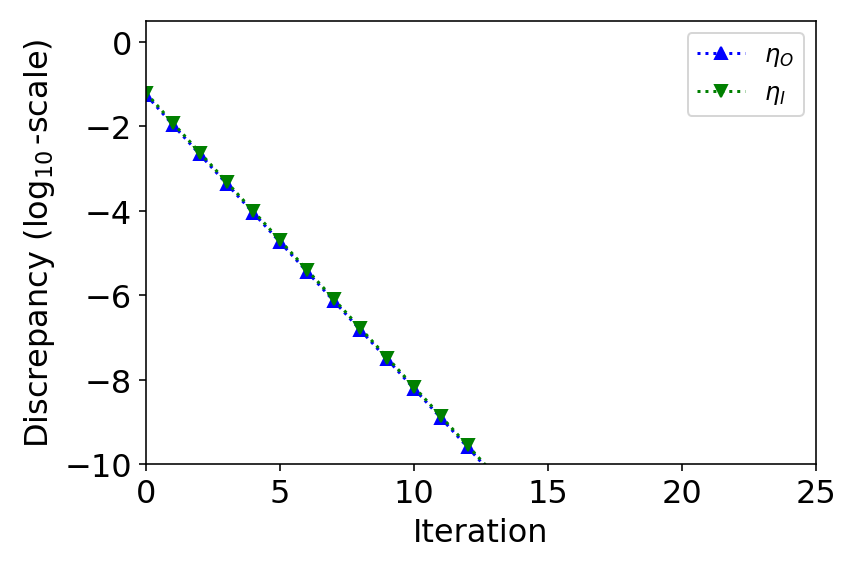}
			&\includegraphics[height = 0.17\textheight]{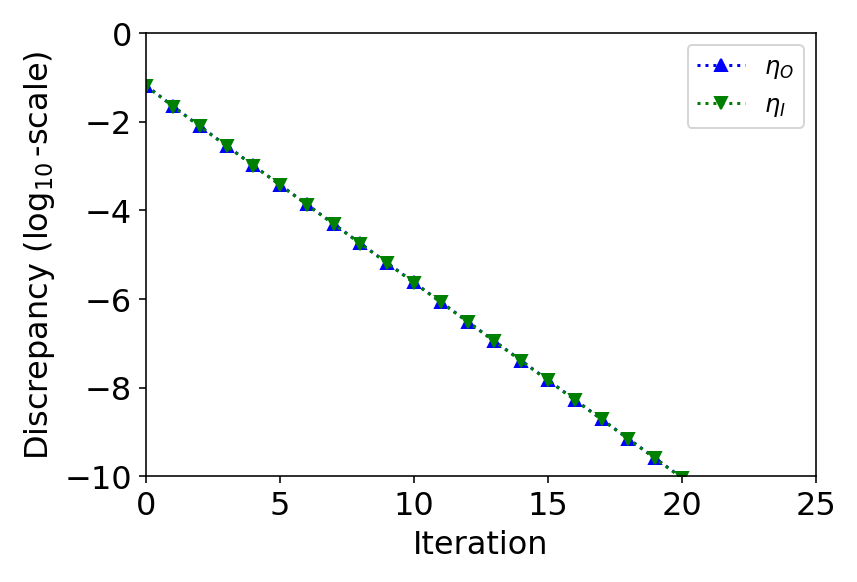}
			\\
			\rotatebox{90}{{\small Infection probability}}
			
			\rotatebox{90}{ \hcm{0.7} $\pi^k_O,\;  \pi^k_I$}
			&\includegraphics[height = 0.17\textheight]{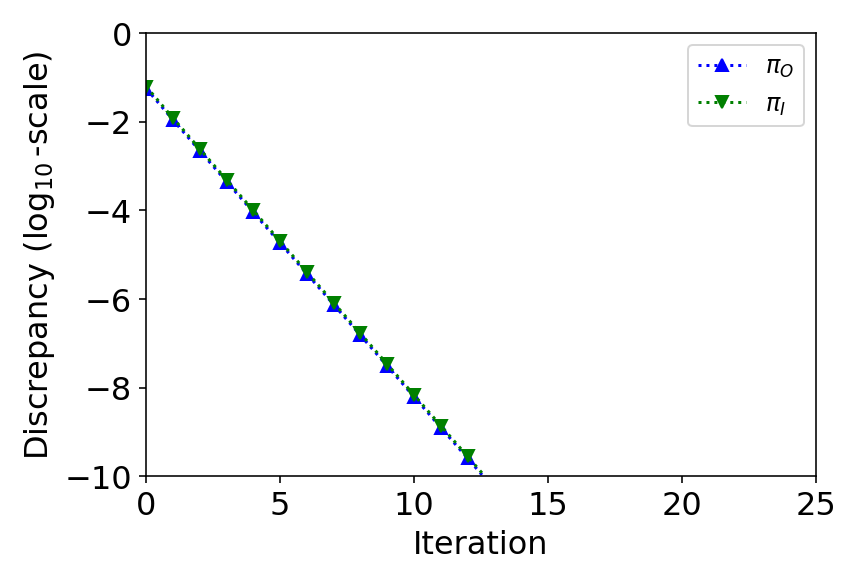}
			&\includegraphics[height = 0.17\textheight]{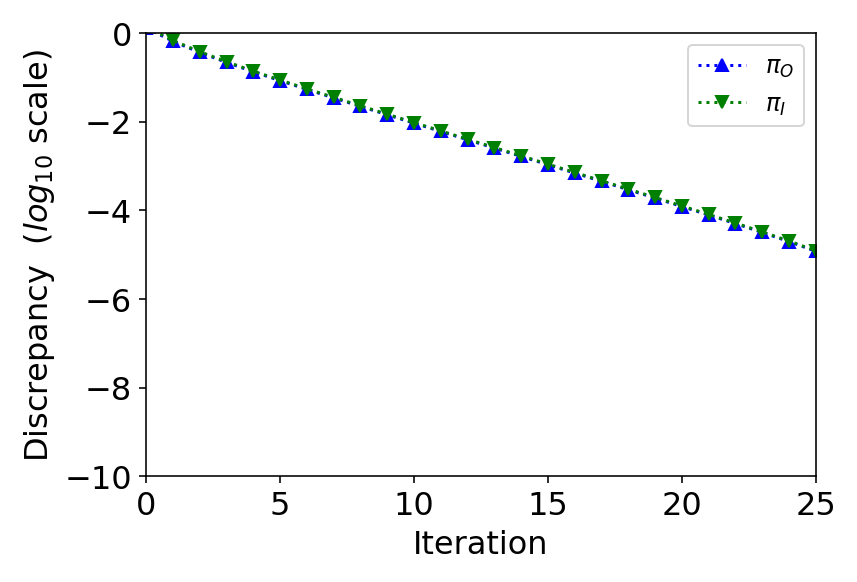}
		\end{tabular}
	\end{center}
	\caption{
		Convergence of $(\eta_O^{k}), (\eta_I^k )$, see \eqref{eta_k_A}, towards $\eta_O$ and $ \eta_I $, see \eqref{eta_A}, as well as of (the analogously defined quantities) $(\pi_O^k)$ and $(\pi_I^k)$ towards  $\pi_O $ and $\pi_I $, see \eqref{piDef}, 
		under strategy~\((P)\) (on the left)
		and under strategy~\((U)\) (on the right-panel). As a city size distribution $\beta$ the empirical size distribution from France (D30+) is chosen.  The upper-panel shows $\log_{10}(|\eta^k_O - \eta^{30}_O|/\eta^{30}_O)$
		as well as the analogous quantity
		relative to $\eta_I $ at the different iterations $k=0,..., 25$, the lower-panel shows the analogous values for $\pi_O$ and $\pi_I$. The constant $k_B$ is chosen such that the theoretical infection probability  (approximated at a large enough iteration) is 0.5. }
	\label{Fig:Convergence}
\end{figure}

According to Lemma 5.6 in \cite{BJR06}, we could estimate $(\eta(x))_x$ iteratively by setting
\begin{equation}\label{eta_k}
	\eta ^{k+1}(x) 
= 1- \exp\Big[  - K_{I,\rA}(x)  \int_0^\infty\eta^{k}(y) \nu_I(dy)
-  K_{O, \rA}(x)\int_0^\infty\eta^{k}(y) \nu_{O, \rA}(dy) \Big]
\end{equation}
and starting with $\eta^0 \equiv 1$.
The functional on the right-hand side is a global contraction 
if certain integrability  conditions are fulfilled, as shown in \cite{BJR06}.

The expression for $\eta$ 
on the r.h.s. of \eqref{survivalprobab} 
is a function of the two unknown parameters
\begin{align}\label{eta_A}
\eta_{I}= \int_0^\infty\eta(y) \nu_{I, \rA}(dy),
\qquad \eta_{O}=\int_0^\infty\eta(y) \nu_{O, \rA}(dy).
\end{align}

We check the performance of the iterative procedure  through the convergence of the associated sequences $(\eta_I^k)_k $ and $(\eta_O^k)_k$
with
\begin{align}
\label{eta_k_A}
    \eta_I^k=\int_0^\infty\eta^k(y) \nu_{I, \rA}(dy),
    \qquad \eta_O^k=\int_0^\infty\eta^k(y) \nu_{O, \rA}(dy),
\end{align} which fulfill the recursions
\begin{align*}
&	\eta^{k+1}_I
	= 1 - \int_0^\infty \exp\Big[  - K_{I, \rA}(x) \cdot \eta^{k}_I
	-  K_{O, \rA}(x) \cdot \eta^{k}_{O}\Big] \nu_{I,\rA}(dx),
	\\
	&\eta^{k+1}_O
	= 1 - \int_0^\infty \exp\Big[  - K_{I, \rA}(x) \cdot \eta^{k}_I
	-  K_{O, \rA}(x) \cdot \eta^{k}_O\Big] \nu_{O, \rA}(dx).
\end{align*}

We also estimate $\pi(x)$ iteratively, in analogy to the estimation procedure for $\eta(x)$.
The performance of the iterative procedure 
is again checked through the convergence of the two summary parameters $\pi_O$ and $\pi_I$, 
i.e. defined with $\pi^{0}_O = \pi^{0}_I = 1$ and at step $k\ge 0$:
\begin{align*}
&	\pi^{k+1}_O
	:= 1- \int_0^\infty \exp\Big[  - K_{O, \lA}(x) \cdot \pi^{k}_O
	-  K_{I, \lA}(x) \cdot \pi^{k}_I\Big] \nu_{O,\lA}(dx) ,
	\\
&	\pi^{k+1}_I
	:= 1- \int_0^\infty \exp\Big[  - K_{O, \lA}(x) \cdot \pi^{k}_O
	-  K_{I, \lA}(x) \cdot \pi^{k}_I\Big] \nu_{I, \lA}(dx).
\end{align*}

For both strategies \((U)\) and \((P)\),
the plots display a very clear trend of convergence (linear in log-scale)
for both $(\eta_O^{k}), (\eta_I^k )$  towards $\eta_O$ and $ \eta_I $, see 
\eqref{eta_k_A} and \eqref{eta_A}, as well as for both (the analogously defined quantities) $(\pi_O^k)$ and $(\pi_I^k)$ towards  $\pi_O $ and $\pi_I $,
see \eqref{piDef}.

\clearpage

\subsubsection{Role of the
	methods for adjusting the virulence parameter and for filtering the datasets}
\label{sec_data_filter}

\begin{figure}
	\begin{center}
		
	\end{center}
	\begin{tabular}{lc@{ \hskip -0.08in }c}
		&\multicolumn{2}{c}{    Mean proportion of people and cities under isolation}\\
		&\quad Strategy \((P)\) & \quad Strategy \((U)\)\\
		\rotatebox{90}{\hcm{1.6} France}
		&  \includegraphics[draft = false, width = 0.44\textwidth]{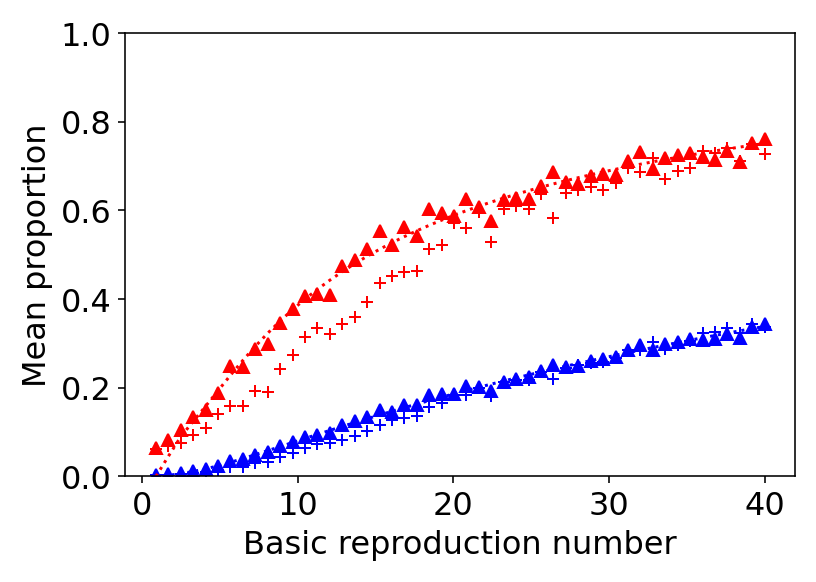}   
		& \includegraphics[draft = false, width = 0.44\textwidth]{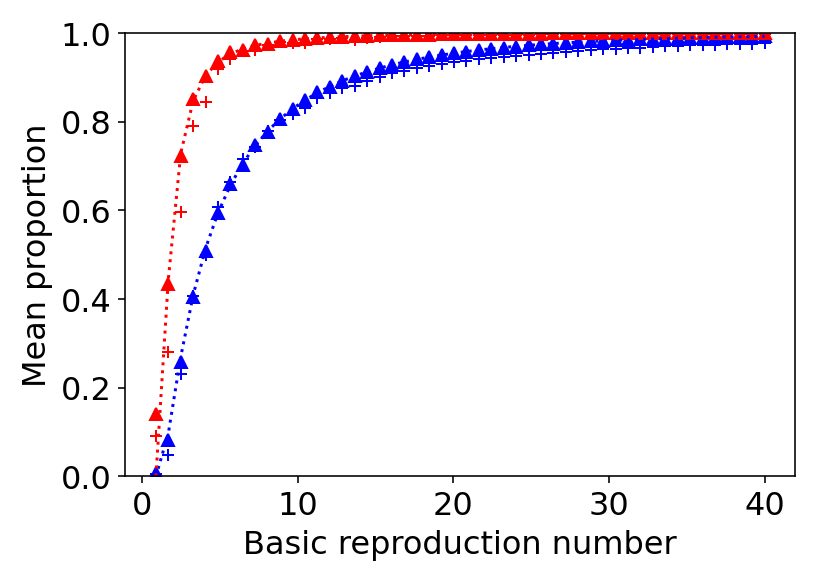} \\
		\rotatebox{90}{\hcm{1.6} Poland}
		&\includegraphics[draft = false, width = 0.44\textwidth]{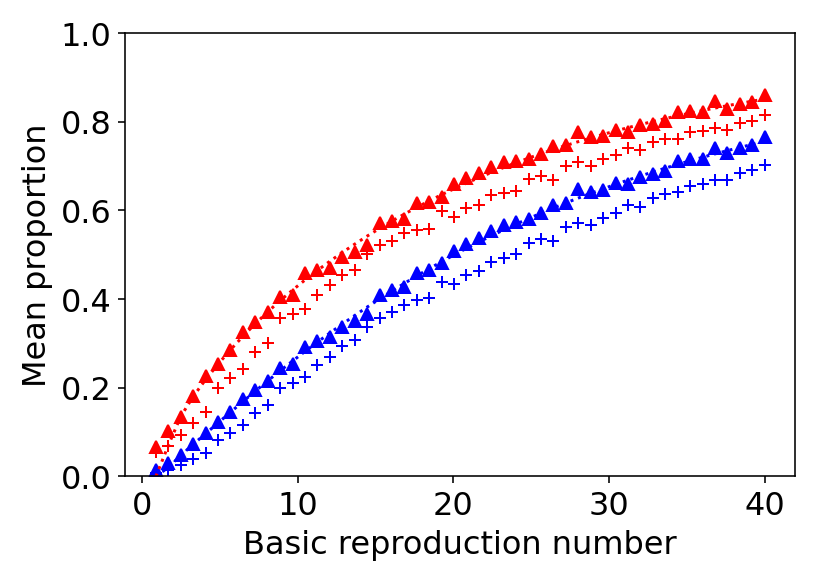} 
		& \includegraphics[draft = false, width = 0.44\textwidth]{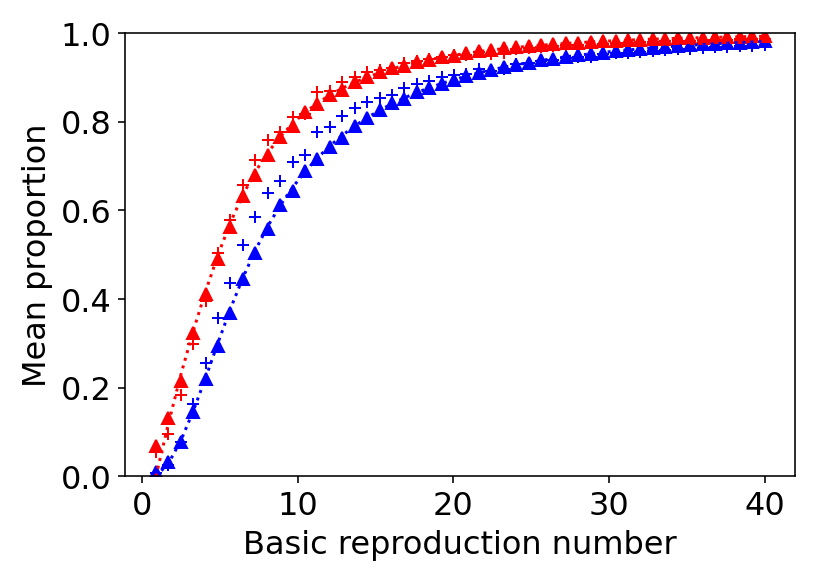}\\
		\rotatebox{90}{\hcm{1.6} Japan}
		&\includegraphics[draft = false, width = 0.44\textwidth]{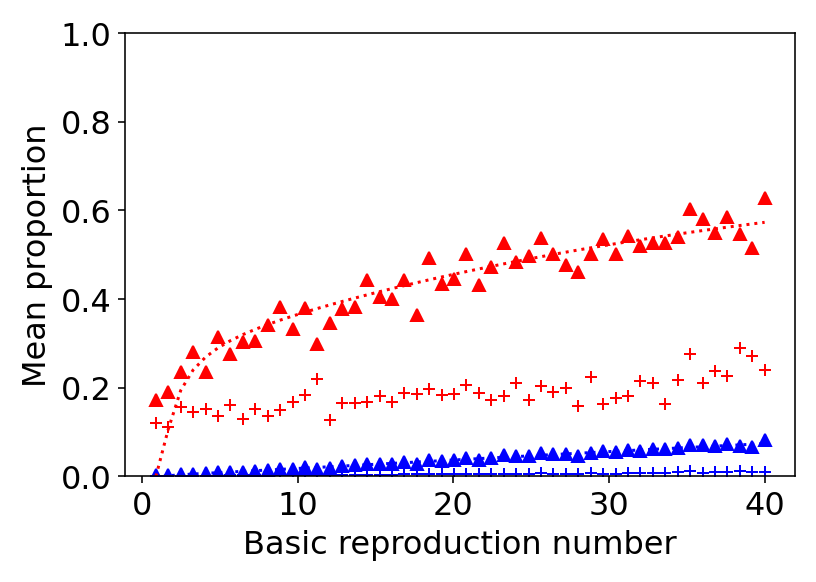} 
		& \includegraphics[draft = false, width = 0.44\textwidth]{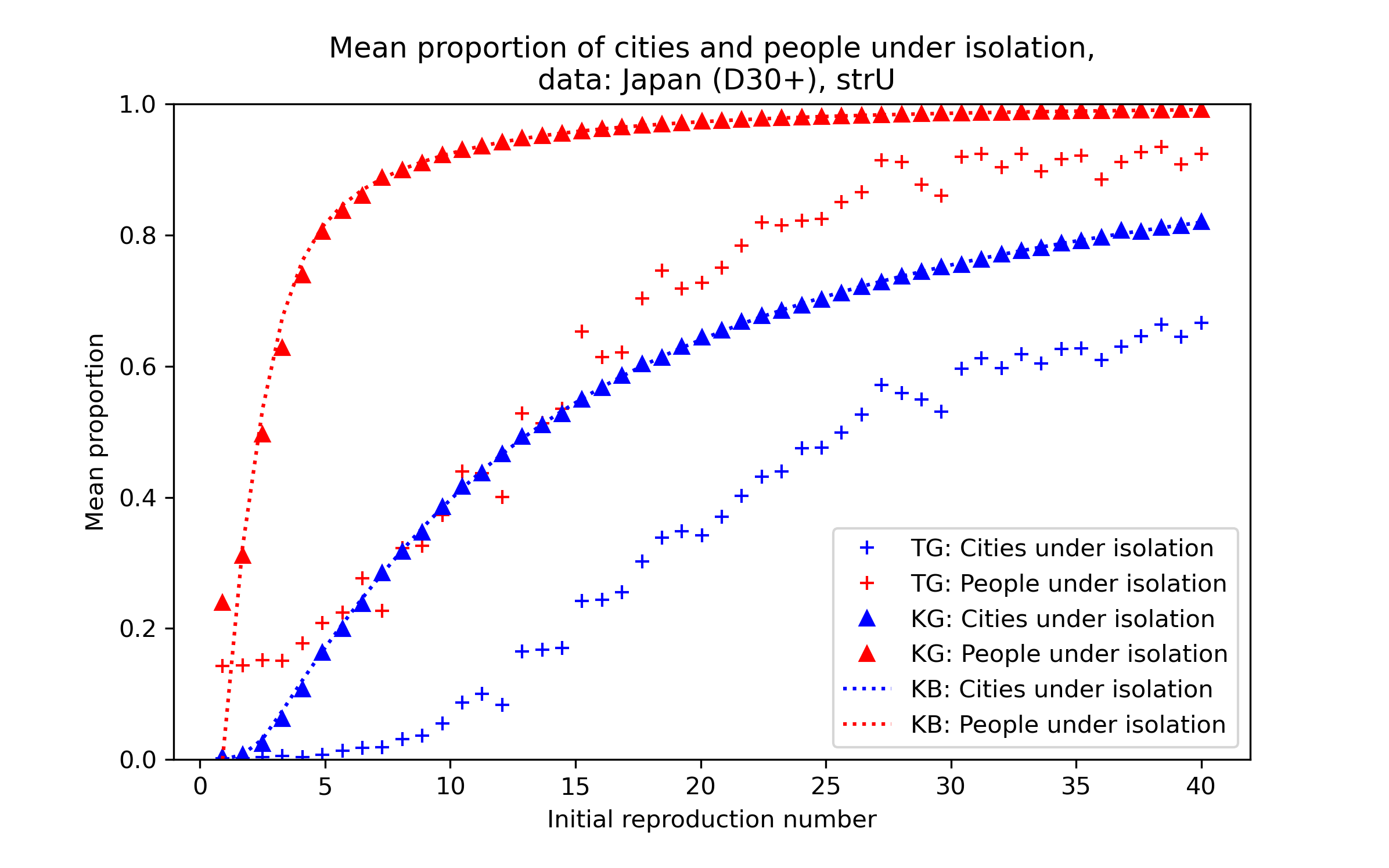}\\
		& \multicolumn{2}{c}{\includegraphics[draft = false, width = 0.80\textwidth]{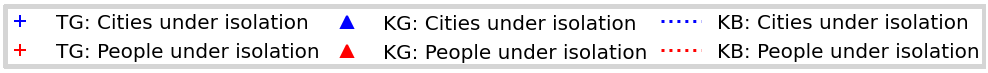}}	
	\end{tabular}
	\caption{
		Comparison of the theoretical and simulated proportion of infected cities and isolated persons depending on the basic reproduction number for mobility data from France, Poland and Japan
		in the upper, middle and lower row, resp.}
	\label{Fig_Final_Incidence}
\end{figure}

Firstly, figure~\ref{Fig_Final_Incidence}
is an alternative way of presenting 
the curves displayed in Figures~\ref{Fig_incidence_ppl} and \ref{Fig_incidence_cit}
by treating each country in a separate panel.

\begin{figure}
	\begin{center}
		
	\end{center}
	\begin{tabular}{l@{\hskip 0.8mm}cc}
		&\multicolumn{2}{c}{ Infection probability}\\
		&\quad Strategy \((P)\) & \quad Strategy \((U)\)\\
		\rotatebox{90}{\hcm{1.6} France}
		&  \includegraphics[draft = false, width = 0.44\textwidth]{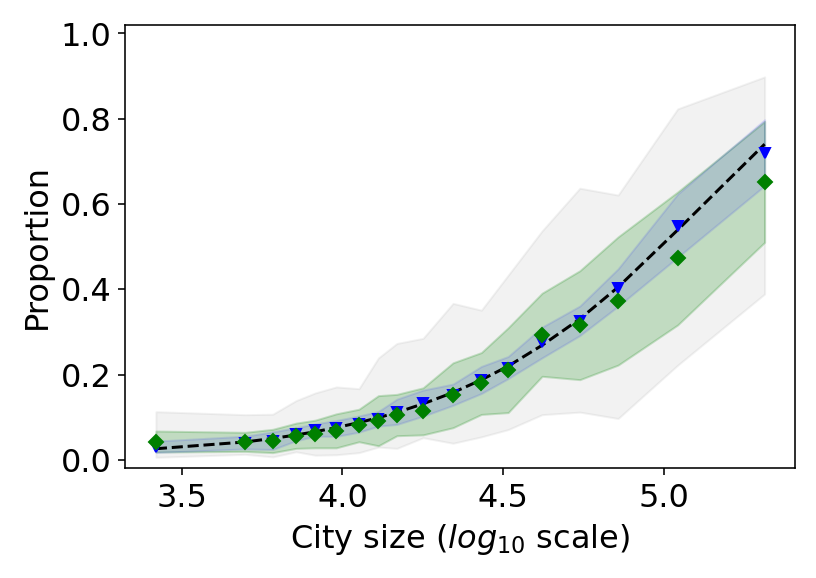}   
		& \includegraphics[draft = false, width = 0.44\textwidth]{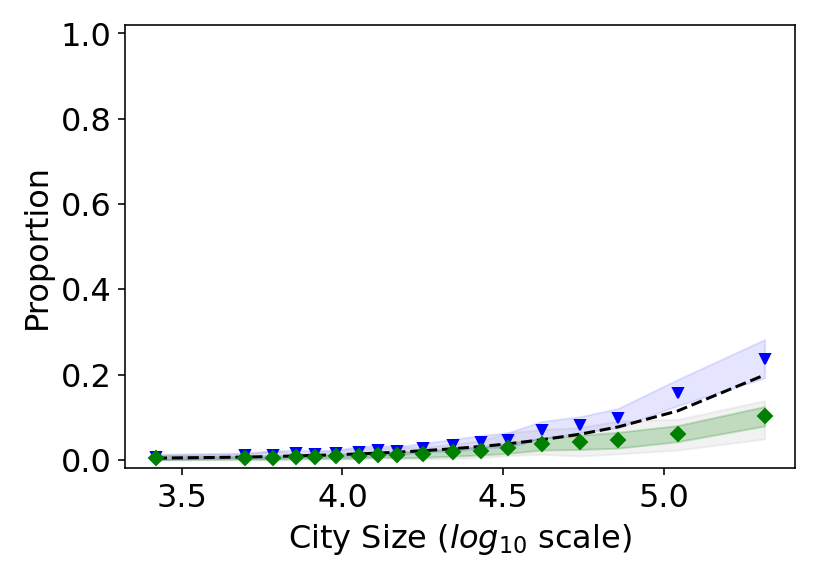} \\
		\rotatebox{90}{\hcm{1.6} Poland}
		&\includegraphics[draft = false, width = 0.44\textwidth]{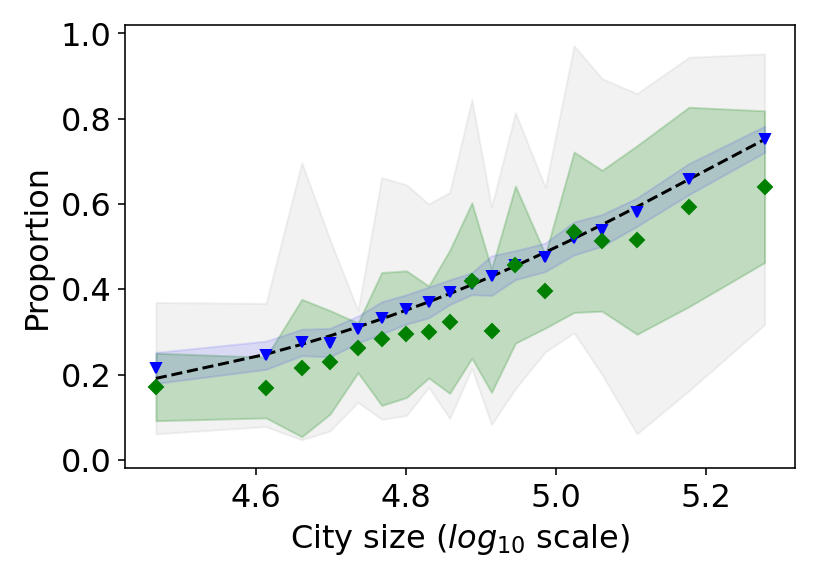} 
		& \includegraphics[draft = false, width = 0.44\textwidth]{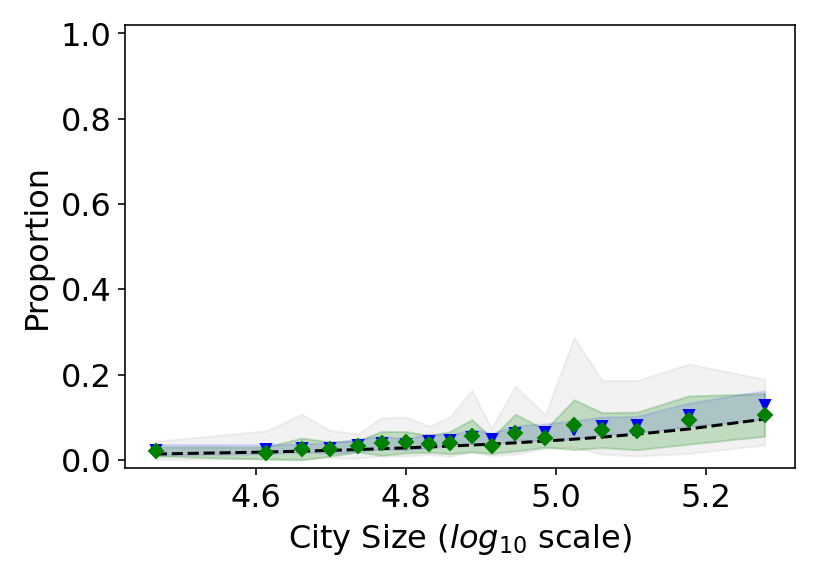}\\
		\rotatebox{90}{\hcm{1.6} Japan}
		&\includegraphics[draft = false, width = 0.44\textwidth]{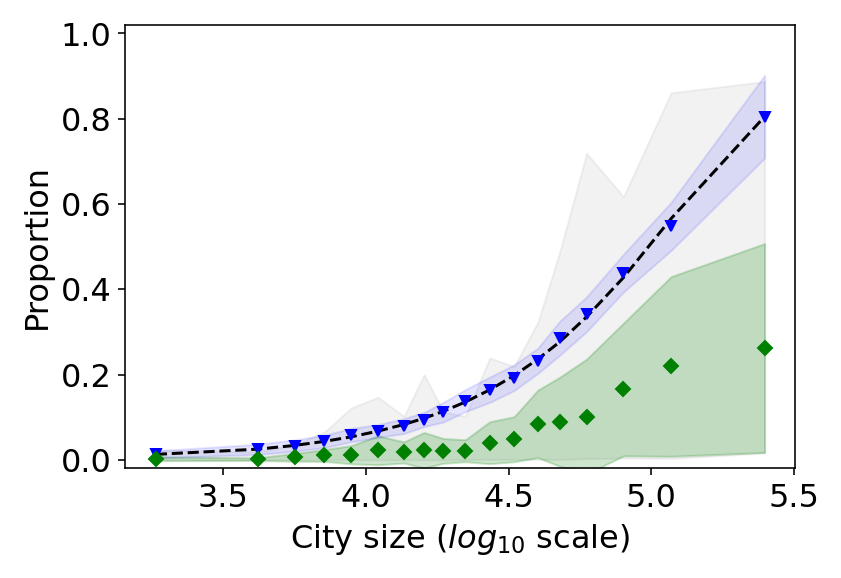} 
		& \includegraphics[draft = false, width = 0.44\textwidth]{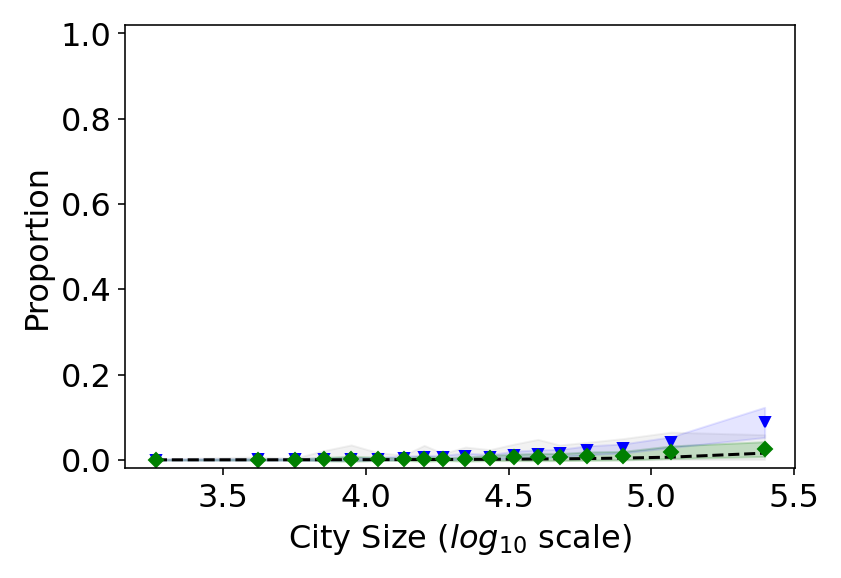}\\
		&\multicolumn{2}{c}{\includegraphics[draft = false, width = 0.80\textwidth]{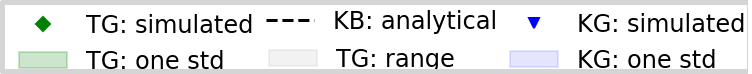}}	
	\end{tabular}
	\caption{
		Comparison of simulated and theoretical infection probabilities, on the left for strategy~\((P)\) and on the right for strategy~\((U)\) for mobility data from France, Poland and Japan in the upper, middle and lower row, resp. The $R_0$ value is adjusted such that  the theoretical outbreak probability is 0.5
         for cities of size $10^5$, see Section \ref{sec_simulated_infection_and_outbreak_prob}.}
	\label{fig_comp_infection_prob_pO}
\end{figure}

\begin{figure}
	
	\begin{tabular}{l@{\hskip 0.8mm}cc}
		&\multicolumn{2}{c}{Outbreak probability}\\
		&\quad Strategy \((P)\) & \quad Strategy \((U)\)\\
		\rotatebox{90}{\hcm{1.6} France}
		&  \includegraphics[draft = false, width = 0.44\textwidth]{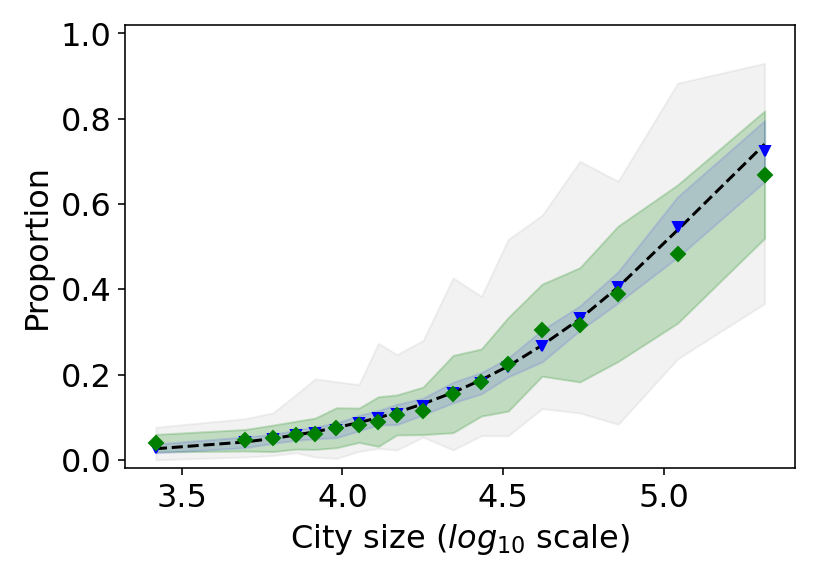}   
		& \includegraphics[draft = false, width = 0.44\textwidth]{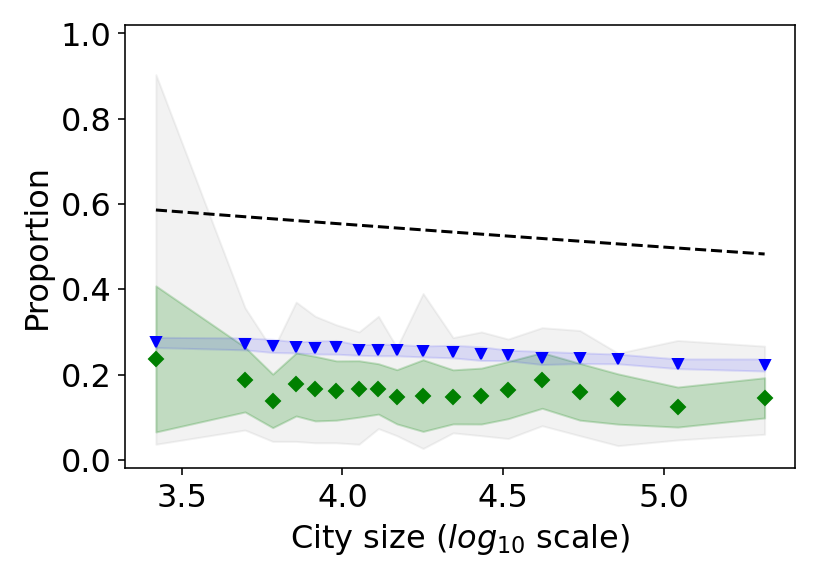} \\
		\rotatebox{90}{\hcm{1.6} Poland}
		&\includegraphics[draft = false, width = 0.44\textwidth]{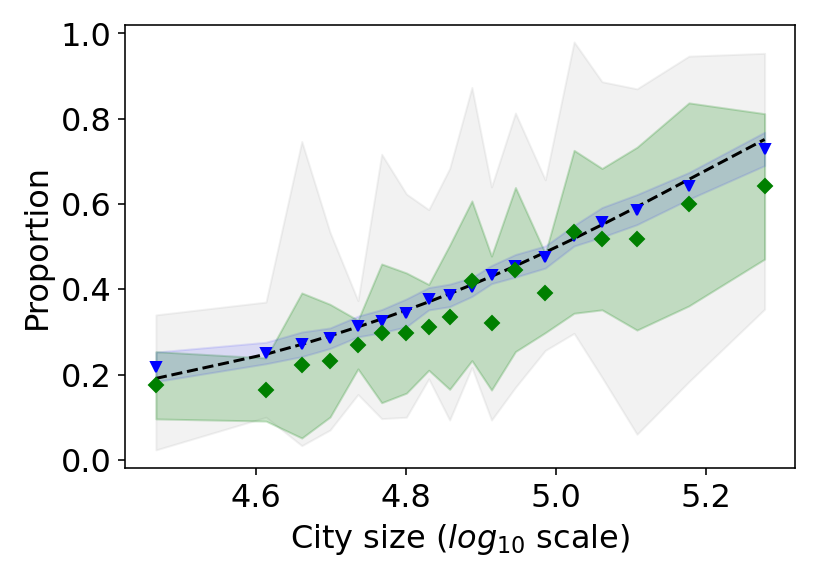} 
		& \includegraphics[draft = false, width = 0.44\textwidth]{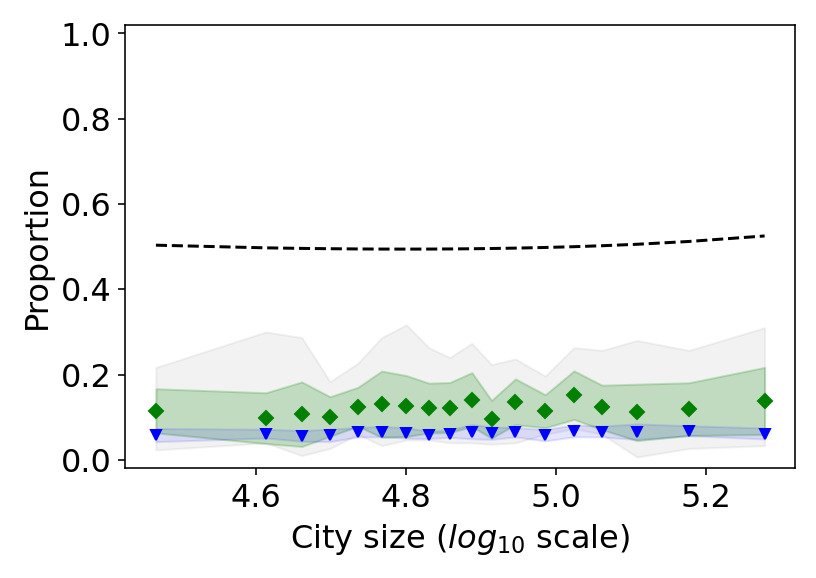}\\
		\rotatebox{90}{\hcm{1.6} Japan}
		&\includegraphics[draft = false, width = 0.44\textwidth]{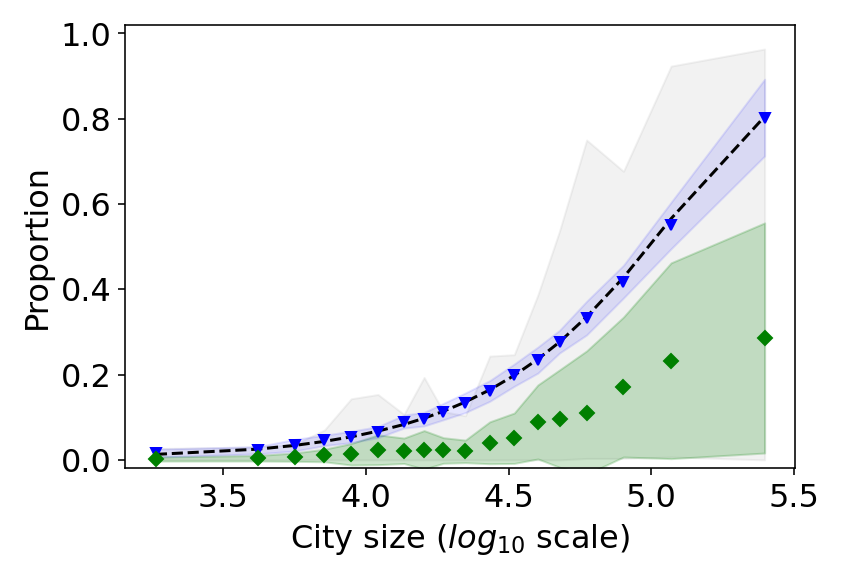} 
		& \includegraphics[draft = false, width = 0.44\textwidth]{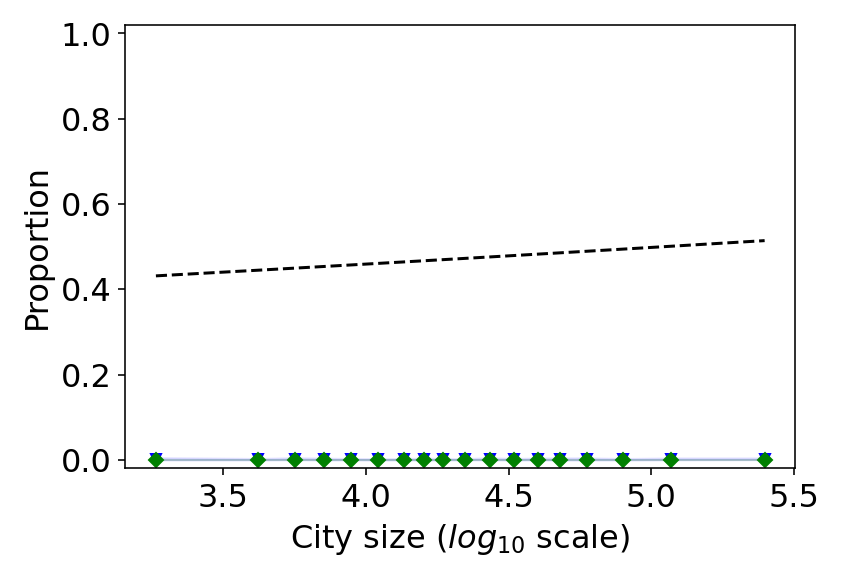}\\
		&\multicolumn{2}{c}{\includegraphics[draft = false, width = 0.80\textwidth]{legend_S2_variation_infected_TKG.png}}	
	\end{tabular}
	\caption{Comparison of simulated and theoretical outbreak probabilities, on the left for strategy~\((P)\) and on the right for strategy~\((U)\) for mobility data from France, Poland and Japan in the upper, middle and lower row, resp. The $R_0$ value is adjusted such that  the theoretical outbreak probability is 0.5
     for cities of size $10^5$, see Section \ref{sec_simulated_infection_and_outbreak_prob}}
	\label{fig_comp_outbreak_prob_pO}
\end{figure}

Figures~\ref{fig_comp_infection_prob_pO}
and \ref{fig_comp_outbreak_prob_pO}
complement 
Figures~\ref{Fig_Compare_Analytical_PI_Empirical_PI} and 
\ref{Fig_Compare_Analytical_PO_Empirical_PO}
with a focus on the infection and the outbreak probabilities, respectively.
Both strategies \((U)\) and \((P)\)  and the three countries France, Poland and Japan are again considered with the distance restriction D30+, recall Section~\ref{Sec:MobilityData}.
The difference to Figures~\ref{Fig_Compare_Analytical_PI_Empirical_PI} and 
\ref{Fig_Compare_Analytical_PO_Empirical_PO}
is that the $R_0$ value,
through the parameter $k_B$,
is adjusted according to the theoretical outbreak probability (rule~$pO$)
instead of the theoretical infection probability (rule~$pI$) (the probability is set equal to $0.5$ 
for a city of size $10^5$). 

We 
make the same observations for these Figures~\ref{fig_comp_infection_prob_pO}
and~\ref{fig_comp_outbreak_prob_pO}
as for the Figures~\ref{Fig_Compare_Analytical_PI_Empirical_PI} and~
\ref{Fig_Compare_Analytical_PO_Empirical_PO}:
(i) there is a relatively good fit for Poland except for the outbreak probability under strategy~\((U)\),
consistent though less optimal compared to 
France;
(ii)
the fit is
poor for Japan, simply less visible under the rule~$pO$ for the infection probability under strategy~\((U)\) because the corresponding values are very close to 0.

However, even though the fit is similar under rule $pO$ and $pI$, we would recommend a fit according to rule $pI$, since the outbreak probability is not fitted well under strategy~\((U)\).
\medskip 

In the following Figures~\ref{Fig:Compare PI v4}-\ref{Fig:Compare PI v9}
we show
how the restriction on the distance affects 
the fits of simulated and theoretical
infection and outbreak probabilities.
The plots associated with datasets
where no distance restrictions were applied (i.e.~with the datasets D1+)
are shown 
in Figures~\ref{Fig:Compare PI v4}, \ref{Fig:Compare PI v5}
and \ref{Fig:Compare PI v6},
while 
those with a distance restriction below 50km
(i.e. with the datasets D50+)
are shown in Figures~\ref{Fig:Compare PI v7}, \ref{Fig:Compare PI v8}
and \ref{Fig:Compare PI v9},
for respectively France, Poland and Japan.
For reasons explained just above, we present only the results obtained under rule~$pI$, and actually we see already under this rule that the 
theoretical outbreak probability
is not fitted well  either under strategy~\((U)\)
with these datasets D1+ and D50+.

Considering the datasets D50+ 
in comparison to the datasets D30+ does not lead qualitatively to different conclusions.
Especially for France and for Poland, good fits are preserved (with only a slight deterioration) 
for scenarios that have already shown  
good performances for D30+,
recall Figures~\ref{Fig_Compare_Analytical_PI_Empirical_PI} and~\ref{Fig_Compare_Analytical_PO_Empirical_PO}.

For Japan,
the fit is 
similarly poor as for D30+.
Of note are also the levels of fluctuations, which are significantly larger than for the corresponding kernel graphs,
similarly for D50+ as for D30+.

Considering the datasets D1+ instead of the datasets D30+ 
do generally lead to fits that are not satisfying. 
For France and Poland the general trend is preserved with an however much poorer fit especially for large cities.

This observation confirms our original expectation that the restriction on relatively long-distance travels is essential for infection dynamics not to be blurred by spatially correlated short-distance travels and to characterize the potential attraction effects of the largest cities.
The comparison with the case of the D1+ dataset exemplifies that a very specific structure of the transportation graph is needed for the prediction with the kernel graph to inform about this spatial model. The heterogeneity in the city size distribution is not the only factor involved. The restriction on distant infection events between cities in this synthetic model isolates a very specific contribution of migrations that is captured both by the D30+ and D50+ datasets. We did not check 
larger values of distance restrictions because we do not see any specific interpretation of such choices, contrary to around 30-50km distances. The reduction in the migration frequency beyond these distances was previously noted in studies of work-mobility \cite{CommuterDistance}.

The difference between D30+ and D50+ on the one hand and D1+ on the other hand is also a reminder that the infection probability that is inferred (especially for D30+ and D50+) only corresponds to the direct aftermath of the first wave, that is to be followed and strengthened through much more local spread of the disease.
Note that the available transportation matrix (even for D1+) is not that reliable at predicting short-distance travels, while we expect the long-distance travels to be more accurately reflected by this work-related mobility. This distinction in the reliability of the data is even more pronounced in times where isolation strategies are applied preventively, due to the spontaneous change in migration behavior that is likely to happen.

There are several reasons 
that lead to a much worse fit of D1+ data in comparison to D30+ or even D50+ data.
Notably, we observe a significant change of the attractiveness coefficient $a$ as well as the emissivity coefficient $b$ for D1+ data for France and Poland (while it is conserved for Japan). For short-distance travels, proximity of cities plays a stronger role than the sizes of the cities. The kernel graph model and in particular the coefficients $a$ and $b$ do therefore
not accurately capture the dynamics of local spread.
Nonetheless regarding the infection probability,
the main observation is that the probability 
derived from the transportation graph (TG)
is in general lower than the estimate derived from the analytical formula,
the latter being similar to the one 
derived from the kernel graph (KG).
This discrepancy could actually be largely compensated by making 
the $R_0$ value larger for TG, 
then with possibly very similar relations between infection probabilities and city sizes
(figures not shown for both strategy~\((U)\) and~\((P)\), for Poland
and to a lesser extent for France).
This hints at the fact 
that  a dumping effect 
is produced
due to the spatial correlations induced by short-distance infection events,
as compared to the branching approximation.
In practice, the effective $R_0$ value appears to be reduced 
when comparing the estimations derived from TG  
to the analytical ones.

For Japan, 
the fit is not good for any of the considered distance restrictions,
which makes it more delicate to interpret the observed discrepancies.
Nonetheless, 
we can say that the spread is much reduced on the transportation graph as compared to the branching graph
(agreeing well with the kernel graph situation  also for Japanese data).
Regarding the relation between the mean proportion of people under isolation as a function of the basic reproduction number
in Figure~\ref{Fig:Compare PI v6},
there is presumably a significant contribution of
the choice of the initially infected city 
(according to the distribution $\nu_{O, \rA}$).
Given that the largest unit in Japan comprises around 28\% of the whole population and is likely to be this first choice, 
the relatively high value of people under isolation 
is reasonable even when only a tiny fraction of cities is under isolation.

Since we observe that the proportion of people under isolation 
follows a similar curve as the proportion of cities under isolation with a relatively constant difference under strategy~\((U)\) and~\((P)\) for the D1+ dataset,
we conjecture that besides this initially infected city,
the infection probability is not depending much on the city size.
This hints at a very strong effect of the spatial correlations
in the epidemics derived from this Japanese dataset. 
This is in contrast to the prediction derived from the branching approximation,
where we observe a much larger initial slope for the proportion of people under isolation than for the one of cities under isolation.

\FloatBarrier

\begin{figure}
	\begin{center}
		France D1+
	\end{center}
	\begin{tabular}{l@{\hskip 0.8mm}cc}
		&Strategy \((P)\)& Strategy \((U)\)\\
		\rotatebox{90}{ \hcm{0.9} {\small Proportion}} \rotatebox{90}{ \hcm{0.6} {\small  under isolation}}
		&\includegraphics[width = 0.43\textwidth]{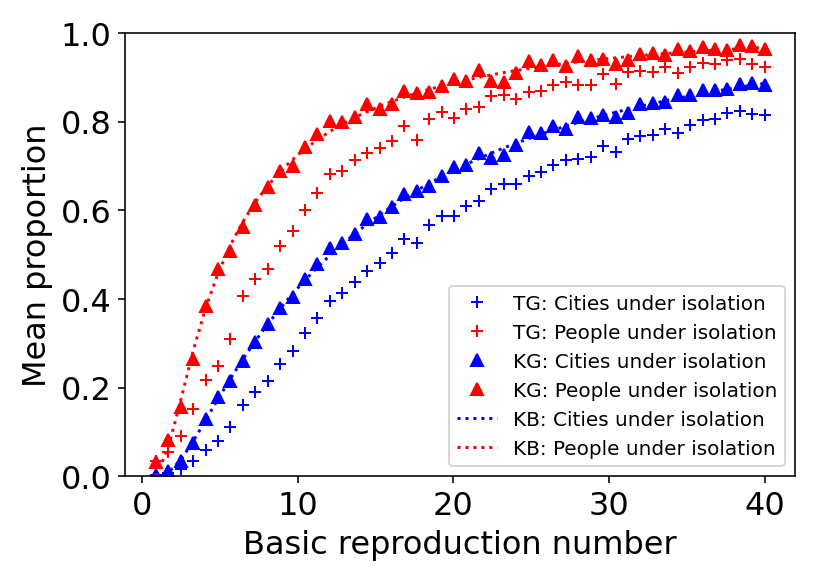}
		&\includegraphics[width = 0.43\textwidth]{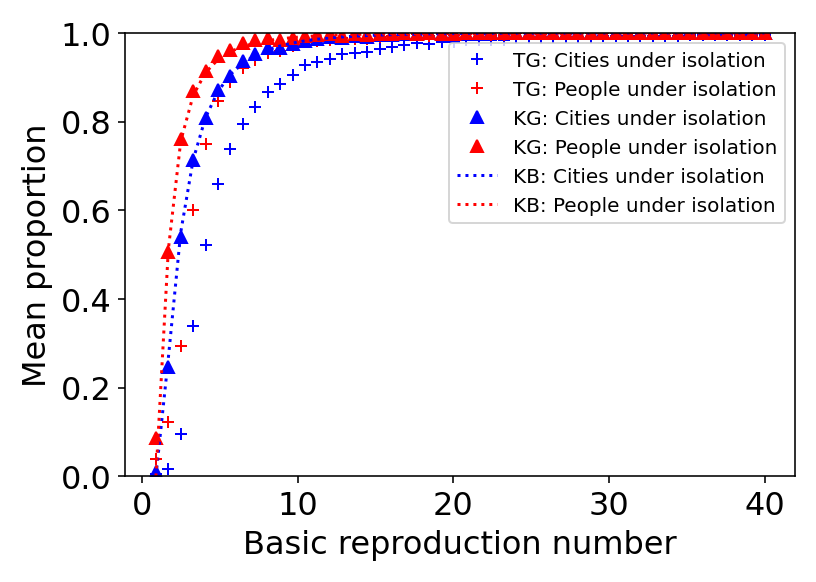}\\
		\rotatebox{90}{\quad {\small Infection probability}}
		&\includegraphics[width = 0.43\textwidth]{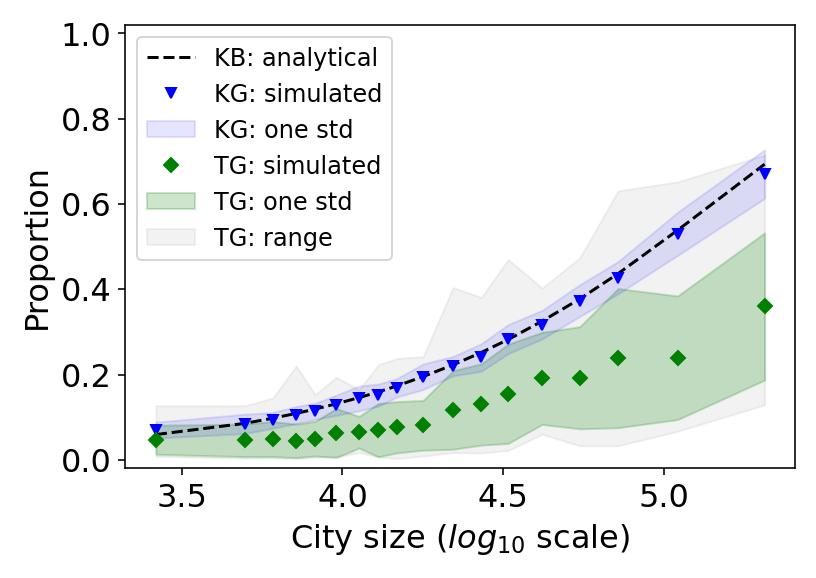}
		&\includegraphics[width = 0.43\textwidth]{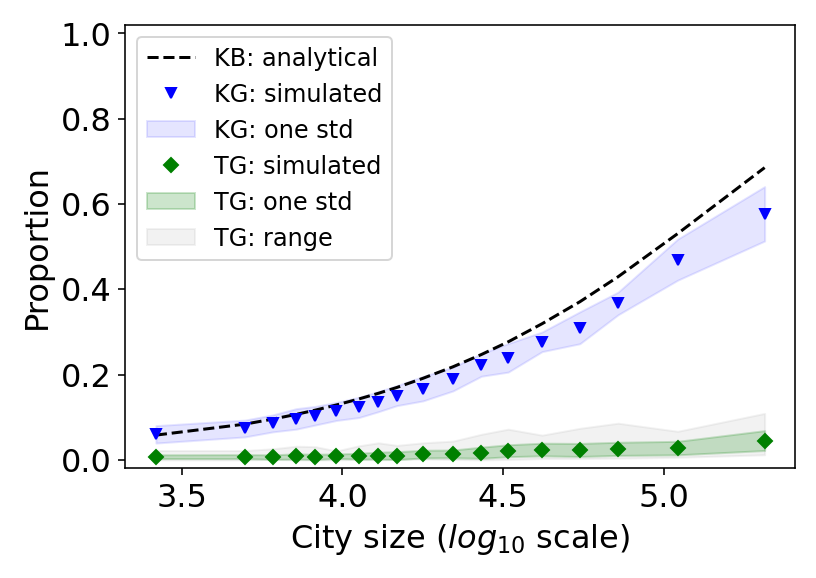}
		\\
		\rotatebox{90}{\quad {\small Outbreak probability}}
		&
		\includegraphics[width = 0.43\textwidth]{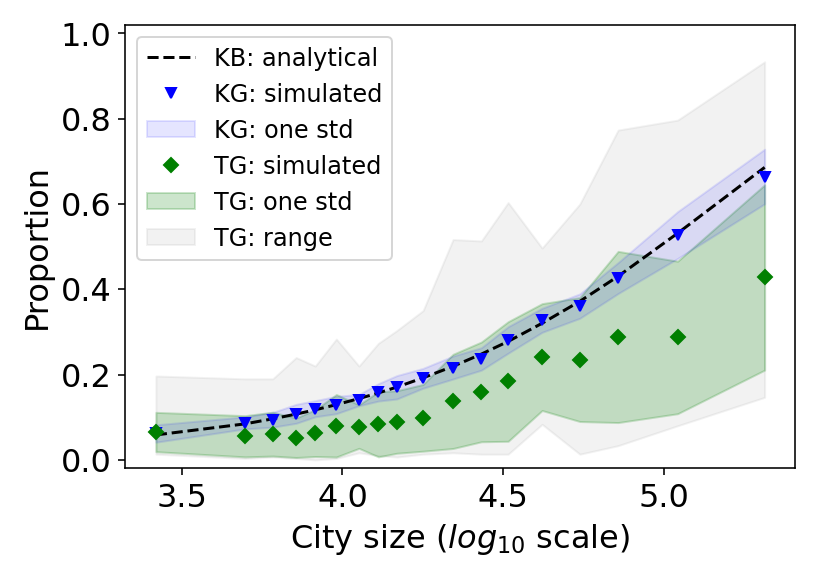}
		&\includegraphics[width = 0.43\textwidth]{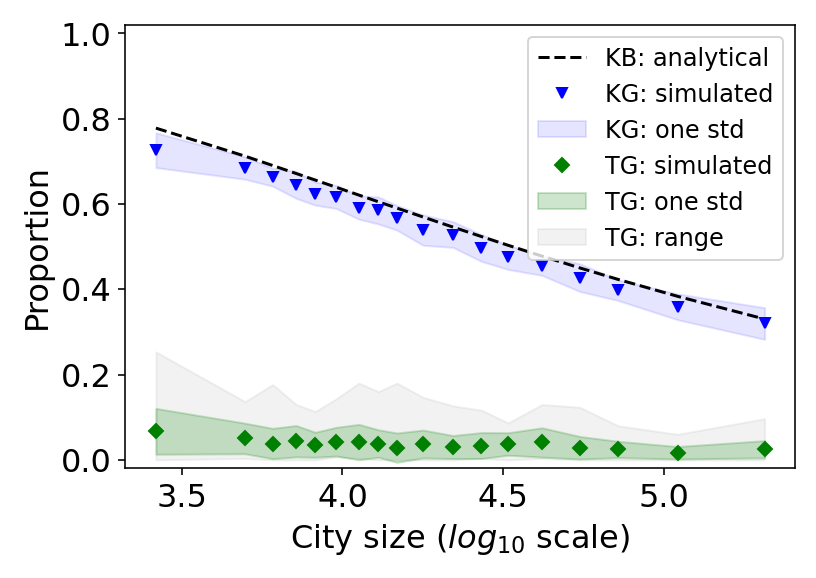}
	\end{tabular}
	\caption{Comparison of infected cities and isolated cities as a function of the theoretical $R_0$ value (top row)
		of simulated and theoretical infection (middle row) and outbreak probabilities (bottom row), on the left for strategy~\((P)\) and on the right for strategy~\((U)\), for France
		without restriction on the distance (data abbreviated as D1+). The $R_0$ value is adjusted such that  the theoretical infection probability is 0.5  for cities of size $10^5$,  see Section \ref{sec_simulated_infection_and_outbreak_prob}.
	}\label{Fig:Compare PI v4}
\end{figure}

\begin{figure}
	\begin{center}
		France D50+
	\end{center}
	\begin{tabular}{l@{\hskip 0.8mm}cc}
		&Strategy \((P)\)& Strategy \((U)\)\\
		\rotatebox{90}{ \hcm{0.9} {\small Proportion}} \rotatebox{90}{ \hcm{0.6} {\small  under isolation}}
		&\includegraphics[width = 0.43\textwidth]{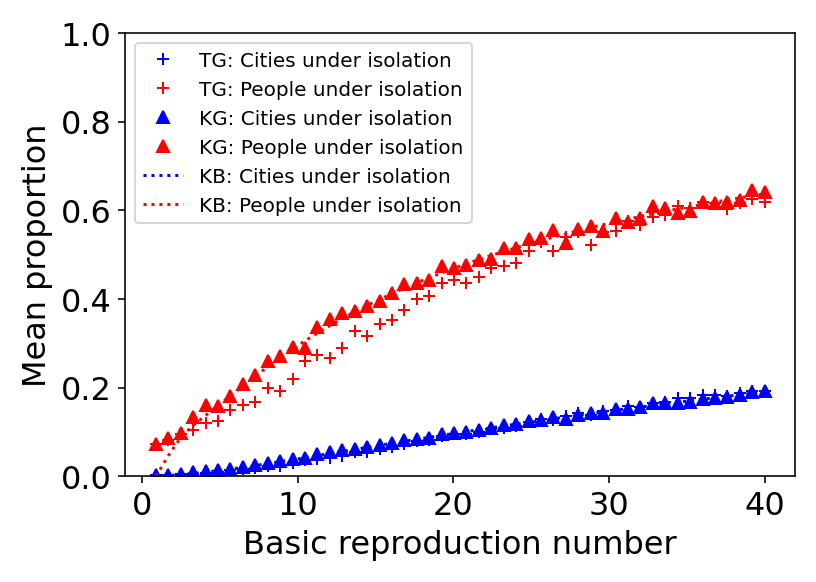}
		&\includegraphics[width = 0.43\textwidth]{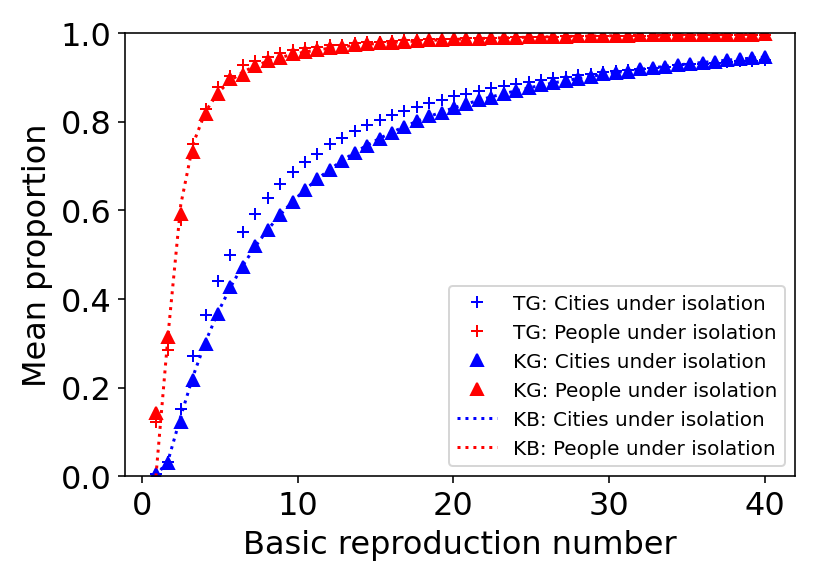}
		\\
		\rotatebox{90}{\quad {\small Infection probability}}
		&\includegraphics[width = 0.43\textwidth]{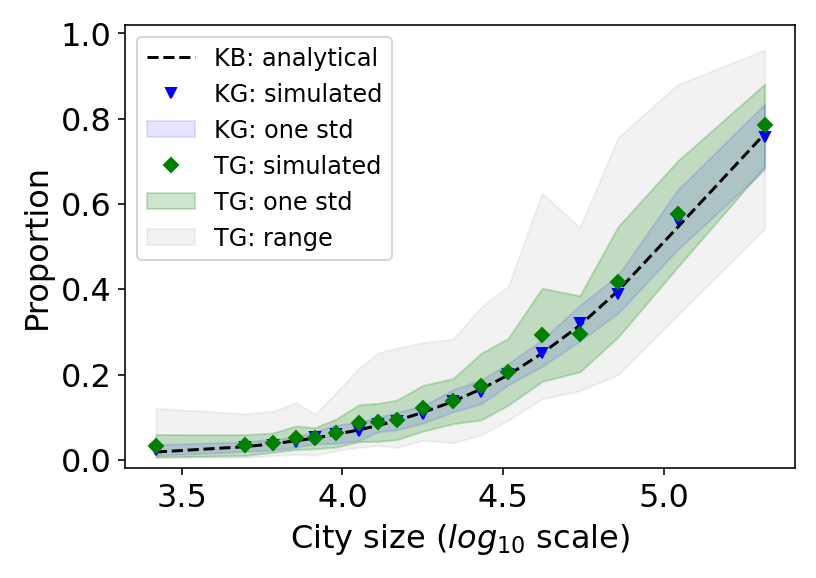}
		&\includegraphics[width = 0.43\textwidth]{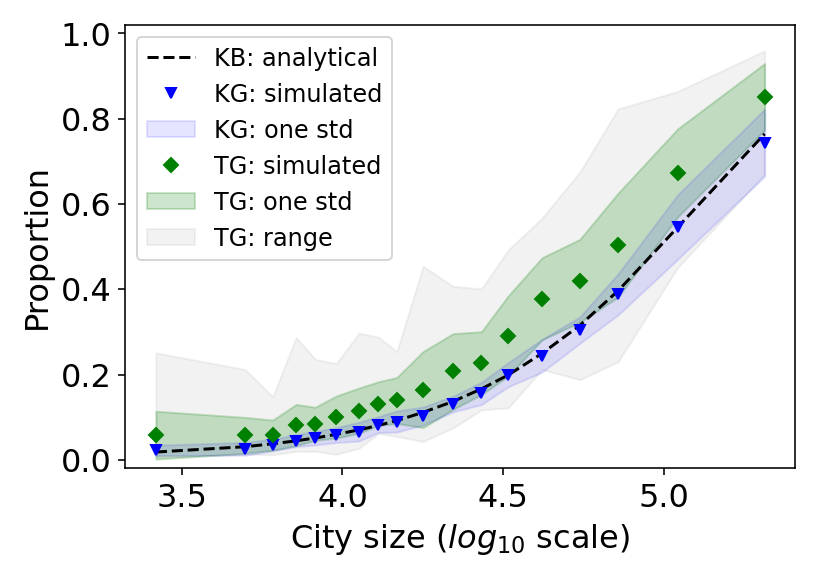}
		\\
		\rotatebox{90}{\quad {\small Outbreak probability}}
		&\includegraphics[width = 0.43\textwidth]{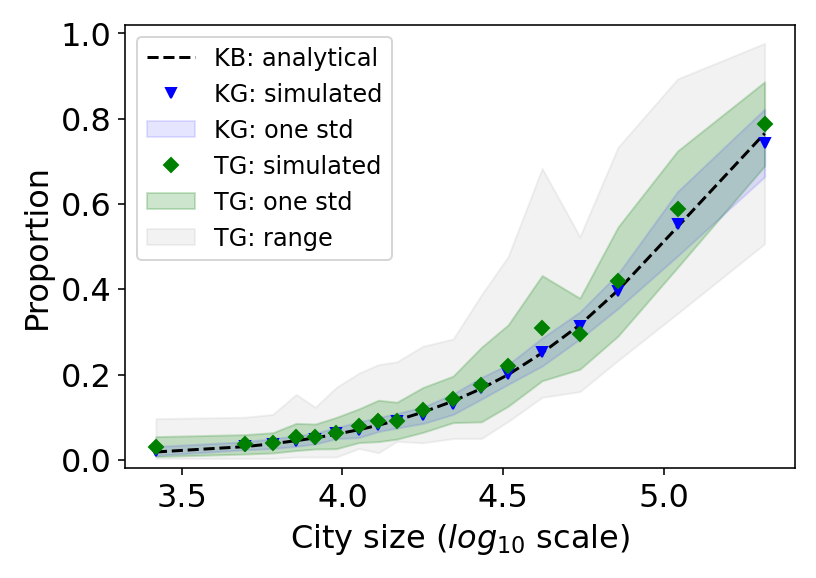}
		&\includegraphics[width = 0.43\textwidth]{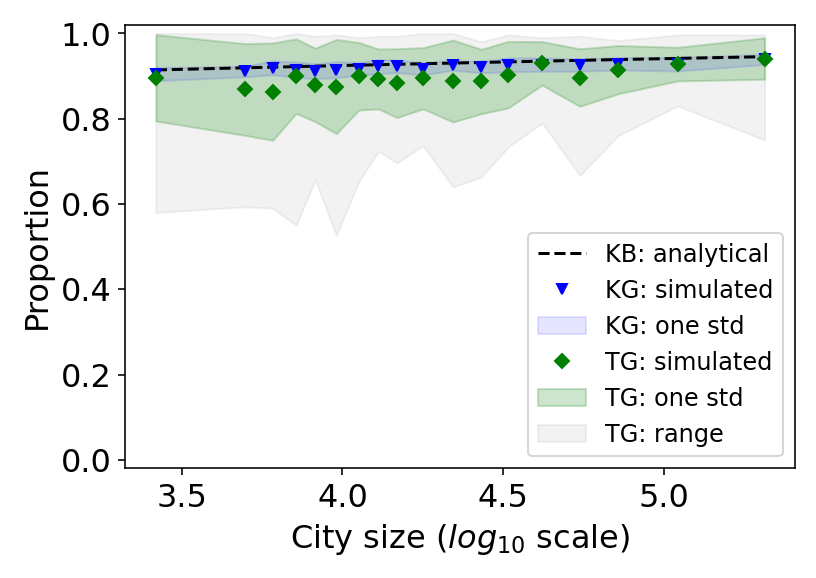}
	\end{tabular}
	\caption{Comparison of infected cities and isolated cities as a function of the theoretical $R_0$ value (top row)
		of simulated and theoretical infection (middle row) and outbreak probabilities (bottom row), on the left for strategy~\((P)\) and on the right for strategy~\((U)\), for France
		with a restriction on the distance of 50km (data abbreviated as D50+). The $R_0$ value is adjusted such that   the theoretical infection probability is 0.5
         for cities of size $10^5$,  see Section \ref{sec_simulated_infection_and_outbreak_prob}.
	}
	\label{Fig:Compare PI v7}
\end{figure}

\begin{figure}
	\begin{center}
		Poland D1+
	\end{center}
	\begin{tabular}{l@{\hskip 0.8mm}cc}
		&Strategy \((P)\)& Strategy \((U)\)\\
		\rotatebox{90}{ \hcm{0.9} {\small Proportion}} \rotatebox{90}{ \hcm{0.6} {\small  under isolation}}
		&\includegraphics[width = 0.43\textwidth]{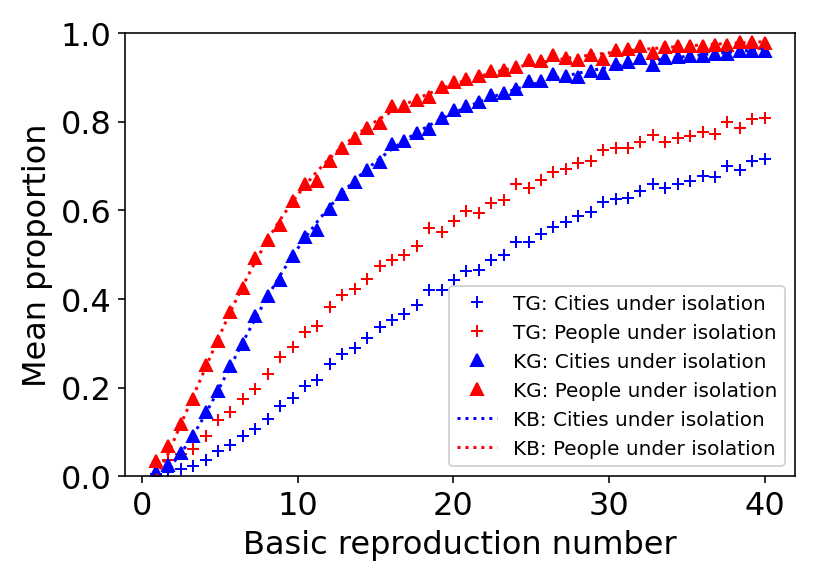}
		&\includegraphics[width = 0.43\textwidth]{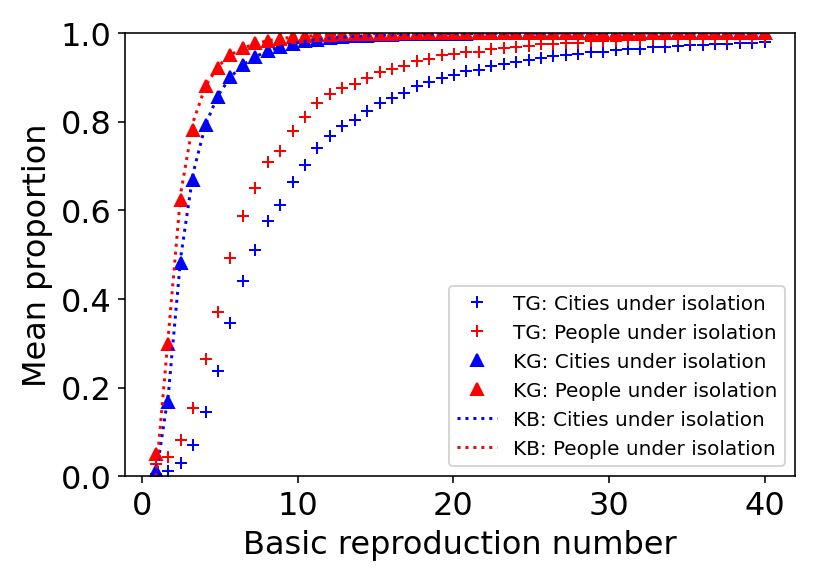}\\
		\rotatebox{90}{\quad {\small Infection probability}}
		&\includegraphics[width = 0.43\textwidth]{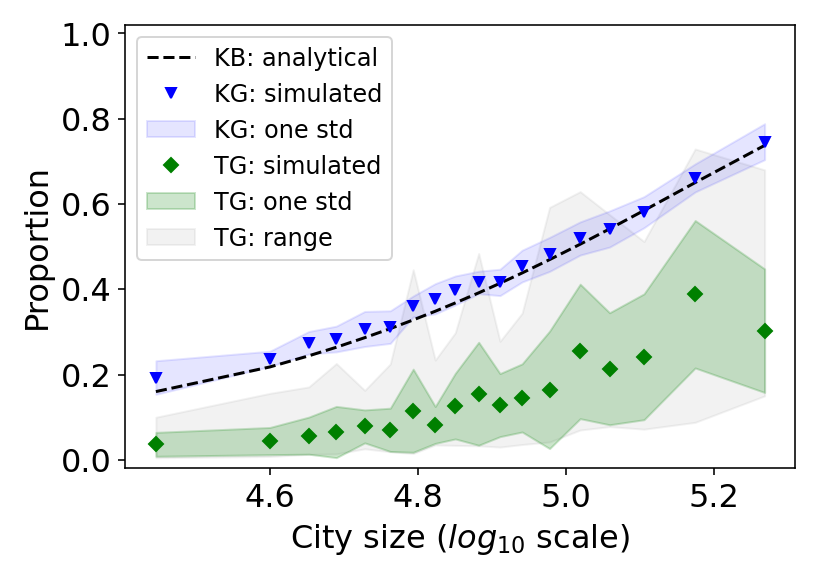}
		&\includegraphics[width = 0.43\textwidth]{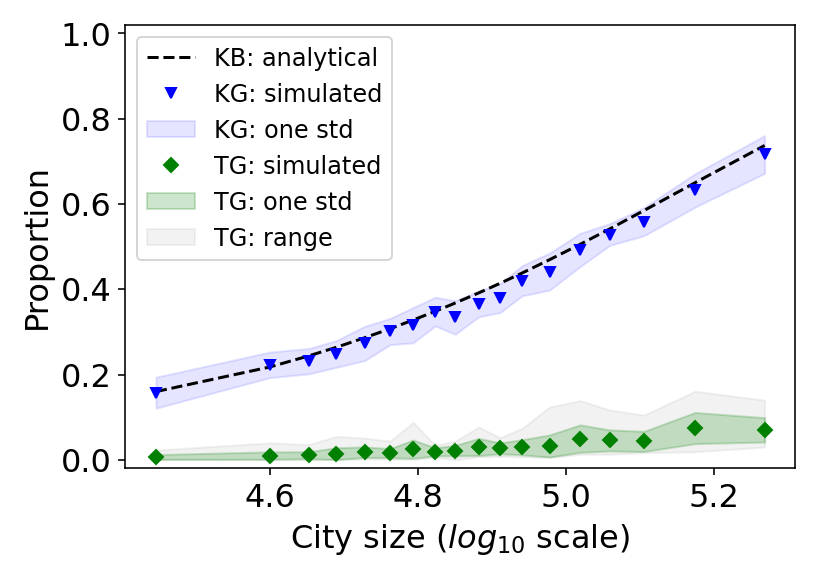}\\
		\rotatebox{90}{\quad {\small Outbreak probability}}
		&\includegraphics[width = 0.43\textwidth]{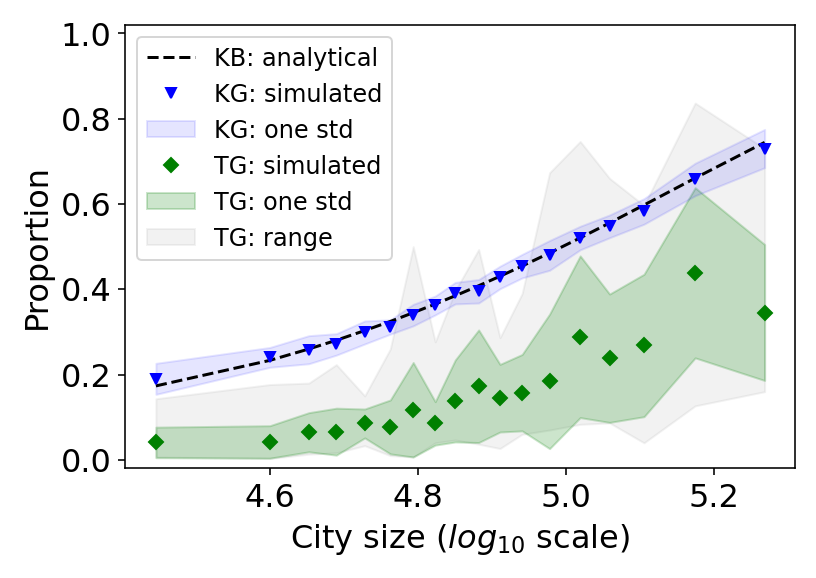}
		&\includegraphics[width = 0.43\textwidth]{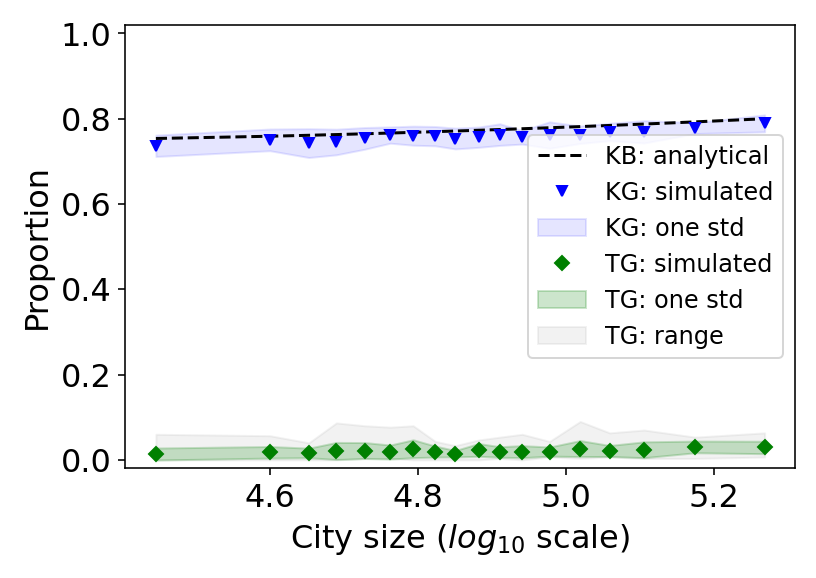}
	\end{tabular}
	\caption{Comparison of infected cities and isolated cities as a function of the theoretical $R_0$ value (top row)
		of simulated and theoretical infection (middle row) and outbreak probabilities (bottom row), on the left for strategy~\((P)\) and on the right for strategy~\((U)\), for Poland
		without restriction on the distance. The $R_0$ value is adjusted such that  the theoretical infection probability is 0.5
         for cities of size $10^5$, see Section \ref{sec_simulated_infection_and_outbreak_prob}.
	}\label{Fig:Compare PI v5}
\end{figure}

\begin{figure}
	\begin{center}
		Poland D50+
	\end{center}
	\begin{tabular}{l@{\hskip 0.8mm}cc}
		&Strategy \((P)\)& Strategy \((U)\)\\
		\rotatebox{90}{ \hcm{0.9} {\small Proportion}} \rotatebox{90}{ \hcm{0.6} {\small  under isolation}}
		&\includegraphics[width = 0.43\textwidth]{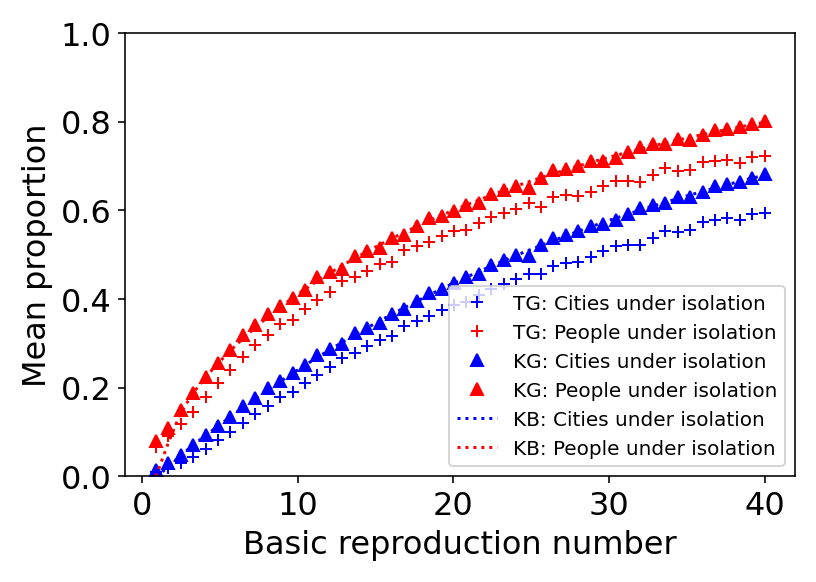}
		&\includegraphics[width = 0.43\textwidth]{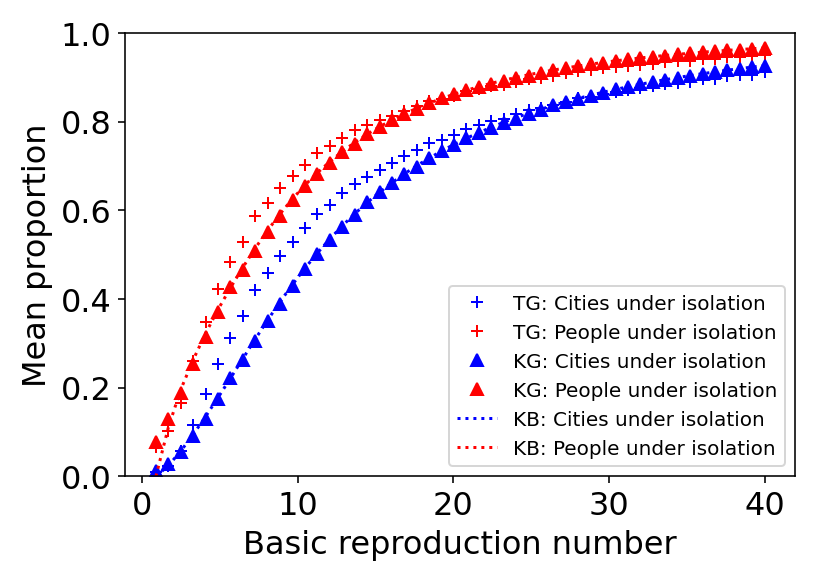}\\
		\rotatebox{90}{\quad {\small Infection probability}}
		&\includegraphics[width = 0.43\textwidth]{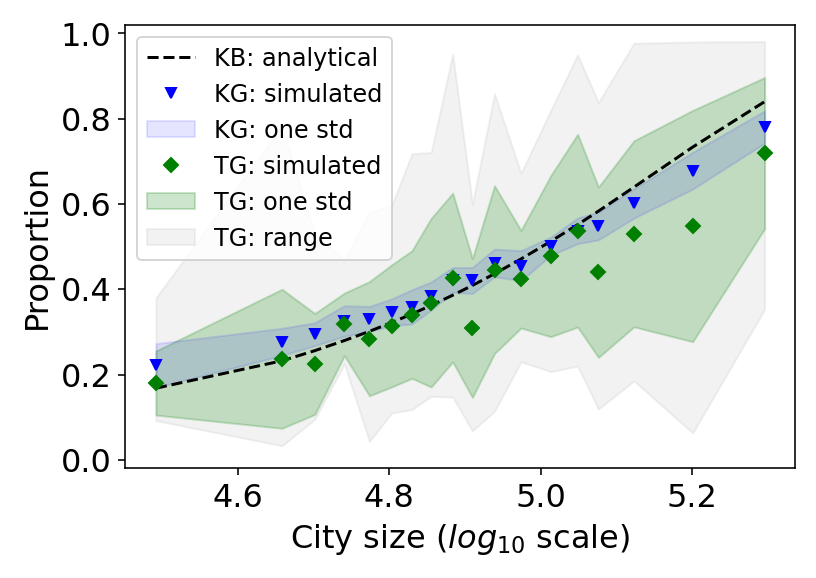}
		&\includegraphics[width = 0.43\textwidth]{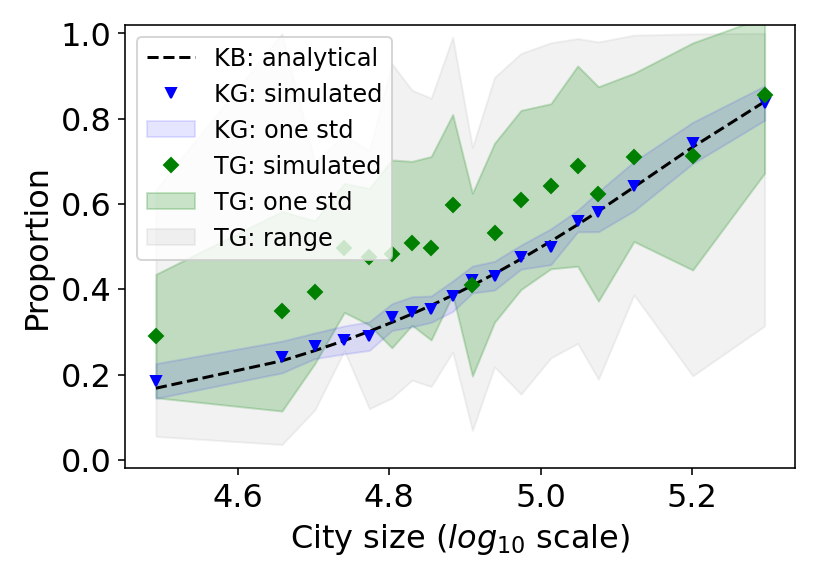}\\
		\rotatebox{90}{\quad {\small Outbreak probability}}
		&\includegraphics[width = 0.43\textwidth]{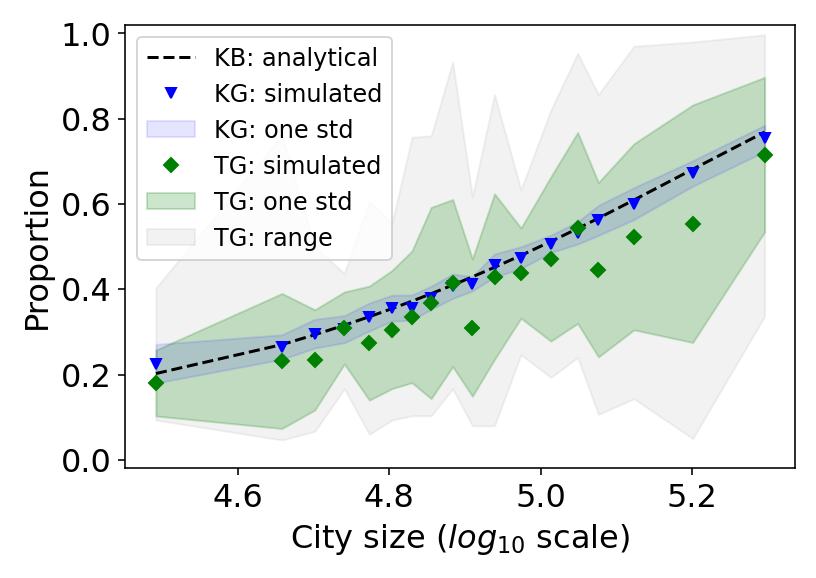}
		&\includegraphics[width = 0.43\textwidth]{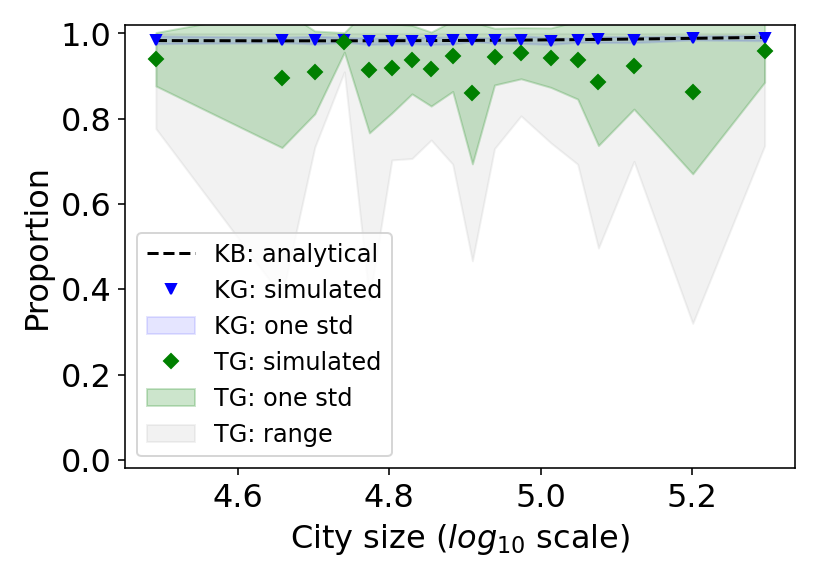}
	\end{tabular}
	\caption{Comparison of infected cities and isolated cities as a function of the theoretical $R_0$ value (top row)
		of simulated and theoretical infection (middle row) and outbreak probabilities (bottom row), on the left for strategy~\((P)\) and on the right for strategy~\((U)\), for Poland
		with a restriction on the distance of 50km. The $R_0$ value is adjusted such that the theoretical infection probability is 0.5
         for cities of size $10^5$, see Section \ref{sec_simulated_infection_and_outbreak_prob}.
	}\label{Fig:Compare PI v8}
\end{figure}

\begin{figure}
	\begin{center}
		Japan D1+
	\end{center}
	\begin{tabular}{l@{\hskip 0.8mm}cc}
		&Strategy \((P)\)& Strategy \((U)\)\\
		\rotatebox{90}{ \hcm{0.9} {\small Proportion}} \rotatebox{90}{ \hcm{0.6} {\small  under isolation}}
		&\includegraphics[width = 0.43\textwidth]{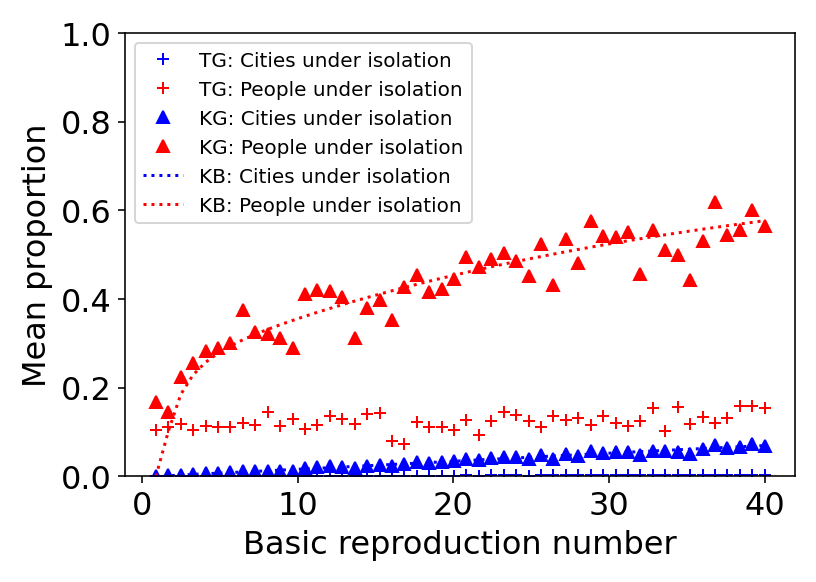}
		&\includegraphics[width = 0.43\textwidth]{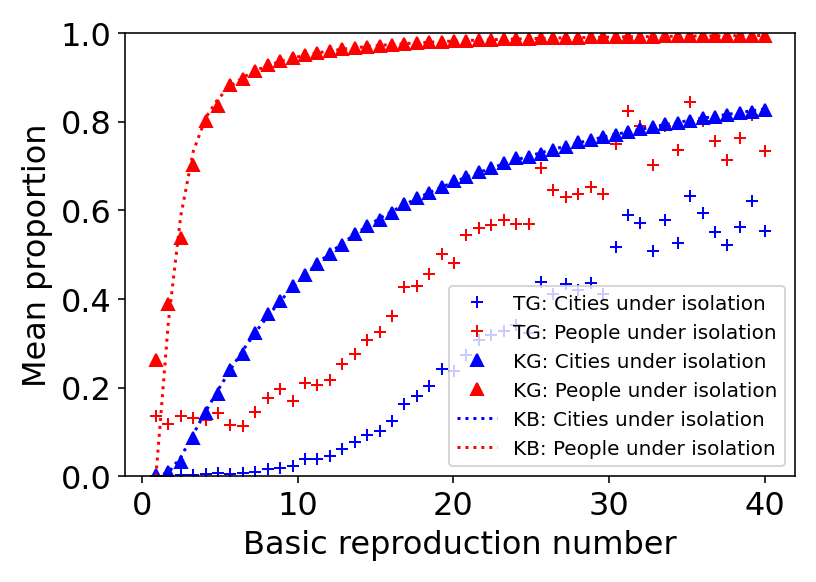}\\
		\rotatebox{90}{\quad {\small Infection probability}}
		&\includegraphics[width = 0.43\textwidth]{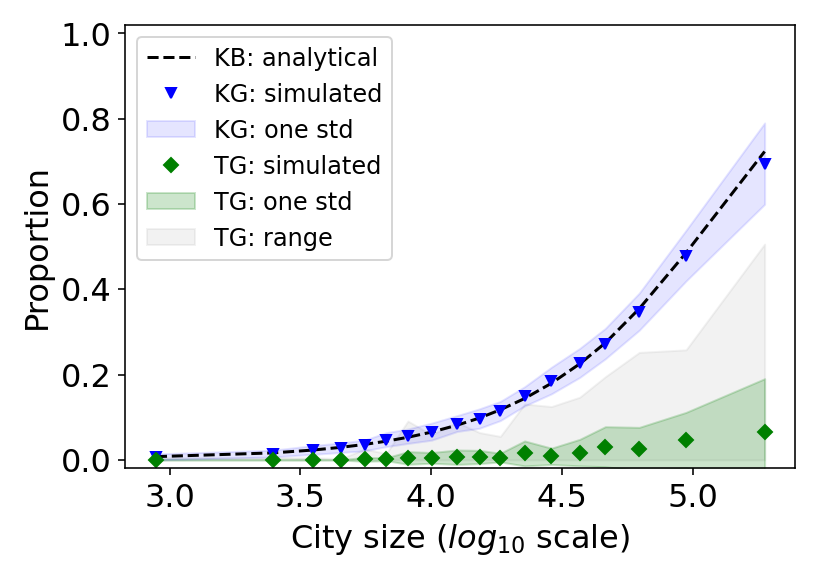}
		&\includegraphics[width = 0.43\textwidth]{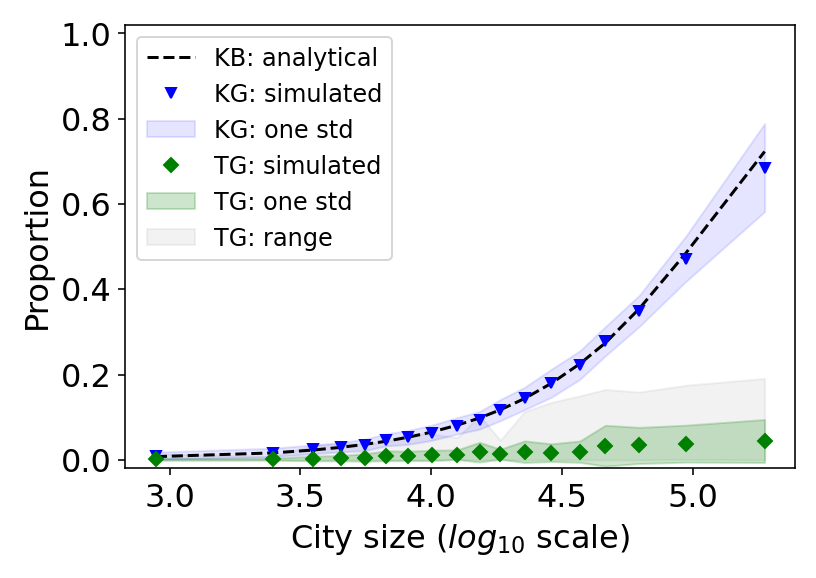}\\
		\rotatebox{90}{\quad {\small Outbreak probability}}
		&\includegraphics[width = 0.43\textwidth]{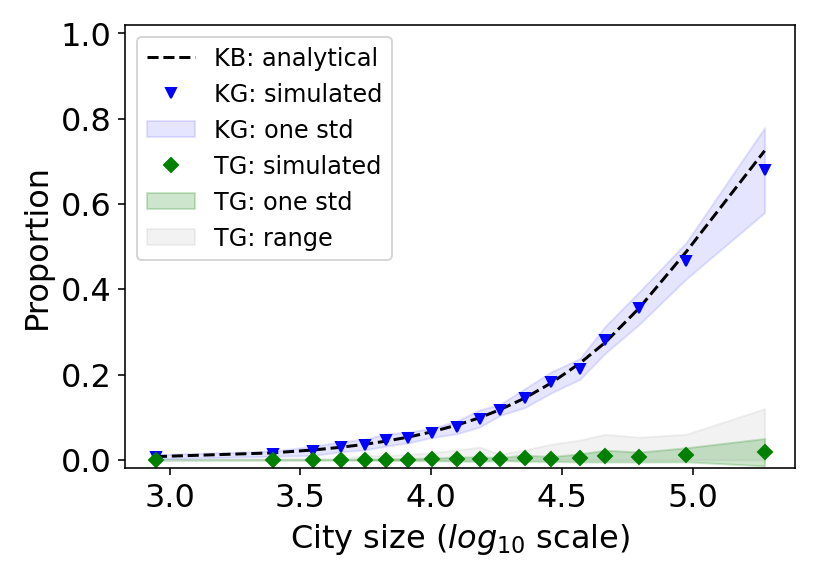}
		&\includegraphics[width = 0.43\textwidth]{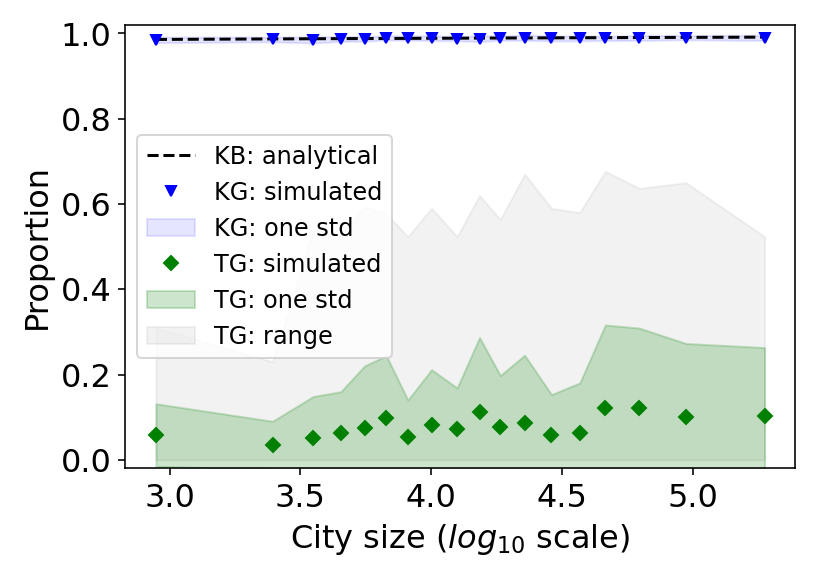}
	\end{tabular}
	\caption{Comparison of infected cities and isolated cities as a function of the theoretical $R_0$ value (top row)
		of simulated and theoretical infection (middle row) and outbreak probabilities (bottom row), on the left for strategy~\((P)\) and on the right for strategy~\((U)\), for Japan
		without restriction on the distance. The $R_0$ value is adjusted such that the theoretical infection probability is 0.5
         for cities of size $10^5$, see Section \ref{sec_simulated_infection_and_outbreak_prob}.
	}\label{Fig:Compare PI v6}
\end{figure}

\begin{figure}
	\begin{center}
		Japan D50+
	\end{center}
	\begin{tabular}{l@{\hskip 0.8mm}cc}
		&Strategy \((P)\)& Strategy \((U)\)\\
		\rotatebox{90}{ \hcm{0.9} {\small Proportion}} \rotatebox{90}{ \hcm{0.6} {\small  under isolation}}
		&\includegraphics[width = 0.43\textwidth]{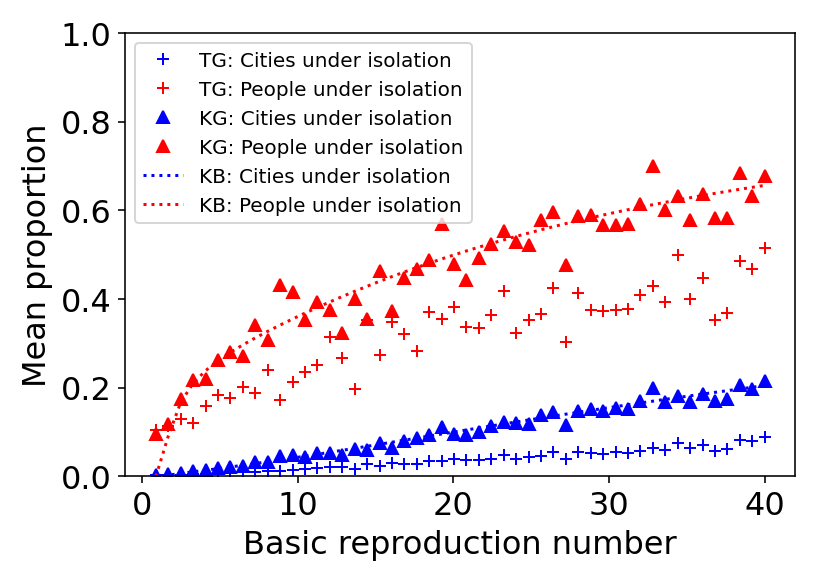}
		&\includegraphics[width = 0.43\textwidth]{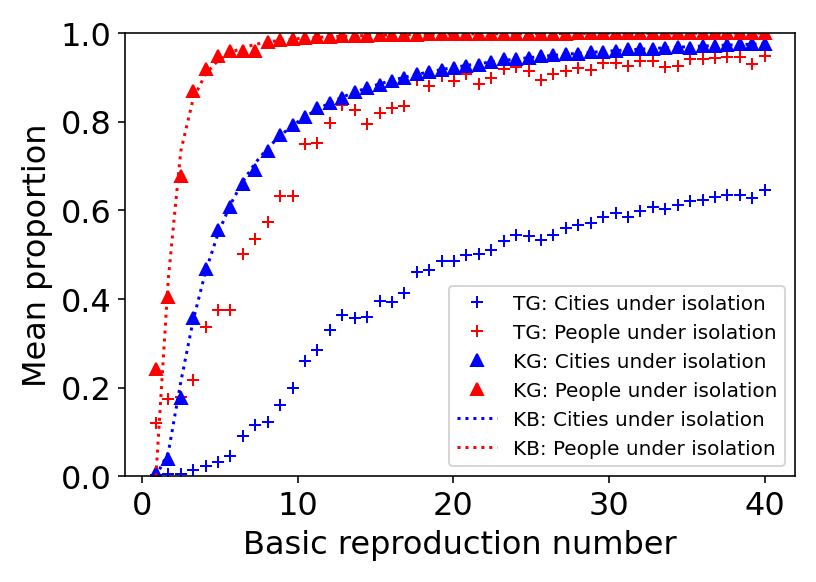}\\
		\rotatebox{90}{\quad {\small Infection probability}}
		&\includegraphics[width = 0.43\textwidth]{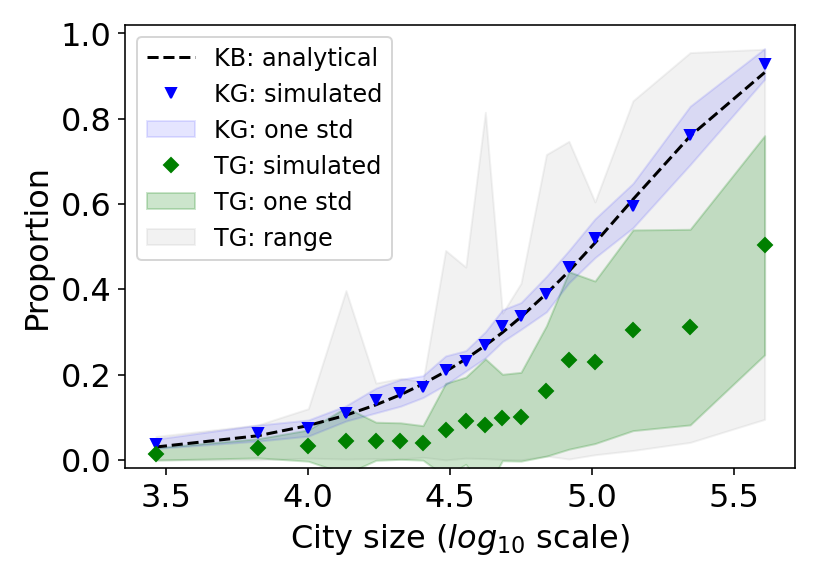}
		&\includegraphics[width = 0.43\textwidth]{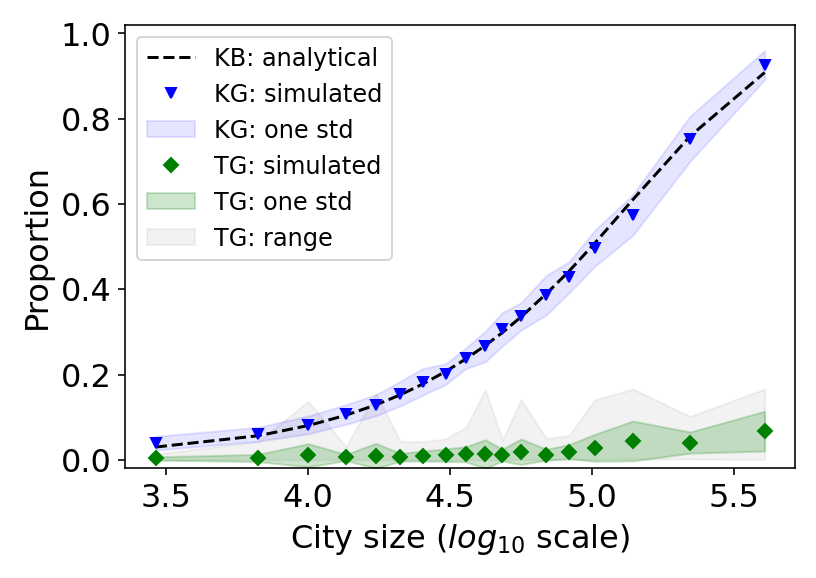}\\
		\rotatebox{90}{\quad {\small Outbreak probability}}
		&\includegraphics[width = 0.43\textwidth]{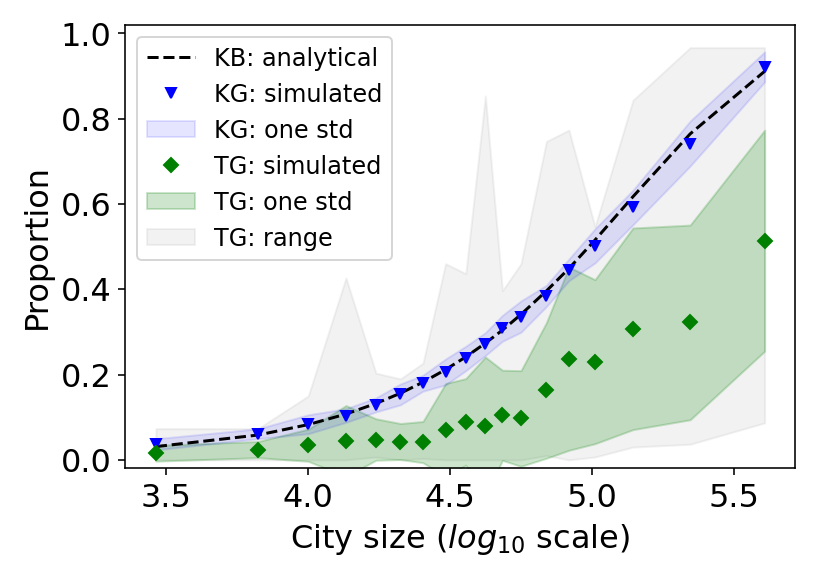}
		&\includegraphics[width = 0.43\textwidth]{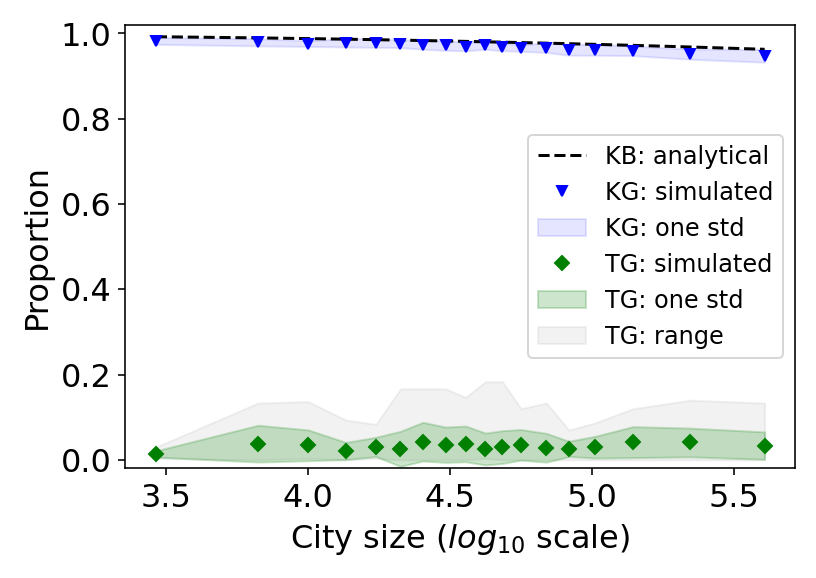}
	\end{tabular}
	\caption{Comparison of infected cities and isolated cities as a function of the theoretical $R_0$ value (top row)
		of simulated and theoretical infection (middle row) and outbreak probabilities (bottom row), on the left for strategy~\((P)\) and on the right for strategy~\((U)\), for Japan
		with a restriction on the distance of 50km. The $R_0$ value is adjusted such that the theoretical infection probability is 0.5
         for cities of size $10^5$, see Section \ref{sec_simulated_infection_and_outbreak_prob}.
	}\label{Fig:Compare PI v9}
\end{figure}

\end{document}